\documentclass[11pt,letterpaper]{article}
\usepackage{a4wide}

\usepackage{amsmath}
\usepackage{amsfonts}
\usepackage{amssymb}
\usepackage{longtable,algorithm,algorithmic}

\newtheorem{theorem}{Theorem}

\newtheorem{corollary}[theorem]{Corollary}

\newtheorem{lemma}[theorem]{Lemma}

\newtheorem{remark}[theorem]{Remark}

\newenvironment{proof}[1][Proof]{\noindent\textbf{#1.} }{\ \rule{0.5em}{0.5em}}

\newcommand\tab[1][1cm]{\hspace*{#1}}

\begin{document}

\date{}

\title{A new and improved algorithm for online bin packing\footnote{Gy. D\'{o}sa was supported by VKSZ\_12-1-2013-0088 ``Development of cloud based
smart IT solutions by IBM Hungary in cooperation with the
University of Pannonia'' and by National Research, Development and
Innovation Office -- NKFIH under the grant SNN 116095. L. Epstein
and A. Levin were partially supported by a grant from GIF - the
German-Israeli Foundation for Scientific Research and Development
(grant number I-1366-407.6/2016).}}

\author{J\'anos Balogh \thanks{Department of Applied Informatics, Gyula Juh\'asz Faculty of Education,
     University of Szeged, Hungary. \texttt{balogh@jgypk.u-szeged.hu}} \and J\'ozsef B\'ek\'esi \thanks{Department of Applied Informatics, Gyula Juh\'asz Faculty of Education,
     University of Szeged, Hungary. \texttt{bekesi@jgypk.u-szeged.hu}} \and Gy\"{o}rgy
D\'{o}sa\thanks{Department of Mathematics, University of Pannonia,
Veszprem, Hungary, \texttt{dosagy@almos.vein.hu}. }   \and Leah
Epstein\thanks{ Department of Mathematics, University of Haifa,
Haifa, Israel. \texttt{lea@math.haifa.ac.il}. } \and Asaf
Levin\thanks{Faculty of Industrial Engineering and Management, The
Technion, Haifa, Israel. \texttt{levinas@ie.technion.ac.il.}}}

 \maketitle

\begin{abstract}We revisit the classic online bin packing problem.
In this problem, items of positive sizes no larger than $1$ are
presented one by one to be packed into subsets called {\it bins}
of total sizes no larger than $1$, such that every item is
assigned to a bin before the next item is presented. We use online
partitioning of items into classes based on sizes, as in previous
work, but we also apply a new method where items of one class can
be packed into more than two types of bins, where a bin type is
defined according to the number of such items grouped together.
Additionally, we allow the smallest class of items to be packed in
multiple kinds of bins, and not only into their own bins. We
combine this with the approach of packing of sufficiently big
items according to their exact sizes. Finally, we simplify the
analysis of such algorithms, allowing the analysis to be based on
the most standard weight functions. This simplified analysis
allows us to study the algorithm which we defined based on all
these ideas. This leads us to the design and analysis of the first
algorithm of asymptotic competitive ratio strictly below 1.58,
specifically, we break this barrier by providing an algorithm AH
(Advanced Harmonic) whose asymptotic competitive ratio does not
exceed 1.57829.\end{abstract}

\section{Introduction}
Bin packing \cite{CoGaJo97,CsiWoe98} is the problem of
partitioning or packing a set of items of rational sizes in
$(0,1]$ into subsets of items, which are called bins, of total
sizes no larger than $1$. In the offline variant the list of items
is given as a set, and in the online environment items are
presented one by one and each item has to be packed into a bin
irrevocably before the next item is presented.

For an algorithm $A$, we denote its cost, that is, the number of
used bins in its packing on an input $I$ by $A(I)$. The cost of an
optimal solution $OPT$, for the same input, is denoted by
$OPT(I)$. The {\it asymptotic approximation ratio} allows to
compare the costs for inputs for which the optimal cost is
sufficiently large. The asymptotic approximation ratio of $A$ is
defined as follows. $ R_{A} =
  \lim\limits_{N \rightarrow \infty} \left(
  \sup\limits_{I: OPT(I) \geq N}
  \frac{A(I)}{OPT(I)} \right)  \,$.
In this paper we only consider the asymptotic approximation ratio,
which is the common measure for bin packing algorithms. Thus we
use the term approximation ratio throughout the paper, with the
meaning of asymptotic approximation ratio. Moreover, the term {\it
competitive ratio} often replaces the term ``approximation ratio''
in cases where online algorithms are considered. We will use this
term for the asymptotic measure. When we discuss the absolute
measure $\sup_{I} \frac{A(I)}{OPT(I)}$ (the absolute approximation
ratio or the absolute competitive ratio), we will mention this
explicitly. A standard method for proving an upper bound for the
asymptotic approximation ratio or the asymptotic competitive ratio
for an algorithm $A$ is to show the existence of a constant $C
\geq 0$ independent of the input, such that for any input $I$,
$A(I) \leq R \cdot OPT(I) +C$ (and then the value of the
asymptotic measure is at most $R$). Most work on upper bounds on
the asymptotic competitive ratio provide in fact an upper bound
using this last method, and we will follow this approach as well.

For the offline problem, algorithms with an approximation ratio of
$1+\varepsilon$ can be designed \cite{FVL81,KK82} for any
$\varepsilon>0$. If the first definition is used, a
$1$-approximation can be designed \cite{KK82}, where the cost of
the solution computed by the algorithm is $OPT(I)+o(OPT(I))$ (see
also recent work on improving the sub-linear function of $OPT(I)$
by Rothvoss \cite{Rothv} and by Hoberg and Rothvoss \cite{HR17}).

The classic bin packing problem, which we study here, was
presented in the early 1970's \cite{Ullman71,J73,J74,JDUGG74}. It
was introduced as an offline problem, but many of the algorithms
initially proposed for it were in fact online. Johnson
\cite{J73,J74} defined and analyzed the simple algorithm Next Fit
(NF), which tries to pack the next item into the last bin that was
used for packing, if such a bin exists (in which case such a bin
is called ``active'') and the item can be packed there, and
otherwise it opens a new bin for the item. The competitive ratio
of this algorithm is $2$ \cite{J73,J74}. Any Fit (AF) algorithms,
as opposed to the behavior of NF which only tests at most one
active bin for feasibility of packing a new item there, pack a new
item into a nonempty bin unless this is impossible (in which case
a new bin is opened). Such algorithms have competitive ratios of
at most $2$. Next, consider a sub-class of algorithms where one
may not select a bin with smallest total size of currently packed
items for packing a new item, unless this minimum is not unique or
this is the only bin that can accommodate the new item except for
an empty bin. The last class of algorithms is called Almost Any
Fit (AAF), and they have competitive ratios of $1.7$
\cite{JDUGG74,J74}. A well-known algorithm, which is in fact a
special case of AAF is Best Fit (BF), which always chooses the
fullest bin where the new item can be packed. First Fit (FF) is
another important special case of AF (but not of AAF) which
selects a minimum index bin for each new item (where it can be
packed). The competitive ratio of FF is 1.7 \cite{JDUGG74,DS12}.

The pre-sorted versions of these algorithms, called NFD, FFD, BFD,
and AFD, were studied as well. In these versions, items are still
presented one by one, but they are sorted in a non increasing
order (according to sizes). For example, the approximation ratio
of NFD is (approximately) 1.69103 \cite{BC81} and that of FFD is
$\frac{11}{9}\approx 1.22222$ \cite{J73}. For AFD in general, the
approximation ratio is at most $1.25$ \cite{J73,J74,JDUGG74}.
These pre-sorted variants are not online algorithms.

We design and analyze a new algorithm  AH (Advanced Harmonic) for
online bin packing, and show that its competitive ratio does not
exceed $1.57828956$. This is the first algorithm whose asymptotic
competitive ratio is below 1.58. We use a new type of analysis of
algorithms which allows us to split the analysis into cases, while
for every case we define only three different values (and even
just one value in a large number of cases), and based on those we
calculate weights for items. The analysis is split into cases in
recent previous work as well, but the analysis of each case is
much more difficult. Items are partitioned into classes according
to sizes. As in previous work, we sometimes do not pack the
maximum number of items of some class into a bin, and leave space
for items of another class (possibly arriving later). One new
feature of AH is that in previous papers, in the algorithms there
were at most two options for every class. For any given class, one
option was a bin with the maximum number of items of a this class
fitting into a bin. For some of the classes there was a second
option consisted of a very small number of items from this class
(with reserved spaces for items of another class, possibly
arriving later). We allow intermediate values as well with more
than two options for some classes and not only two kinds of bins
for a given class.

We use simple weight functions for the analysis, rather than the
much more complicated tool called {\it weight systems}
\cite{Seiden02J}. Weight functions are an auxiliary tool used for
the analysis of bin packing (and other) algorithms (this technique
is also called dual fitting). In this method, a weight is defined
for each item (usually, based on its size, and sometimes it is
also based on its role in the packing). If there are multiple
kinds of outputs, it is possible to define a weight function for
each one of them. The total weight of items is then used to
compare the numbers of bins in the output of the algorithm and in
an optimal solution. The list of weights of one item for different
output types, also called scenarios, can be seen as a vector
associated with the item. Thus, the weights can be seen as one
function from the items to vectors whose dimension is the number
of scenarios. Briefly, a weight system is a generalization where
the weight function also maps items (or item sizes) to vectors,
but in order to compute the weight of some item for a given
scenario, another function, called a consolidation function, is
used. This last function is a piecewise linear function (mapping
real vectors to reals). The slightly simplified approach is to use
convex combinations of weights according to subsets of scenarios.
It has not been proved that weight systems are a stronger tool
than just weights defined for the different scenarios. However,
for simple weights every scenario can be analyzed independently
from other scenarios. We exploit the simplicity of weight
functions to obtain a clean and full analysis, which is easier to
implement and verify (compared to the analysis resulting from
weight systems). The main advantage is that every case is analyzed
in a separate calculation using a standard knapsack solver without
considering any other cases at that time. This simplicity allows
us to analyze the new features that we introduce. Obviously, as
these are cases for one algorithm, they have a common set of
parameters, but once the algorithm has been fixed, there is no
connection between the various cases.

The significance of our approach is that we combine many existing
methods, including that of Babel et al. \cite{BCKK04} (recently
used by Heydrich and van Stee \cite{HvSbparxiv,HvS16} for classic
bin packing), adding several new features, and applying a simple
analysis, which can be verified easily. We define the action of
our algorithm AH, we prove a number of invariants and properties
of AH in detail, and then we provide the specific parameters and
compact representations of the lists of weights. For every
possible output type and scenario, there is a small number of
values used for the calculation of weights for it. We also provide
explicit lists of weights calculated based on the values and the
parameters.

To explain the new features of our work, we start with a
discussion of the design of harmonic type algorithms. Already in
much of the previous work on online algorithms for bin packing,
items were partitioned into classes by size. The simplest such
classification is based on harmonic numbers, leading to the
Harmonic algorithm of Lee and Lee \cite{LeeLee85}. In the harmonic
algorithm of index $k$ (for an integer parameter $k \geq 2$),
subset $j$ is the intersection of the input and $(\frac
1{j+1},\frac1{j}]$ (where $1 \leq j \leq k-1$), and subset $k$ of
tiny items is the intersection of the input and  $(0,\frac 1k]$.

In these algorithms each subset is packed independently from other
subsets using NF (so for $j\leq k-1$, any bin for subset $j$,
except for possibly the last such bin, has $j$ items, but for
subset $k$, every bin except for the last bin for this subset has
a total size of items above $\frac{k-1}k$), and for $k$ growing to
infinity, the resulting competitive ratio is approximately
$1.69103$ \cite{LeeLee85}. The drawback of those algorithms is
that bins of subsets with small values of $j$ can be packed with
small sizes of items (for example, a bin of subset $2$ may have
total size just above $\frac 23$ and a bin of subset $1$ may have
just one item of size just above $\frac 12$).

The first idea which comes to mind is to try to combine items of
those two subset into common bins. However, if items of class $2$
arrive first, one cannot just pack them one per bin, as this
immediately leads to a competitive ratio of $2$ (if no items of
subset $1$ arrive afterwards). Lee and Lee \cite{LeeLee85}
proposed the following method to overcome this. A fixed fraction
of items of subset $2$ (up to rounding errors) is packed one per
bin and the remaining items are packed in pairs. Thus, there are
two kinds of bins for subset $2$. The items we refer to here can
only be sufficiently small items, so there is a threshold $\Delta
\in (\frac 12,\frac 23)$ such that items of sizes in $(\Delta,1]$
and $(1-\Delta,\frac 12]$ are packed as before, while the
algorithm tries to combine an item of size in $(\frac 13,
1-\Delta]$ with an item of size in $(\frac 12,\Delta]$. Even if
those two items (one item of each one of the two intervals) are
relatively small, still their total size is above $\frac 56
\approx 0.83333$. This last algorithm was called Refined-Harmonic,
and its competitive ratio is smaller than 1.636. Ramanan et al.
\cite{RaBrLL89} designed two algorithms called Modified Harmonic
and Modified Harmonic-2. The first one has a competitive ratio
below $1.61562$, and it allows to combine items of many subsets
with items of sizes above $\frac 12$ (and at most $\Delta$). The
second algorithm does not use only a single value of $\Delta$, but
splits the interval $(\frac 12,1]$ further, allowing additional
kinds of combinations. Its competitive ratio is approximately
$1.612$. For most subsets of items (where $k$ is chosen to be in
$[20,40]$ in all these algorithms), the last two algorithms pack
some proportion of the items in groups of smaller sizes, to allow
it to be combined with an item of size above $\frac 12$.
Intuitively, for an illustrative example, assume that
$\Delta=0.6$, and consider the items of sizes in $(\frac
1{11},\frac 1{10}]$. The items that are not packed into groups of
ten items should be packed into groups of four items (the
parameters of the algorithms are different from those of this
example). For some of the subsets the proportion is zero, and they
are still packed using NF.  The drawback of such algorithms (as it
is exhibited by Ramanan et al. \cite{RaBrLL89}) is that no matter
how many thresholds there are, there can be pairs of items that
can be combined into bins of optimal solutions while the algorithm
does not allow it as it has fixed thresholds. Specifically, such
algorithms allow to combine items of different intervals only in
the case that the largest items of the two intervals fit together
into a bin. This is the case with the next two harmonic type
algorithms as well.

The next two papers, that of Richey \cite{Ric} and that of Seiden
\cite{Seiden02J} deal with a more complicated algorithm where many
more subsets can be combined. The general structure is proposed in
\cite{Ric}, and a full and corrected algorithm with its analysis
is provided in \cite{Seiden02J}. For illustration, the items
packed into smaller groups are called red and those packed into
bins with maximum numbers of items of the subset are called blue.
The goal is to combine as many bins with blue items with bins
having red items as possible. Bins with red items always have
small numbers of items, to allow them to be combined with
relatively large items of sizes above $\frac 12$. The analysis is
far from being simple, though it leads to a competitive ratio of
at most 1.58889 (Heydrich and van Stee \cite{HvSbparxiv,HvS16}
mention that this last value can be decreased very slightly).

The carefully designed subset structure eliminates many worst-case
examples, but the drawback mentioned above still remains.
Recently, Heydrich and van Stee \cite{HvSbparxiv,HvS16} proposed
to use a method introduced by Babel et. al \cite{BCKK04}, where
some items are packed based on their exact size rather than by
their subset. The approach of \cite{HvSbparxiv,HvS16} which we
adopt is to apply the methods of Babel et. al \cite{BCKK04} on the
largest items, of sizes in $(\frac 13,1]$. This approach means to
combine items of sizes above $\frac 12$ with items of sizes in
$(\frac 13, \frac 12]$ based on their exact sizes. Moreover, the
approach involves combining pairs of items of subsets of sizes
contained in $(\frac 13, \frac 12]$ while keeping the smallest
items of such a subset
  to be matched with items of sizes
above $\frac 12$ (and larger items of such a subset are used to be
packed into pairs), as much as possible. Prior to the work of
\cite{HvSbparxiv,HvS16}, all previous algorithms for classic bin
packing that partition items into classes always assumed that an
item of a certain subset has the maximum size when its possible
packing was examined. This method simplifies the algorithm and its
analysis, but it is not always a good strategy as this excludes
the option of combining items that can fit together into a bin in
many cases. This approach is very different from that of AF
algorithms and even from NF. Moreover, an approach similar to that
of Babel et. al \cite{BCKK04} was used in an online algorithm
designed in \cite{ftsoda}. Heydrich and van Stee
\cite{HvSbparxiv,HvS16} claim a competitive ratio of 1.5815 (see a
discussion regarding this in Appendix \ref{rvssh}).

In algorithm AH, we do not just have red and blue items, but we
potentially allow several kinds of bins. For example, for the
subset of items of sizes in $(\frac 1{15},\frac 1{14}]$ we group
items into subsets of $14$ items or three items or just one item.
We also use bins of the smallest items (our value of $k$ is $43$)
where the total size of items is at most $\frac{17}{60}$, to allow
them to be combined (among others) with items of sizes in $(\frac
12,\frac{43}{60}]$. These two features are possible due to the
simple nature of our analysis, and they are crucial for getting
the improved bound. Note that all items of sizes in $(0,\frac
1{43}]$ are treated together (by the algorithm and its analysis).

In order to use just a small number of values (one or three) for
each scenario, we use the concept of {\it containers}. A container
is a set of items of one class (in the partition of potential
inputs into items of similar sizes, called classes), and it can be
complete if its planned number of items has arrived already or
incomplete otherwise (but it is treated in the same way in both
cases). Containers are of two types, where a container is either
positive or negative, and a bin may contain at most one of each of
them. The goal is to have as many bins as possible with both a
positive and a negative container. Roughly speaking, positive
containers have total sizes above $\frac 12$ and negative
containers have total sizes of at most $\frac 12$. This last
statement is imprecise as in most cases we consider volumes and
not exact sizes, where volumes are based on the maximum sizes for
the corresponding classes. There is one exception which is
containers with one item of size above $\frac 13$, where the exact
size is taken into account (both by the algorithm and the
analysis), and it is defined to be the volume. A positive
container and a negative one fit together if their total volumes
does not exceed $1$, and does not depend only on the classes. Our
positive containers and negative containers have some relation to
concepts used in \cite{Seiden02J}.

In our weight based analysis, we assign weights to containers,
where the number of different weights is small. Specifically, let
the minimum volume of any positive container not packed with a
negative container  be denoted by $a$. We have two cases. In the
simple case where all positive containers packed without negative
containers have volumes of at least $\frac 23$ (i.e., $a\geq \frac
23$), we define weights as follows. Assign weights of $1$ to
positive containers packed without negative containers and
negative containers packed without positive containers. Since we
later base our weights of items on sizes, we assign these weights
of $1$ to all positive containers of volume at least $a$ and all
negative containers of volumes above $1-a$. We have a variable $w$
($0 \leq w \leq 1$) such that other positive containers have
weights of $w$ and other negative containers have weights of
$1-w$. Those weights are called the required weights of containers
(the actual weights can be larger but not smaller). Given the
approximate proportions of items of each class packed in every
type of container, we compute a weighted average (based on the
containers of every item) to define weights of items using the
required weights of containers. The case where $a < \frac 23$ is
more interesting as a negative container with one item of size in
$(\frac 13,\frac 12]$ and a positive container with one item of
size above $\frac 12$ can be packed into one bin if the total size
of the two items does not exceed $1$ (i.e., the volumes of their
containers are the exact sizes of these two items). Thus, the
exact value $a$ is crucial and not only its class, and
additionally the class and even the exact value of $1-a$ play an
important role. Here, for other classes we do the same as in the
previous case, but for one class we perform a more careful
analysis. This is the class containing the value $1-a$. For this
class we define weights of items directly. We let the weight of an
item of this class of size at most $1-a$ be a variable $u$, and
otherwise it is a variable $v$, where $v \geq u$ (this class is
contained in $(\frac 13,\frac 12]$). For the analysis, we found
suitable values for the variables for all scenarios (this was done
separately for each scenario), that is, for all possible values of
$a$ (the number of scenarios is still finite, as they are based on
the dividing points of the algorithm, though not only on the
classes). For every scenario where $a<\frac 23$, there are
additional constraints on $u$, $v$, and $w$. As we do not use
weights of containers in this case (for the class containing
$1-a$), while the packing of pairs of items of classes contained
in $(\frac 13, \frac 12]$ is performed carefully for all such
classes. After selection suitable values for those variables, all
other item weights are also computed using the parameters of the
algorithm.

It should be noted that there are also improved algorithms based
on First Fit. Yao \cite{Yao80A} designed an algorithm where
certain size based subsets are packed separately, resulting in a
competitive ratio of $\frac 53$. Many years later, an algorithm of
absolute competitive ratio $\frac 53$ was designed \cite{ftsoda},
which is the best possible with respect to this last measure
\cite{Zhang}. The absolute competitive ratios and approximation
ratios of other bin packing algorithms were studied as well
\cite{SL94,DS12,DS14}. The (asymptotic) competitive ratios should
be compared to lower bounds on the competitive ratio. The current
best such lower bound is 1.5403 \cite{BBG} (see also
\cite{Vliet92}).

\section{Notation and definitions}
Similarly to previous algorithms' definitions, AH has a sequence
of boundary points that are used in its precise definition: $1=
t_0>t_1= \frac 12> t_2> \cdots
> t_b= \frac 13
> \cdots > t_M>t_{M+1}=0$. That is $1/2$ and $1/3$ are always
boundary points, and there is no boundary point in $(1/2,1)$.

For every $j$, all items of sizes in the interval $(t_j,t_{j-1}]$
are called items of class $j$.  We say that a class of items (and
every item of this class) is {\it huge} if $j=1$, it is {\it
large} if $1<j\leq b$ (these are all items of sizes above $1/3$
and at most $1/2$), small if $b<j\leq M$, and tiny if this is the
class of items of size at most $t_M$ (i.e., the last class which
is the class of tiny items is class $M+1$, and in general the
index of a class corresponds to the index $j$ such that $t_j$ is
the infimum size of any item of the class).

Our algorithm will pack items into containers and pack containers
into bins. As the algorithm is online, a container will be packed
into a bin immediately when it is created, even though it may
receive additional items later. In the last case, when we say that
an item is packed into a container, this means that the bin
containing the container receives that item. Any container will
contain items of a single class, and at most two different
containers can be combined (packed) into a bin. We provide
additional details on combining two containers into a bin later.
Every container of items that are not tiny has a cardinality
associated with it, and this is the (maximum) number of items that
it is supposed to receive.

Let $\gamma_j = \lfloor \frac{1}{t_{j-1}} \rfloor$ for $j \leq M$.
For class $j$ that is either large or small (but not huge or tiny,
i.e., for values of $j$ such that $2\leq j \leq M$ holds), and for
every $i$ (where $1 \leq i \leq \gamma_j$) there is a nonnegative
parameter $\alpha_{ij}$, where $0 \leq \alpha_{ij} \leq 1$. The
values $\alpha_{ij}$ will denote the proportions of container
numbers of cardinalities $i$ of class $j$ items among the number
of container of class $j$ (the term proportion corresponds to the
property of the sum of proportions satisfies $\sum_{i} \alpha_{ij}
=1$ for all $j$). Such containers that will eventually receive $i$
items of class $j$ (unless the input terminate before this becomes
possible) will be called {\it type $i$ containers of class $j$}.
That is, intuitively if we let $x$ denote the number of containers
for items of class $j$, we will have approximately $\alpha_{ij}
\cdot x$ type $i$ containers each of which having exactly $i$
items of class $j$.  For every $j$ such that $2\leq j\leq M$ and
every $i$, we let $A_{i,j}=i\cdot t_{j-1}$.  While the values
$\alpha_{ij}$ are defined so far only for large and small classes,
we see one huge item as a type 1 container. Note that the values
of $\alpha_{ij}$ are not proportions of item numbers but of
container numbers for class $j$, and the resulting proportions of
items can be computed from them (we will prove such bounds
accurately later).

For classes of {\it large} items the notion of the cardinality of
a container is slightly more delicate, and we will have exactly
four possible types of containers. The first type is a {\it
regular type 2} container (already) containing exactly two items
of this class. The second type is a {\it declared type 2
container}, where this type consists of containers for which the
algorithm already decided to pack two items of this class in the
container (so the planned cardinality of the container is $2$) but
so far only one such item was packed into the container (one of
the few next arriving items of this class, if they exist, will be
packed there, in which case the type will be changed into a
regular type 2 container). The third is a {\it regular type 1}
container, where such a container has one item of the class and
cannot ever have (in future steps) an additional item of this
class (such a container will be always already combined with a
container of another class that is packed into the same bin). The
fourth and last type of a container of large items is a {\it
temporary type 1 container}. A container of this last type
currently has one item of the class but sometimes it will get an
additional item of this class in future steps (and in this case
its type will be changed at that time to regular type 2, its type
can change to declared type 2 or regular type 1 as well, but in
those cases it does not happen as a result of receiving a new
item). Given a class of large items, the number of declared type 2
containers will be at most four throughout the execution of the
algorithm (as we will prove below) while the numbers of containers
of type 1 (of both kinds) and containers of regular type 2 can
grow unbounded as the length of the input grows, though we will
show certain properties on the relations between their numbers
maintained by the algorithm. The set of the union of containers of
regular type 2 and declared type 2 are called type 2 containers,
and the set of the union of containers of regular type 1 and
temporary type 1 are called type 1 containers. The parameters
$\alpha_{1j}$ and $\alpha_{2j}$ of a large class $j$ determine the
approximate proportions of type 1 containers and type 2
containers, respectively.

For class $M+1$ (of the tiny items), instead of the definitions
above, there is a sequence of $p$ possible upper bounds on the
total sizes of items packed into containers of this class: $1 \geq
A_{p,M+1}>A_{p-1,M+1}> \cdots > A_{1,M+1} \geq t_M$, and we let
the positive parameters $\alpha_{i,M+1}>0$ for $i=1,\ldots,p$
denote the proportion of numbers of containers of class $M+1$ with
items of total size in the interval $(A_{i,M+1}-t_M,A_{i,M+1}]$
(this is the planned total size of items for such a container).
Such containers will be called {\it type $i$ containers of class
$M+1$}.

The {\em volume} of a container of type $i$ of class $j$ is
defined as follows: If $i=1$ and $1 \leq j\leq b$ (that is, for
items of sizes above $1/3$), the volume of the container is the
size of its (unique) item, and otherwise ($i=2$ and $2 \leq j \leq
b$ or $i\geq 1$ and $j>b$) it is $A_{i,j}$. That is, the volume is
usually simply the largest total size that the container can
occupy, but for a container that contains a single large or huge
item, the volume is the {\em exact} size of the item (there is one
exception where the bin already contains one large item and it is
planned to contain another item of the same class). In most cases
we would like the volume of a container to be known when it is
created, which is possible for containers such that their planned
contents are known (in the sense that for example type $i$
containers of a non tiny class $j$ are planned to contain $i$
items finally). However, for large items such containers with a
single item may be temporary type 1 containers, in which case
there is still no planning of contents for them. In this last
case, the volume of the container is the size of its unique item.
However, the volume of such a container may change in the case the
algorithm will decide to pack another item of the same class (no
matter if it packs that other item immediately at the time of
decision or whether we decide to pack such an item later) into
this container and transform it into a type 2 container. The
volume of a declared type 2 container of class $j$ is
$A_{2,j}=2\cdot t_{j-1}$ (the volume is based on its complete
contents, no matter whether they are present already or not, as it
is the case for classes of small or tiny items).

%volume for huge items???

We say that a container is {\em negative} if its volume is at most
$1/2$ and otherwise it is {\em positive}. Obviously, two positive
containers cannot be packed into one bin. We will also not pack
two or more negative containers into a bin together. Thus, a bin
containing two containers will contain one positive container and
one negative container, and no bin will contain more than two
containers.

\section{Algorithm AH}
The algorithm AH which we define next will pack items into
containers and pack containers into bins according to rules we
will define. Recall that the packing of containers into bins will
be such that every bin will have at most one positive container
and at most one negative container. Obviously, a bin is nonempty
if it has at least one container and at most two containers. We
say that a nonempty bin is negative if it has a negative container
and does not have a positive container, it is positive if it has a
positive container and does not have a negative container, and it
is neutral if it has both a negative container and a positive
container.

It is unknown whether a temporary type 1 container will eventually
be positive or negative. Therefore, such a container will not be
combined in a bin with another container as long as its type is
not changed. Moreover, it is seen as a negative container until it
changes its type (so its bin is negative as long as the container
is of temporary type 1). Specifically, it remains a negative
container if a positive container joins it (and its bin becomes
neutral), and in this case it becomes a regular type 1 container
(and remains negative), and it becomes a positive container if its
type changes to type 2. It can also happen that a temporary type 1
container will remain such till the termination of the input and
the action of AH (and its bin remains negative). It is important
to note that the difference between regular type 1 containers of a
large class and temporary type 1 containers of the same class is
that each of the former containers is already packed into a bin
with a positive container (of some class), while the latter are
not packed with other containers (in fact, the corresponding items
are placed into their own bins, one item per bin).

For every class $j$, we denote by $n_j$ the number of containers
of class $j$. Let $n_{ij}$ denote the number of containers of type
$i$ of class $j$. We also let $N_j$ denote the number of items of
class $j$ at that moment.  We often consider the values $n_j$ and
$n_{ij}$ just prior to the packing of a new item, when $N_j$ was
already increased but the new item not packed yet so the values
$n_j$ and $n_{ij}$ are not updated yet.

We say that two containers {\em fit together} if their total
volume is at most $1$. In what follows, when we refer to packing
an item $e$ - or more precisely, packing a container containing
$e$ (which was just created and therefore contains only $e$) into
existing bins using Best Fit -  we refer to packing $e$ (or the
container containing $e$) into the bin with a container of largest
volume where the existing container and $e$ (or the container
containing $e$) fit together. For the original version of Best
Fit, actual sizes are taken into account, but here we base this
rule on volumes (as for a container with a single large or huge
item the volume is equal to the size of the item, if we select one
such container among a set of this last kind of containers, our
action is equivalent to the standard application of Best Fit).

Next, we define the packing rules of the algorithm when a new item
of class $j$ arrives. The algorithm is defined for each step,
based on the class of the new item.

{\bf A huge item.} Recall that a huge item is immediately packed
into a positive container containing only this item. Use Best Fit
(applied on volumes, as explained above) to pack the created
container into an existing bin, out of existing negative bins,
such that the two containers (the new one with the huge item and
the negative one of the negative bin) fit together. The only case
where the new huge item joins a bin with a large item of some
class $j'$ is the case where the container of class $j'$ is a
temporary type 1 container, and in this case the type of this
container of class $j'$ is changed into regular type 1. If no bin
can accommodate the container of the new item according to those
packing rules, that is, for every negative bin, the total volume
together with the new item is too big (or there is no negative bin
at all), then use a new bin for the positive container of the new
item (this new bin becomes a positive bin).

{\bf An item of a class of small or tiny items.} For these classes
we define the concept of an open container. Informally, an open
container (of class $j$) can receive at least one additional item
of class $j$. As a new container is introduced in order to pack an
item, any container (of any type and class) already has at least
one item of the corresponding class. If $b<j\leq M$, an open type
$i$ container of class $j$ is one where the total number of the
items in the container is strictly smaller than $i$. Once such a
container receives $i$ items, it is closed. For $j=M+1$, a type
$i$ container of this class will be open starting the time it is
created and while the total size of items in it is positive and at
most $A_{i,M+1}-t_M$. Once it reaches a total size above
$A_{i,M+1}-t_M$, it will be closed. For all cases of packing a
small or tiny item, a new container of some class will be used
only if there is no open container of the same class, and thus, in
particular, there will be at most one open container for each $j$
(and the corresponding value of $i$ will always be one such that
$\alpha_{ij}>0$).

When a new item of class $j$ (such that $j>b$) arrives, if there
is an open container of some type $i$ of class $j$, then pack the
item there (there can be at most one such container, so there are
no ties in this case). Otherwise, open a new container for it (the
details of the type are given below). After packing the new item
into the container (and packing its container into a bin if it is
a new container), close the container if necessary, based on its
type and the rules above.

In the case that a new container is used for the item, we define
the process of packing the item in more detail. Prior to packing
the item, we define the type of the new open container. As the
item is not packed yet, $n_j$ is the number of containers of class
$j$ excluding the container opened for the new item. Find the
minimum value of $i$ such that $\alpha_{ij}>0$ and so far there
are at most $\lfloor \alpha_{ij} \cdot n_j \rfloor$ type $i$
containers of class $j$ (i.e., $n_{ij}\leq \lfloor \alpha_{ij}
\cdot n_j \rfloor$, where the values $n_{ij}$ do not include the
new container which will be opened). Such an index $i$ exists as
otherwise there are more than $n_j$ containers of class $j$. More
precisely, since $\sum_i \alpha_{i,j}=1$, there is always a value
of $i$ satisfying that $\alpha_{ij}>0$ such that so far we opened
at most $\lfloor \alpha_{i,j} \cdot n_j\rfloor$ type $i$
containers of class $j$. Open a new type $i$ container of class
$j$ containing the new item (increasing both $n_j$ and $n_{ij}$).
Observe that this opening of a new container defines its volume as
well as whether it is a positive container or a negative
container.

Next, we decide where to pack this new container.  First consider
the case where this container is a negative container. Then, if
there is a positive bin, such that the new container fits into the
bin according to its volume, then use that bin to pack the new
container. This last case includes the possibility that the
positive container is a type 2 container of a large class (regular
or declared). If there are multiple options for choosing a bin,
one of them is chosen arbitrarily.

Otherwise (there is no positive bin where the new negative
container can be added), the algorithm checks the option of using
a bin with a temporary type 1 container of some class of large
items. Assume that there is a negative bin $B$ such that the
following two conditions are satisfied. The first condition is
that the bin $B$ has a temporary type 1 container of class $j'$
such that a positive container of class $j'$ (with two items) will
fit together with the new (negative) container. The second
condition is that there are at most $\lfloor \alpha_{2j'}\cdot
n_{j'}\rfloor - 1$ type 2 containers of class $j'$ (before the
packing of the new item is performed). Then, pack the new negative
container into $B$, and define the container of class $j'$ packed
into $B$ as a declared type $2$ container. This last container of
class $j'$ will get one of the next items of class $j'$ that will
arrive, which will happen before any new container is opened for
any new class $j'$ item, see below. If there are multiple options
for choosing $B$, one of the classes of large items is chosen
arbitrarily (among those that can be used), and a temporary type 1
container of this class with maximum volume is selected, i.e., we
use Best Fit in this case. This last packing step is possible as a
temporary type 1 container is never packed with another container
into a bin (if another container joins it, its type is changed).

Otherwise (if there is no suitable positive bin and no class of
large items has a suitable temporary type 1 container that can be
used under the required conditions), pack the new negative
container into a new bin.

Finally, consider the case where the new container is a positive
container. Then, if there is a negative bin whose container is not
a temporary type 1 container, such that the new container fits
together with it, then use such a bin to pack the new container.
Otherwise, if there is a temporary type $1$ container with one
large item of a class $j'$ where the new container fits, then pack
the new positive container into this bin and define the container
of class $j'$ in this bin as a regular type $1$ container. The
class $j'$ can be chosen arbitrarily if there are multiple
options, and among the temporary type 1 containers of class $j'$,
one of maximum volume (out of those that can be used) is selected,
i.e., once again we use Best Fit. Otherwise, pack the new positive
container into a new bin.

{\bf A large item of a class $\boldsymbol{j}$.} If there is a
declared type 2 container of class $j$, pack the item there (as a
second item) and change it into a regular type 2 container
(breaking ties arbitrarily). This packing rule is checked first,
and we apply it whenever possible. We continue to the other cases
in the situation where there is no such declared type 2 container.

If the number of type 2 containers equals $\lfloor
\alpha_{2j}\cdot n_{j}\rfloor$ (that is, we should not increase
the number of type 2 containers at this stage), then pack the new
item into a new negative container.  To pack the container into a
bin, do as follows. If there is a positive bin where the new
negative container fits, then use Best Fit to pack it as a regular
type 1 container of class $j$ (its volume is defined accordingly
as the size of the new item) together with a positive container
(this positive container is not of large items, as three large
items cannot be packed into a bin together). Otherwise the new
container is packed into a new bin, in which case it is defined to
be a temporary type 1 container.

Otherwise (that is,  the number of type 2 containers is strictly
smaller than $\lfloor \alpha_{2j}\cdot n_{j}\rfloor$), we will
increase the number of regular type 2 containers or the number of
declared type 2 containers of this class in the current iteration
as follows. If there is a negative bin $B$ where a type 2
container of class $j$ fits, then pack the item into a new
declared type 2 container of class $j$ and pack this container
into this bin $B$. Otherwise, if there is a temporary type 1
container of class $j$, then we pack the new item using Best Fit
(considering only temporary type 1 containers of class $j$, and
selecting such a container of largest volume) and change the type
of this container into a regular type 2 container. Otherwise (all
containers of class $j$ are either regular type 1 or regular type
2, we should increase the number of type 2 containers, and a new
container with two items of this class cannot be packed into an
existing bin), we open a new declared type 2 container for the new
item and open a new bin for this declared type 2 container (and
pack it there).

The value $\alpha_{2j}$ is strictly positive for every large class
$j$ (as packing every item of a certain large class in its own bin
will lead to a competitive ratio of $2$ for inputs consisting only
of such items). However, there may be values of $j$ ($2 \leq j
\leq b$) for which $\alpha_{1j}=0$. In those cases, the algorithm
above is still applied. Moreover, in those cases it could happen
that there will be a constant number of type 1 containers for
class $j$, as we prove below (the general proof is valid in the
case $\alpha_{1j}=0$ too).

\begin{remark}\label{rmk-temp}
Note that the change of types of containers (of large classes) is
a unique and delicate feature of AH.  While the change of a
declared type 2 container into a regular type 2 container when a
new item is packed into this container, can be described also by
previous approaches, our rules for changing the type of temporary
type 1 containers are new and particularly important.  We
summarize those rules as follows.
\begin{enumerate} \item If a (new) positive container of another
class is packed into a bin containing a temporary type 1
container, then we change the type of the temporary type 1
container into a regular type 1 container. \item If a (new)
negative container of another class is packed into a bin
containing a temporary type 1 container, then we change the type
of the temporary type 1 container into a declared type 2
container.  \item If a (new) item of the same class of the
temporary type 1 container joins the same bin as the temporary
type 1 container, then we change the type of the temporary type 1
container into a regular type 2 container.
\end{enumerate}
Furthermore, in all these cases, we pack the new container or
large item using Best Fit.  That is, we pick the largest temporary
type 1 container where the new container or new item fits.
\end{remark}
Note that there are, however, previous papers where it was not
always decided in advance whether for a class of items for which a
bin can contain at most two items of this class, the bin will
contain one item or two items. In \cite{BCKK04}, in the studied
problem a bin can never contain more than two items, so there are
just two classes of items, larger items of sizes above $\frac 12$,
and the smaller items, which are all other items. In the algorithm
of \cite{BCKK04}, whenever a new smaller item is to be packed with
another smaller item, Best Fit is applied. However, in
\cite{BCKK04} there is no concept of packing other kinds of items
into such bins (as according to their model, those bins already
have the maximum number of items). In \cite{HvSbparxiv,HvS16} the
difficulty of deciding whether large items should be packed in
pairs or alone (in order to be packed with items of other classes)
is solved in a slightly different way; there is a provisional
decision (so for some large items it is decided that they will be
packed in pairs and for others that they will be packed alone).
The final decision is set after a sufficient number of items of
the class have arrived. If in the meantime some items were
combined with items of other classes, for those items the
decisions are final. After sufficiently many items arrive without
other items being combined with them, a decision is made for all
this large set at once.

\begin{remark}\label{rmk-type2}
Consider a large class $j$. Consider an iteration $\ell$ of the
algorithm, i.e., the arrival of the $\ell$th input item, which is
not necessarily of class $j$. Assume that as a result of packing
the $\ell$th item a given container of class $j$ becomes a type
$2$ container. That is, this last class $j$ type $2$ container
either did not exist before the current step, or it was a type 1
container prior to this iteration. Assume also that this container
is not packed with a negative container in a bin. Then, the
$\ell$th item is of class $j$.
\end{remark}

\section{Analysis}

\subsection{Properties of the packing of positive and negative containers}
In the analysis, we see a pair of a negative container and a
positive container, packed together in a bin, as matched to each
other, and each one of them is seen as matched (while every
container packed into a bin without another container is
unmatched).  Our next goal is to prove the properties of this
matching. Let $a'=1-s_{\min}/2$ where $s_{\min}$ is the smallest
item size in the examined input, and let $a$ be the smallest
volume of a positive container that is unmatched, if it exists. If
no unmatched positive container exists, let $a=a'$. If $a>a'$,
decrease the value of $a$ to be $a'$. A simple property of the
algorithm is that it tries to match a positive container and a
negative container whenever possible.

\begin{lemma}\label{bignono}
Consider some time during the execution of the algorithm, just
after an item has been packed. If there exists at least one
positive bin and at least one negative bin, let $\theta_{neg}$
denote the smallest volume of any container of a negative bin and
let $\theta_{pos}$ denote the smallest volume of any container of
a positive bin. Then, $\theta_{neg}+\theta_{pos}>1$.
\end{lemma}
\begin{proof}
For a negative container of a large class, its volume is above
$\frac 13$, and for a positive container of a large class, its
volume is above $\frac 23$. Thus, if both containers are of large
classes, we are done. It is left to consider several cases, based
on whether one of the containers is of a large class (this can be
the negative container or the positive container, or none of
them), and on which of the two containers was created first (in
the case of a positive container of a large class, we also need to
consider the time when this container changes its type to type 2).

If none of the two containers is of a large class, by the rules of
packing a new container of a class that is tiny, small, or huge, a
new positive or negative container is packed into an empty bin
only if the total volume of the new container and the container in
any relevant bin (a relevant bin is a positive bin if the new
container is negative, and it is a negative bin if the new
container is negative) is above $1$.

If the negative container is of a large class, then since its bin
is negative until the current time, it is of temporary type 1
until the current time. If the positive container was created
after this negative container, then as it was not combined with
the negative container of the large item, their total volume is
above $1$. Otherwise, when the temporary type 1 container was
created, it was not possible to combine it with a positive
container, and therefore the total volume is above $1$ in this
case as well.

If the positive container is of a large class, it is of type 2,
and we consider three cases. If the container of the large class
becomes of type 2 before the time when the negative container is
created, then it is already a positive container when the negative
container is created, and therefore by the packing rules their
total volume is above $1$. Assume now that the negative container
was created before the time when the type 2 container was defined
as type 2 (either by changing its type from a temporary type 1
container to type 2, or by the creation of a new declared type 2
container). A temporary type 1 container becomes of type 2 without
being combined with a negative container only in the case where no
negative container that can be combined with a type 2 container of
this large class exists. A new declared type 2 container is packed
into a new bin only if there is no negative container in a
negative bin such that they could be packed together. Thus, in all
three cases the total volume of the two containers is above $1$.
\end{proof}

%
%\begin{lemma}\label{bignono}
%Consider a positive container $M_p$ and a negative container $M_n$
%such that their total volume is at most $1$. At any stage, either
%$M_p$ is packed with a negative container, or $M_n$ is packed with
%a positive container (or both).
%\end{lemma}
%\begin{proof}
%For containers whose volume is defined at their creation time, the
%claim holds as there is always an attempt to combine a new
%container with a previously existing container (a negative one
%with a positive one and vice versa) before a new bin is used. The
%only kinds of containers changing their volumes are of large
%items, and this happens when a temporary type 1 container becomes
%a type 2 container. We note that otherwise the value of the volume
%is determined at the creation of the container.  If this happens
%due to the arrival of a large item of class $j$, the only case
%where this can happen is that there does not exist any negative
%container that is not combined with a positive container and that
%can be combined with a type 2 container of class $j$. The other
%situation where a temporary type 1 container becomes a (declared)
%type 2 container of class $j$ is when a small or tiny item arrives
%and its (new) negative container is combined with the container of
%class $j$ which was a temporary type 1 container till now, and it
%becomes a declared type 2 container, but in this case the
%resulting positive container is combined with a negative container
%right away.
%\end{proof}

\begin{lemma}\label{sizet}
Every type 1 container of a large class $j$, where its unique item
has size no larger than $1-a$ is combined with a positive
container in the output, and in particular, all such containers
are of regular type 1 in the output.
\end{lemma}
\begin{proof}
If $a=a'$, the claim is trivial. Otherwise, by the definition of
$a$, there is a positive container that is not packed with a
negative container, whose volume is $a$. Thus, at termination, by
Lemma \ref{bignono} there are no negative containers of volumes at
most $1-a$ that are not combined with positive containers. For
classes of large items, negative containers that are not combined
with another container into the same bin are only temporary type 1
containers (recall that a declared type 2 container is a positive
one). Thus, there may be regular type 1 containers for large
classes of all possible volumes (that are combined with positive
containers), and there can be temporary type 1 containers of
volumes strictly above $1-a$ (but not smaller) that are not
combined with positive containers.
\end{proof}

Then, by Lemmas \ref{bignono} and \ref{sizet}, every negative
container of volume at most $1-a$ is matched, and every positive
container of volume strictly smaller than $a$ is matched (by the
definition of $a$).

\subsection{The set of scenarios}
We define a {\em finite} set of scenarios according to the value
of $a$ so that in particular the index of the scenario will reveal
the class that contains the value $1-a$ and so that the index of
the scenario will determine for each container containing at least
two items (of the class of the container) the relation between the
volume of this container and the values of $a$ and $1-a$.

To do that we define a set of values $V$ as follows.  $V=\{
A_{i,j}, 1-A_{i,j} : j= 2,3,\ldots ,M+1, \alpha_{ij}>0\} \cup \{
t_1,t_2,\ldots , t_M,t_{M+1}\}$ and $V'=\{ x\in V: x\leq 1/2\}$
(in particular, $\frac 12 \in V'$). Note that the set $V'$
contains (among other) all boundary points $t_j$ (for all $j\geq
1$), even for values of $j$ for which $\alpha_{1j}=0$. The name of
a scenario is an interval $(x,y]$ between consecutive values in
$V'$. For a given value of $a$, we find the values $x,y$ such that
$1-a \in (x,y]$, and $y>x$ are two values in $V'$ such that $(x,y)
\cap V' = \emptyset$. The analysis is based on the value of
$(x,y]$, where the motivation is that given this value, it is
known exactly which positive bins we may have (excluding
containers with one item which is large or huge, for which the
volumes are not limited to those in $V$). We define the index $k$
of the {\it threshold class} to be the value of $k$ such that
$(x,y] \subseteq (t_k,t_{k-1}]$. As $V$ and $V'$ contain values
that are not in $\{ t_1,t_2,\ldots , t_M\}$, it is possible that
$x>t_k$ or $y<t_{k-1}$ or both.

\medskip

The analysis process in the remainder of this paper is performed
for every possible value of $(x,y]$, and the overall asymptotic
competitive ratio is the maximum value of $R$ resulting in the
next procedure for a given value of $(x,y]$. That is, we analyze
the algorithm with respect to all possible scenarios (where there
is a large number of scenarios, but they can be analyzed
independently given the parameters of the algorithm), and as it is
not known in advance which scenario will occur, a worst-case
assumption is applied.

\subsection{Defining the weight function via a linear program \label{weight_sec}}

The first step for analyzing each scenario is to obtain a good
weight function for the scenario, in the sense that the analysis
will be as tight as possible and can be done using a computer
assisted proof within a small running time. The weight function
defines size based weights for values in $(0,1]$. The goal is to
define weights such that the cost of the algorithm is roughly the
total weight of all input items (formally, the total weight should
be at least the number of bins minus a constant $c$ independent of
the input), and if the target competitive ratio is $R$, the cost
of an optimal solution is at least the total weight divided by $R$
(this can be proved by showing that no bin can contain items of
total weight
above $R$). %%% We explain the exact usage of the weight function later.
Then, for an input $I$, letting $w(I)$ denote its total weight,
(and as defined above, letting $OPT(I)$ the optimal cost for $I$,
and $A(I)$ the number of bins used by $A$), we will have $A(I)
\leq w(I)+c$, $OPT(I) \geq \frac{w(I)}{R}$, which shows that $A(I)
\leq R \cdot OPT(I) +c$. This last argument is the standard
argument for weight functions based analysis. In \cite{Seiden02J}
generalizations of weight functions were used, but we just use the
approach of \cite{J73,J74,JDUGG74,LeeLee85,RaBrLL89}.

In order to define a suitable function, we will solve a linear
program defined below (this linear program has only four variables
$w$, $u$, $v$ and $R$, and in some cases it actually has only two
variables $w$ and $R$). More precisely, we will provide a feasible
solution for this linear program that is very close to the optimal
one (but we only use its feasibility and do not prove that it is
almost optimal). The weights of specific sizes will be based on
the values $w$, $u$, $v$ (or just on $w$, if the others are
undefined), and on some of the parameters of the algorithm (the
$\alpha_{ij}$ values for the given class). The variable $w$ will
be required to satisfy $0 \leq w \leq 1$. For $u$ and $w$, we also
require $0 \leq u,v \leq 1$, for the scenarios where these
variables are defined.

\subsubsection{The (minimum) required weight of a container}
We say that a class is basic if it is a small or tiny class, and
also if it is a large class that is not the threshold class. Thus,
if the threshold class is a large class it is not basic, and
otherwise it is basic. We define a quantity for each container.
This quantity will be  called the {\em required weight} of the
container, and its goal
 is to introduce a uniform value such that weights of items are
 defined based on these values, in order to satisfy all
 requirements. This quantity is defined for any basic class and for the class of huge items. If the threshold class $k$ is a large class, we keep this quantity undefined for that class.  In all other cases it is defined as
follows.

%Regarding the required weight of a container of a large class, we
%define this quantity only if this class is not the threshold class
%(and for containers of a large class in the case where that is the
%threshold class we keep this quantity undefined).

For a positive container of volume at least $a$, the required
weight of the container is $1$.  Note that this means that the
weight of a huge item of size at least $a$ will be $1$, and for
every other positive container of volume at least $1-x$, we will
ensure a weight of $1$ for the container. Recall that for any
positive container that is not a container of a huge item (for
which the volume is the exact size of the item), its volume is an
element of $V$. For a positive container of volume in the interval
$(1/2,a)$, the required weight of the container is denoted as $w$.
This will be a decision variable of the forthcoming linear
program. Thus, for a positive container of a class of items that
are not huge, the required weight of the container is $w$ if its
volume is at most $1-y$ and the required weight of the container
is $1$ if its volume is at least $1-x$ (and there are no
containers of volumes in $(1-y,1-x)$, except for possibly
containers of huge items). The last definition does not depend on
the exact value of $a$ but it depends on the scenario index.
However, for a container of one huge item, its required weight
depends on the exact value of $a$ (and the size of the huge item).
The reasoning is that a positive container of volume at least $a$
may be packed into a positive bin, while other positive containers
are packed in neutral bins.

Next, we consider negative containers.  The intuitive definition
of the required weight of a negative container is that it is $1$
if we cannot guarantee that it is matched to a positive container,
and otherwise it is $1-w$.

Formally, we partition the definition of required weight of a negative container to the following cases.

{\bf Assume that $\boldsymbol{a \geq 2/3}$.} Here, the threshold
class is of small or tiny items (it is a basic class). For a
negative container of volume in the interval $(1-a,1/2]$, the
required weight is $1$. Since $1-a \leq 1/3$, this also means that
the required weight of a negative container of volume larger than
$1/3$ is $1$. The required weight of the negative containers of
volume in the interval $(0,1-a]$ is $1-w$. All containers of
volume at most $1-a$ are of small or tiny items, and their volumes
are in $V$. Thus, no negative container has volume in the open
interval $(x,y)$ (but there might be negative containers with
volume $x$ or volume $y$), and thus we can refine the statement of
the required weight of a negative container in this case as
follows. For a negative container of volume at least $y$, the
required weight is $1$, while for a negative container of volume
at most $x$, the required weight is $1-w$.

{\bf Assume that $\boldsymbol{a < 2/3}$.} Here, the threshold
class $k$ is of large items (and it is not basic). Recall that the
required weight of a container of a large class is defined for all
large classes except for $k$.

For a negative container of volume in the interval $(1-a,1/2]$,
the required weight is $1$.  The required weight of the negative
containers of volume in the interval $(0,1-a]$ is $1-w$.  This
last rule can be stated in terms of $x$ and $y$ for a negative
container of a class that is not the threshold class.  For such
negative container, we define the required weight of the container
to be $1-w$ if its volume is at most $x$, and otherwise its
required weight is $1$. The threshold class (for this case where
it is not basic) is discussed later in more detail.

% For items in the threshold class (and since $a < \frac 23$, it is a large
% class) we will not use the required weight of its containers, and thus we % do not need to define it as a function of the index of the scenario.

\subsubsection{The (amortized) weight of an item.}\label{amor}
In order to define weights of items, for each class separately, we
will define the weight of an item of this class in an amortized
way that ensures that the sum of the weights of items of this
class will be approximately the sum of the required weights of
(all) its containers (excluding a constant number of such
containers per class).  The weight of an item in the threshold
class (if it is not basic) will be defined without using the
required weight of the containers of this class.

For a huge item, since its container is always a positive
container, its weight is $1$ if its size is at least $a$ and $w$
if it is smaller than $a$.

Consider next classes of items that are not huge.  For a basic
class $j$, we let $r_{x,y}(i,j)$ be the required weight of a type
$i$ container of class $j$ (recall that we deal with a specific
scenario defined by $(x,y]$). Note that $r_{x,y}(i,j) \in \{
w,1-w, 1\}$, and we have already defined this value for every
possible values of $x,y,i,j$ such that $j\neq k$ or $j=k > b$. For
such an item of a non tiny class $j$, let the weight of the item
be $$\omega_j=\frac{\sum_{i} \alpha_{ij} \cdot
r_{x,y}(i,j)}{\sum_i i\cdot \alpha_{ij}} .$$  This term is the
ratio between the average required weight of a container of this
class and the average number of items in a container of this
class. As $r_{x,y}(i,j)\leq 1$ for all $x,y,i,j$ and $i\geq 1$, we
get $\omega_j \leq 1$ for all basic classes $j\leq M$. Similarly,
for $j=M+1$, we define the weight of an item of this class to be
its size multiplied by
$$\rho=\frac{\sum_{i} \alpha_{i,M+1} \cdot r_{x,y}(i,M+1)}{\sum_i
(A_{i,M+1}-t_M) \cdot \alpha_{i,M+1}} .$$ The parameters will
always be chosen such that $\rho\leq 2$, as otherwise the value of
the competitive ratio will exceed $2$.

In order to define the weight of items of the threshold class $k$
for scenarios where it is not basic (this means that $1/3 < 1-a <
1/2$), we introduce the last two decision variables $u$ and $v$
(these are decision variables of the linear program below, and
together with $w$ and the competitive ratio value $R$ for this
scenario this will conclude the introduction of the decision
variables). For such an item, we let its weight be $u$ if its size
is at most $1-a$ and otherwise its weight is $v$. We will impose
the constraint
\begin{equation}\label{cons1}u\leq v. \end{equation}

Intuitively, we can see the output as if for every type $1$
container of this class we have a collection of
$\frac{\alpha_{2k}}{\alpha_{1k}}$ (fractions of) containers of
type $2$ associated with it (as this is the ratio between the
fractions of containers of the two types out of the total number
of containers of class $k$), such that one of the following
conditions hold:

The first option is that the type $1$ container  is a regular type
1 container and it is matched to some positive container. For this
case we assume that its weight is $u$ and each item in the
associated type 2 containers has weight $u$, however the
additional weight of $w$ of the positive container (the one that
is matched to the type 1 container, where we do not know its
class) helps us to obtain a sufficient total weight, which is the
total number of bins. Our actual claim will be simpler and this
discussion is provided just to motivate the constraints (a full
set of properties and proofs is given later).

That is, we will have the constraint

\begin{equation}\label{cons2}u \cdot (\alpha_{1k}+2\alpha_{2k}) +
w \cdot \alpha_{1k} \geq \alpha_{1k}+\alpha_{2k}. \end{equation}

As $\alpha_{1k}+\alpha_{2k}=1$, this inequality is equivalent to
\begin{equation}\label{cons2pr} u \cdot (1+\alpha_{2k}) + w \cdot
\alpha_{1k} \geq 1. \end{equation}

Otherwise, that is, the type $1$ container is not matched to a
positive container (and thus it is a temporary type 1 container).
In this case, the item of this type 1 container is of size larger
than $1-a$, and moreover (as we formally justify below) every
container of type $2$ in its associated containers of type $2$
satisfies either that it has (at least) one item of size larger
than $1-a$ or it is matched to a negative container.  In our
constraints we consider (only) the two extreme cases where all
associated containers of type $2$ are of a common case (as the
constraint for every intermediate case is a convex combination of
the two extreme constraints that we explain now). For the first
case where all associated containers of type $2$ (as well as the
type $1$ container) have at least one item of size larger than
$1-a$, we have the constraint
\begin{equation}\label{cons3} u \cdot \alpha_{2k} + v \cdot (\alpha_{1k}+\alpha_{2k}) \geq \alpha_{1k}+\alpha_{2k}. \end{equation}

This inequality is equivalent to  \begin{equation}\label{cons3pr} u \cdot \alpha_{2k} + v \geq 1. \end{equation}

 For the other case where every type 2 container is matched to a negative container and the type 1 container has an item of size larger than $1-a$, we have the constraint
 \begin{equation}\label{cons4}v \cdot \alpha_{1k}+ 2\cdot u \cdot \alpha_{2k}  + (1-w) \cdot \alpha_{2k} \geq \alpha_{1k}+\alpha_{2k}.
 \end{equation}

The last inequality is equivalent to \begin{equation}\label{cons4pr} v \cdot \alpha_{1k}+ 2\cdot
u \cdot \alpha_{2k}  + (1-w) \cdot \alpha_{2k} \geq 1. \end{equation}

We impose all the constraints (\ref{cons1}), (\ref{cons2}),
(\ref{cons3}), and (\ref{cons4})  to ensure that we allocate
sufficient weight in all cases. While we presented these
constraints using a pictorial fractional allocation of type 2
containers to the type 1 containers, our proof is not based on
such arguments. We will prove that these constraints are
sufficient to guarantee that the resulting weight function
satisfies that the cost of the algorithm is at most the total
weight of the items plus a constant (which is independent of the
input) for each scenario.

\subsection{The linear program}

In addition to these four constraints (\ref{cons1}),
(\ref{cons2}), (\ref{cons3}), and (\ref{cons4})  (ensuring that
this is indeed a valid weight function, as we show below), we have
the knapsack constraints expressing the following properties.

For every subset of items that can fit into one bin, the total
weight of the items is at most $R$. Next, we elaborate further on
these knapsack constraints.  We intend to ensure that the total
weight of items in any bin in an optimal solution is at most $R$.
To do this, we consider every possible subset of items that may
fit into a bin, and for each non tiny item we replace its size by
the infimum size of an item of the same weight.  The resulting set
of non tiny items has total size strictly smaller than $1$. We may
consider only sets with total size strictly smaller than $1$ and
not sets with total size of exactly $1$ as the infimum size of an
item of the same weight is never attained (except for one special
case, see below), and thus for every nonempty set of such items we
can strictly decrease the total size of the items in the set and
obtain another set of the same weight and smaller size. There is
one case where the minimum is attained, which is a huge item of
size $a$. For decreasing the total size of items strictly below
$1$ it is sufficient for the bin to have at least one other non
tiny item, and if the huge item of size $a$ is the only non tiny
item of a bin, its size is already below $1$. We will consider
sets of non tiny items, and to find an upper bound on the total
weight that could result from this set, we add to the multiset of
items sand of (strictly) positive total size consisting of an
arbitrary set of tiny items of total size that equals $1$ minus
the total size of the multiset of items we consider.

Any set of non tiny items (of total size strictly below $1$)
belongs to one of the following cases:

\begin{enumerate}
\item\label{fam1} Assume that there is a huge item of size at
least $a$ in the set.  Such a set of items contains (beside the
huge item) only items smaller than $1-a$, of total size below
$1-a$. Thus, the remaining items have total size below $y$ (and in
the calculation of total weight we will obviously take into
account the huge item, whose weight is $1$). If class $k$ is
large, the set could possibly contain an item of size in
$(t_k,t_{k-1}]$, but such an item has size smaller than $1-a \leq
y$, so its size is in $(t_k,1-a]$. We define a set of sizes with
weights, where the set is called $\Delta$, such that for every set
of non tiny items of sizes in $(t_M,1-a]$, there is a multiset of
items of $\Delta$ of the same weight, such that its total size is
not larger. The set $\Delta$ consists of all sizes $t_j$ for $k
\leq j \leq M$. The weight of $t_j$ for $k+1 \leq j \leq M$ is
defined to be the weight of an item of size in $(t_j,t_{j-1}]$. If
$k$ is a basic class, the weight of $t_k$ is defined to be the
weight of an item of size in $(t_k,t_{k-1}]$, and otherwise ($k$
is a large class), the weight of $t_k$ is defined to be $u$. Given
a set of items of sizes in $(t_M,1-a]$ of total size below $y$,
replacing any item whose size is in $(t_j,t_{j-1}]$ with an item
of size $t_j$ results in a total size that is not larger than the
original total size, and it has the same weight. Given a set of
items of sizes in $(t_M,1-a]$ of total size $z < y$ and a set of
tiny items of total size $y-z$, we obtain a set of items that has
a weight that is no larger than the weight of the corresponding
multiset of items of $\Delta$, whose total size is $z'\leq z$ plus
tiny items of total size $y-z' \geq y-z \geq 1-a-z$. Thus, we can
consider multisets of items of $\Delta$ instead of arbitrary sets
of non tiny items, and an upper bound on the weights of such
multisets together with tiny items, such that their total size is
exactly $y$ (plus the weight $1$ of the huge item) is an upper
bound on the total weight of any packed bin with a huge item of
size at least $a$.

\item\label{fam2} Assume that there is no huge item of size at
least $a$ in a considered set of items. In this case, the total
sizes of sets of non tiny items should be below $1$, and no item
has size of $a$ or more. The set $\Delta$ consists of all values
$t_j$ for $j=1,2,\ldots,M$, and the value $x$ is included as well
if $k$ is a large class (unless $x=t_k$). For any other class $j
\neq 1,k$ and also for $j=k$ if $k$ is not a large class, the
weight of an item of size $t_j$ of $\Delta$ is equal to the weight
of an item of size in $(t_j,t_{j-1}]$. Assume that $k$ is a large
class. In this case another two elements of $\Delta$ is an item of
size $t_k$ and weight $u$, and an item of size $x$ and weight $v$.
If $t_k=x$, there is only one additional element whose size is
$t_k=x$ and whose weight is $v$ (since $u \leq v$). The choice of
the weight of $x$ is based on the property that an item of size in
$(t_k,t_{k-1}]$ either has weight of $v$ or of $u \leq v$. The
weight of an item of size $t_1=\frac 12$ in $\Delta$ is $w$ (as
every item of size above $\frac 12$ has at least this weight).

In this case, we consider all multisets of items of sizes $t_j$
for $1 \leq j \leq M$ and $x$, where a set of tiny items whose
size is the complement to $1$ is added for the purpose of weight
calculation of a multiset. Once again the resulting value for
$\Delta$ is an upper bound on the value we would like to compute
for the original items (as we may have increased the weight of
some items from $u$ to $v$, and the total size of tiny items may
have increased).
\end{enumerate}

This provides us with two knapsack problems for each scenario, one
resulting from case \ref{fam1} and the other from case \ref{fam2}.
In the first one, the target total size is $y$ (such that the
total size of items of $\Delta$ is strictly below $y$), and the
upper bound on $R$ is $1$ plus the upper bound on the weight of
the multiset of items of $\Delta$ plus tiny items complementing
the total size to $y$. In the second one, the target total size is
$1$ (such that the total size of items of $\Delta$ is strictly
below $1$), and the upper bound on $R$ is the upper bound on the
weight of the multiset of items of $\Delta$ plus tiny items
complementing the total size to $1$. Under a worst-case
assumption, we calculate the maximum of the two values for each
scenario (and then the maximum of the upper bounds for all
scenarios).

%For each such set we should enforce a constraint saying that the
%total weight of the items (and the additional sand) is at most
%$R$.  The linear program is solved using the knapsack constraints
%resulting from both families \ref{fam1} and \ref{fam2} (one linear
%program per scenario), and its objective function value (of an
%optimal solution) is the asymptotic competitive ratio of that
%scenario.  The maximum of the scenarios will be the asymptotic
%competitive ratio of the algorithm.

\section{Some properties of the algorithm}

We state and prove a number of properties, where some of these
properties were mentioned above.  The goal of this section is to
prove formally that the weight function we define (for the given
scenario $(x,y]$) is indeed a valid weight function.

\subsection{Bounding the number of containers of each type}
We will bound the number of containers of each type
 in terms of the number of items of the class for small and large items, and in terms of total size for tiny
 items.

Our first goal is to bound the values of $n_{2j}$ and $n_{1j}$ for
a large class $j$ in terms of the number of items of class $j$. To
do that, we use the fact that large items can be packed in type
$1$ or type $2$ containers, and thus for such a class $j$, there
are only two values of $\alpha$, namely, $\alpha_{1j}$ and
$\alpha_{2j}$ (whose sum is $1$). All values are analyzed just
after an item has been packed or just before any item arrived. All
properties proved in this section for a large class $j$ hold even
if $\alpha_{1j}=0$.

\begin{lemma}\label{minibnum}
For any large class $j$, at any time just after an item was packed, $n_{2j} \leq \lfloor
\alpha_{2j}\cdot n_{j}\rfloor$, and $n_{1j} \leq \lfloor
\alpha_{1j}\cdot n_{j}\rfloor +2$.
\end{lemma}
\begin{proof}
We use induction. Initially, $n_j=n_{1j}=n_{2j}=0$ and the two
properties hold.

Consider the kind of all possible modifications in the set of
containers of class $j$. In the proof we use $n_j$, $n_{1j}$, and
$n_{2j}$ for the values before a modification and $n'_j$,
$n'_{1j}$, and $n'_{2j}$ for the values after the modification.
Thus, we will show $n'_{2j} \leq \lfloor \alpha_{2j}\cdot
n'_{j}\rfloor$, and $n'_{1j} \leq \lfloor \alpha_{1j}\cdot
n'_{j}\rfloor + 2$ (assuming that $n_{2j} \leq \lfloor
\alpha_{2j}\cdot n_{j}\rfloor$, and $n_{1j} \leq \lfloor
\alpha_{1j}\cdot n_{j}\rfloor + 2$ hold before the modifications).
The first kind of modifications is the result of the arrival of a
new item of class $j$. In this case, the following actions of the
algorithm are possible. The first one is that the new item is
added to a declared type 2 container to become a regular type 2
container. In this case neither the number of type 2 containers of
class $j$ nor the total number of containers for this class change
(so $n'_j=n_j$, $n'_{1j}=n_{1j}$, and $n'_{2j}=n_{2j}$). The next
option is where $n_{2j}=\lfloor \alpha_{2j}\cdot n_{j}\rfloor$ and
we will not place the new item into a type 2 container. In this
case a new container  of type 1 is formed, and as $n'_j=n_j+1$ and
$n'_{2j}=n_{2j}$, we find $n'_{2j} = n_{2j} \leq \lfloor
\alpha_{2j}\cdot n_{j}\rfloor \leq \lfloor \alpha_{2j}\cdot
n'_{j}\rfloor$. Moreover, $n'_{1j}=n'_j-n'_{2j}=n'_j-\lfloor
\alpha_{2j} \cdot (n'_{j}-1) \rfloor \leq
n'_j(\alpha_{1j}+\alpha_{2j})-\alpha_{2j}\cdot (n'_{j}-1)+
1=\alpha_{1j}n'_j +2 -\alpha_{1j} = \alpha_{1j}n_j +2 < \lfloor
\alpha_{1j}n_j \rfloor + 3 $ (and therefore $n'_{1j} \leq \lfloor
\alpha_{1j}n_j \rfloor + 2 $, as both the first and the last
expressions in the sequence of inequalities are integral).

Finally, if $n_{2j} \leq \lfloor \alpha_{2j}\cdot n_{j}\rfloor - 1
$, there are two options where we might put the item into a type 2
container. In the first option, one new (declared) type 2
container is created, and we have $n'_{2j}=n_{2j}+1$,
$n'_{1j}=n_{1j}$, and $n'_j=n_j+1$. In this case, we have $n'_{2j}
= n_{2j}+1 \leq (\lfloor \alpha_{2j}\cdot n_{j}\rfloor -1)+1 =
\lfloor \alpha_{2j}\cdot n_{j}\rfloor \leq \lfloor
\alpha_{2j}\cdot n'_{j}\rfloor$, and $n'_{1j}= n_{1j} \leq \lfloor
\alpha_{1j}\cdot n_{j}\rfloor +2 \leq \lfloor \alpha_{1j}\cdot
n'_{j}\rfloor +2$. In the second option, a (temporary) type 1
container is transformed into a (regular) type 2 container, and in
this case, $n'_{2j}=n_{2j}+1$, $n'_{1j}=n_{1j}-1$, and $n'_j=n_j$.
In this case, we have $n'_{2j} = n_{2j}+1 \leq (\lfloor
\alpha_{2j}\cdot n_{j}\rfloor -1)+1 = \lfloor \alpha_{2j}\cdot
n_{j}\rfloor = \lfloor \alpha_{2j}\cdot n'_{j}\rfloor$, and
$n'_{1j}= n_{1j} -1 \leq \lfloor \alpha_{1j}\cdot n_{j}\rfloor +2
= \lfloor \alpha_{1j}\cdot n'_{j}\rfloor + 2$.

A possible second kind of modifications is a result of the arrival of a
small or tiny item. In this case, the change can be that a
temporary type 1 container of class $j$ becomes a declared type 2
container, or that a temporary type 1 container of class $j$ becomes
a regular type 1 container. In the latter case there is no change in
the numbers of containers of types 1 and 2. In the former case,
$n'_j=n_j$, and the change is performed only if $n_{2j} \leq
\lfloor \alpha_{2j}\cdot n_{j}\rfloor - 1 $ (see Remark \ref{rmk-temp} for the summary of the relevant steps of the algorithm). In this case
$n'_{2j}=n_{2j}+1$, $n'_{1j}=n_{1j}-1$, and such a situation was
already considered.
\end{proof}

\begin{corollary}\label{minibnum2}
For any large class $j$, at any time just after packing an item, $n_{2j} \geq \alpha_{2j}
\cdot n_j -2$. Additionally, $n_{1j} \geq \alpha_{1j}\cdot n_{j}$.
\end{corollary}
\begin{proof}
By Lemma \ref{minibnum}, $n_{2j} = n_j-n_{1j} \geq n_j - (\lfloor
\alpha_{1j}\cdot n_{j}\rfloor +2) \geq n_j - \alpha_{1j}\cdot
n_{j}-2 =\alpha_{2j}n_j-2$, and $n_{1j}=n_j-n_{2j} \geq n_j-
\lfloor \alpha_{2j}\cdot n_{j}\rfloor \geq \alpha_{1j}n_j$.
\end{proof}

Recall that $N_j$ denotes the number of items of class $j$ (that arrived so far).  We next bound the number of containers of a large class $j$ in terms of $N_j$.

\begin{lemma}\label{itmnum}
During the action of the algorithm, for any large class $j$, it
holds that $n_j \leq \frac{N_j}{1+\alpha_{2j}} +
2=\frac{N_j}{2-\alpha_{1j}} + 2$ and $n_j \geq
\frac{N_j+\beta}{1+\alpha_{2j}}=\frac{N_j+\beta}{2-\alpha_{1j}}$,
where $\beta$ is the number of declared type 2 containers at this
time.
\end{lemma}
\begin{proof}
Initially $n_j=N_j=\beta = 0$ and the two inequalities hold. A new
container of class $j$ may be created when a new item of this
class arrives, that is, when $N_j$ increases (but a new item of
class $j$ does not always cause the creation of a new container,
as in some cases it is packed into an existing container for this
class). No existing containers can be destroyed, and therefore we
only consider an arrival of an item of class $j$. Assume that item
$i$ of class $j$ has just arrived and packed. Let $n_j$, $n_{2j}$,
$n_{1j}$, $N_j$, and $\beta$ be the values of these variables
prior to the arrival of $i$ and let $n'_j$, $n'_{2j}$, $n'_{1j}$,
$N'_j$, and $\beta '$ be their values after the arrival and
packing of $i$, and thus, $N'_j=N_j+1$. Furthermore, if a new
container is created, then $n'_j=n_j+1$ and otherwise $n'_j=n_j$.

Consider the first inequality. If $n'_j=n_j$, we are done. Since
the option of adding the new item into a declared type 2 container
is tested first, the creation of a new container means that there
were no such declared type 2 containers prior to the arrival of
item $i$. Thus, just before $i$ is presented to the algorithm,
every type 2 container has two items. We have $N_j=n_{1j}+2\cdot
n_{2j} = 2 n_j - n_{1j} $. By Lemma \ref{minibnum}, $n_{1j} \leq
\alpha_{1j}\cdot n_{j} + 2 $, and therefore $N_j \geq
(2-\alpha_{1j})\cdot n_j -2 = (1+\alpha_{2j})\cdot n_j  - 2$,
proving $N'_j-1 \geq (1+\alpha_{2j})\cdot (n'_j-1)-2$, or
alternatively, $N'_j \geq (1+\alpha_{2j})\cdot (n'_j-1) -1 \geq
(1+\alpha_{2j})\cdot (n'_j-2)$, as required.

Consider the second inequality. We prove this inequality directly
(i.e., without induction). We have $N_j= n_{1j}+2\cdot n_{2j}
-\beta$, due to the numbers of items in the different types of
containers. By Lemma \ref{minibnum}, $n_{2j} \leq \alpha_{2j}
n_j$, and we have $N_j= n_{j}+ n_{2j} -\beta \leq
n_j(1+\alpha_{2j})-\beta$. The inequality results from
rearranging.
\end{proof}

\begin{corollary}\label{relationss}
During the action of the algorithm, for any large class $j$, it
holds that $n_{2j} \leq \alpha_{2j}\cdot
(\frac{N_j}{1+\alpha_{2j}}) + 2$, $n_{1j} \leq \alpha_{1j}\cdot
(\frac{N_j}{1+\alpha_{2j}})+4$, $n_{2j} \geq \alpha_{2j} \cdot
(\frac{N_j}{1+\alpha_{2j}}) -2$, and $n_{1j} \geq \alpha_{1j}\cdot
(\frac{N_j}{1+\alpha_{2j}})$.
\end{corollary}
\begin{proof}
The inequalities follow from Lemma \ref{minibnum}, Corollary
\ref{minibnum2}, and Lemma \ref{itmnum} using $\beta \geq 0$.
\end{proof}

\begin{lemma}\label{four_declared_lem}
For every large class $j$, there are at most four declared type 2
containers at each time.
\end{lemma}
\begin{proof}
Assume by contradiction that at a given time there are at least
five declared type 2 containers of class $j$. Then, using the
numbers of large items of class $j$ in all four types of containers
for this class, $N_j \leq 2\cdot n_{2j}+n_{1j}-5 = n_j +n_{2j}-5
$. Using Lemma \ref{itmnum} and Corollary \ref{relationss}, we
have $N_j \leq n_j +n_{2j}-5\leq
(\frac{N_j}{1+\alpha_{2j}})(1+\alpha_{2j})-1 < N_j$, a
contradiction.
\end{proof}

Next, we consider the case where $j$ is a small or a tiny class.

\begin{lemma}\label{one_open_container_lem}
For any class $j$ of small  items, there is at most one value $i$
(such that $\alpha_{ij}>0$), for which there is a container of
class $j$ and type $i$ with less than $i$ items, and there is at
most one such container (and for any $i' \neq i$, every container
of class $j$ and type $i'$ has exactly $i'$ items). For the class
of tiny items,  there is at most one value of $i$ for which there
is a container of class $j$ and type $i$ with total size at most
$A_{i,M+1} -t_M$, and there is at most one such container (and for
any $i' \neq i$, every container of class $j$ and type $i'$ has a
total size larger than $A_{i',M+1} -t_M$).
\end{lemma}
\begin{proof}
The lemma holds because for $j\geq b$ there is at most one open
container of class $j$ and (if $j \leq M$) such a container is of
some type $i$ for which $\alpha_{ij}>0$ (as the algorithm does not
open a new container of class $j$ until the previous open
container of this class is closed).
\end{proof}

\subsection{Analysis of the total weight of bins of the algorithm}
Our analysis is partitioned into two cases.  We first analyze the
total weight of items that are packed in bins containing
containers of the threshold class (assuming it is not a basic
class, otherwise there is no special analysis for this class). We
consider all other bins afterwards.

\subsubsection{Bins containing containers of the threshold class}
We start with the analysis of bins containing containers of the
threshold class, for the case where the threshold class is a large
class.

Consider the threshold class $k$ and assume that it is a large
class. We say that a temporary type 1 container of this class is
{\it smaller} if its volume is at most $1-a$, and that it is {\it
bigger} if its volume is above $1-a$. We note that smaller
temporary type 1 containers do not exist at termination, but we
will analyze arbitrary times during the execution of the
algorithm.

For this class, we use the following notation for the analysis.
Let $\nu(i)$ denote the number of containers of class $k$ after
$i$ items have arrived (and have been packed, this time is called
time $i$). Out of those containers, let $\nu_1(i)$ and $\nu_2(i)$
denote the numbers of containers of types 1 and 2, respectively
(so that the assignment of item $i$ is based on $n_{\ell
k}=\nu_{\ell}(i-1)$ for $\ell=1,2$, and
$\nu(i)=\nu_1(i)+\nu_2(i)$). Furthermore, let $\nu^r_1(i)$,
$\nu^{ts}_1(i)$, and $\nu^{tb}_1(i)$, denote the numbers of
regular type 1 containers of class $k$, smaller temporary type 1
containers of class $k$, and bigger temporary type 1 containers of
class $k$, respectively, after item $i$ has been packed (so
$\nu_1(i)=\nu^r_1(i)+\nu^{ts}_1(i)+\nu^{tb}_1(i)$), and let
$\nu^d_2(i)$ and $\nu^r_2(i)$ denote the numbers of declared and
regular type 2 containers of class $k$, respectively, after item
$i$ has been packed (so $\nu_2(i)=\nu^d_2(i)+\nu^r_2(i)$). Let
$\tau$ be the minimum index of an item such that for any $i>
\tau$, $\nu^{tb}_1(i)>0$ (note that the bigger temporary type 1
container created at time $\tau+1$ may change its type later on,
we only guarantee that there will always be a bigger temporary
type 1 container at all times after $\tau$). Letting $\mu$ denote
the total number of items in the input, and we use $\tau=\mu$ if
at termination there are no bigger temporary type 1 containers of
class $k$ (i.e., if $\nu^{tb}_1(\mu)=0$). Since (as argued above)
at termination there are no smaller temporary type 1 containers,
$\nu^{tb}_1(\mu)=0$ means that all type 1 containers of the output
are regular type 1 containers (this special case can be analyzed
more easily, but it will be included in the general analysis). The
case $\tau=0$ is possible, and in this case there is always a
bigger temporary type 1 container of class $k$.

Consider the case $\tau<\mu$, that is, there is at least one
additional input item after item $\tau$. By the definition of
$\tau$, item $\tau+1$ is of class $k$, its size is above $1-a$,
and a temporary type 1 container is created for it. Assume that
there exists at least one smaller temporary type 1 container of
class $k$ after item $\tau$ has been packed. All the smaller
temporary type 1 containers existing after item $\tau$ is packed
will exist also after item $\tau+1$ has been packed. In the next
lemma we show that all these containers will become regular type 1
containers.

\begin{lemma}\label{stillmb}
Consider a smaller temporary type 1 container of class $k$
existing at time $\tau$. This container will become a regular type
1 container before termination (and in particular it will not
become a type 2 container).
\end{lemma}
\begin{proof}
As there are no smaller temporary type 1 containers at
termination, this container changes its type some time during the
arrival of items $\tau+2,\ldots,\mu$. We will show that it does
not become a type 2 container. In all cases where a temporary type
1 container becomes a type 2 container (no matter whether it
becomes a regular type 2 container or a declared type 2
container), the largest available temporary type 1 container of
the class is selected (see Remark \ref{rmk-temp}). Here, any
temporary type 1 container of class $k$ can be used, in the sense
that all type 2 containers of class $k$ have the same volume, so
the chosen temporary type 1 container of class $k$ is always the
largest one. Thus, if a smaller temporary type 1 container is
chosen, this means that there is no bigger temporary type 1
container, contradicting the choice of $\tau$.
\end{proof}

Let $N_k(i)$ denote the number of class $k$ items out of the first
$i$ arriving items (the class $k$ items existing at time $i$).

We say that a type 2 container of a large class $k$ is {\it convenient} if
it is combined in a bin with a negative container or if it has at
least one item of size above $1-a$.

In what follows, we will say that a type 2 container is {\it
created} at a certain time if this container was just defined at
this time and it was defined as a declared type 2 container
immediately (of class $k$) or if it was a temporary type 1
container and its type was just changed to type 2 (regular or
declared). That is, whenever the number of type 2 containers of
class $k$ increases, the container responsible for this change is
considered to be created. Note that type 2 containers remain type
2 containers (of the same class) till termination. For a type 1
container it is created simply when an item is packed into it.

\begin{lemma}\label{convi}
Every type 2 container of class $k$ created at time $\tau+1$ or later is convenient.
\end{lemma}
\begin{proof}
The only case where a type 2 container of class $k$ is created and
it is not combined with a negative container in a bin immediately
is the case where an item of class $k$ just arrived (see Remark
\ref{rmk-type2}). Starting time $\tau+1$ there is always a
temporary type 1 container (of class $k$), so a declared type 2
container of class $k$ cannot be created unless it is combined
with a negative container in a bin immediately. Thus, it remains
to consider the case where the new item of class $k$ is added to a
temporary type 1 container to create a regular type 2 container
(of class $k$). Since after this is done there is still a bigger
temporary type 1 container of class $k$ (by the choice of $\tau$)
and such a container of maximum volume was selected to become a
regular type 2 container, the created type 2 container also has an
item of size above $1-a$ (its first item is such).
\end{proof}

\begin{lemma}\label{thosenicebins}
The total number of temporary type 1 containers of class $k$ at
termination of the algorithm
 is below
$$\alpha_{1k}\frac{N_k(\mu)-N_k(\tau)}{1+\alpha_{2k}}+5 \ . $$

The total number of convenient type 2 containers of class $k$ is
larger than

$$\alpha_{2k}\frac{N_k(\mu)-N_k(\tau)}{1+\alpha_{2k}}-5 \ . $$
\end{lemma}
\begin{proof}
Consider a temporary type 1 container of class $k$ that is present
at termination. As all temporary type 1 containers of class $k$
existing at termination are bigger, it was created no earlier than
time $\tau+1$ (i.e., not prior to the packing of the $\tau+1$-th
item). Indeed, after it is created, there will be such bigger
temporary type 1 container present at all times (as this container
is present at termination), and moreover, there is at least one
bigger temporary type 1 container of class $k$ at all times after
the creation of such a container when packing the $\tau+1$-th item
(but that specific container created at time $\tau+1$ does not
necessarily remain a temporary type 1 container until
termination).

We will use Corollary \ref{relationss} for times $\tau$ and $\mu$.
After $N_k(\tau)$ items of class $k$ have arrived, there are at
most $\alpha_{2k} \cdot (\frac{N_k(\tau)}{1+\alpha_{2k}})+2$
containers of type 2 and at least $ \alpha_{1k}\cdot
(\frac{N_k(\tau)}{1+\alpha_{2k}}) $ containers of type 1 of class
$k$. At termination, using the same corollary, there are at least
$ \alpha_{2k}\cdot (\frac{N_k(\mu)}{1+\alpha_{2k}}) -2 $
containers of type 2 and at most $ \alpha_{1k}\cdot
(\frac{N_k(\mu)}{1+\alpha_{2k}}) +4$ containers of type 1 of class
$k$.

Assume that there are at least
$\alpha_{1k}\frac{N_k(\mu)-N_k(\tau)}{1+\alpha_{2k}}+5$ temporary
type 1 containers of class $k$ at termination. Every type 1
container of class $k$ existing at time $\tau$ (all of them are
either regular type 1 containers or smaller temporary type 1
containers at time $\tau$) will still be a type 1 container at
termination (regular type 1 containers remain such, and by Lemma
\ref{stillmb} temporary type 1 containers become regular type 1
containers). Thus, the total number of type 1 containers of class
$k$ at termination is at least their number after packing the
$\tau$-th item plus the number of temporary type 1 containers of
this class (at termination), as those were created at time
$\tau+1$ or later. In total, we get at least
$$\alpha_{1k}\frac{N_k(\mu)-N_k(\tau)}{1+\alpha_{2k}}+5 +
\alpha_{1k}\cdot \frac{N_k(\tau)}{1+\alpha_{2k}}= \alpha_{1k}\cdot
\frac{N_k(\mu)}{1+\alpha_{2k}}+5$$
 containers of type 1
of class $k$ at termination, a contradiction.

Next, assume that at most
$\alpha_{2k}\frac{N_k(\mu)-N_k(\tau)}{1+\alpha_{2k}}-5$ type 2
containers of class $k$ are created starting time $\tau+1$ (by
Lemma \ref{convi} they are all convenient). All type 2 containers
remain such, so the total number of type 2 containers of class $k$
is at least the number of these containers at time $\tau$ plus the
number of such containers created starting time $\tau+1$. This
number is at most
$$\alpha_{2k} \cdot \frac{N_k(\tau)}{1+\alpha_{2k}} +2  +
\alpha_{2k}\frac{N_k(\mu)-N_k(\tau)}{1+\alpha_{2k}}-5 =
\alpha_{2k} \cdot \frac{N_k(\mu)}{1+\alpha_{2k}}-3 , $$ a
contradiction.
\end{proof}

Recall that our parameters are selected such that $\alpha_{2j}>0$
for every large class $j$.

\begin{corollary}\label{cor-conv}
Let $C$ be the number of convenient type $2$ containers of class $k$ at termination.  Then, the number of temporary type 1 containers of class $k$ at termination is at most $\frac{\alpha_{1k}}{\alpha_{2k}}(C+5)+5$.
\end{corollary}

Next, we consider the weight function we defined using the values
of $u,v,w$ that have real values in $[0,1]$ and satisfy the
constraints (\ref{cons1}), (\ref{cons2}), (\ref{cons3}), and
(\ref{cons4}). We consider the total weight of the items of class
$k$ together with the required weight of the containers (of other
classes) that are packed together (i.e., in common bins) with the
containers of class $k$. We let $\phi_k$ denote the total weight
of the items of class $k$ together with the required weight of the
containers (of other classes) that are packed together with the
containers of class $k$. Recall that $\nu(i)$ is the number of
bins containing containers of class $k$ after $i$ items are
packed.

\begin{lemma}\label{consts1}
If $\alpha_{1k}=0$, then $\phi_k \geq \nu(\mu)-3$.
\end{lemma}
\begin{proof}
Since $\alpha_{1k}=0$, by Lemma \ref{minibnum}, every container of
class $k$ is a type 2 container, except for at most two containers
(and thus there are at most six containers of this class with
exactly one item). By Lemma \ref{four_declared_lem}, the number of
items in these $\nu(\mu)$ containers of class $k$ is at least
$2\nu(\mu)-6$, and each such item has weight of at least $u$
(using constraint (\ref{cons1})). Thus, $\phi_k \geq 2u \cdot
(\nu(\mu)-3) \geq \nu(\mu)-3$ where the last inequality holds by
constraint (\ref{cons2}) which is equivalent for this case to the
constraint $2u \geq 1$.
\end{proof}

Thus, we next assume that $\alpha_{1k} >0$.  Thus, in the next
lemma we assume that $\alpha_{1k}, \alpha_{2k} >0$. Let $\lambda =
5\cdot \frac{\alpha_{1k}}{\alpha_{2k}} +5$.

\begin{lemma}\label{consts2}
Assume that $\alpha_{1k}>0$ holds.   Then $\phi_k \geq
\nu(\mu)-2\lambda -8 -\frac{\lambda}{\alpha_{1k}}$.
\end{lemma}
\begin{proof}
The number of containers of class $k$ at termination is
$\nu(\mu)$, and this is also the number of bins containing such
containers.  To ease the description of the proof below, we let
the {\it weight} of a container of a class not equal to $k$ but
packed with an item of class $k$ in a bin to be its required
weight.  Thus, the total weight of a bin $B$ is the total weight
of items of class $k$ packed into $B$ together with the required
weight of a container packed into $B$ of a class not equal to $k$
(if there is such a container).

For a declared type 2 container of class $k$, the total weight of
the container is at least $u$. For a regular type 2 container of
class $k$, there are three cases. If this container is convenient
in the sense that it has an item of size above $1-a$, then the
weight of the bin is at least $u+v$. If it is convenient in the
sense that this bin also contains a negative container, then the
weight of the bin is at least $2u+(1-w)$. In any case, the weight
of a bin containing a regular type 2 container is no smaller than
$2u$ (using constraint (\ref{cons1}) and $0\leq w \leq 1$).

Let $C_1$ and $C_2$ denote the numbers of the two kinds of
convenient type 2 containers, respectively (where $C_1$ is the
number of convenient containers with an item of size above $1-a$
and $C_2$ is the number of all other convenient containers). Then,
as there are at most four declared type $2$ containers, the total
weight of bins containing type 2 containers of class $k$ is at
least \begin{equation}\label{lem16-1}2u\cdot
(\nu_2(\mu)-4-C_1-C_2)+4u+C_1\cdot(u+v)+C_2\cdot(2u+1-w)
.\end{equation} Here, the first expression in the sum (i.e.,
$2u\cdot (\nu_2(\mu)-4-C_1-C_2)+4u=2u\cdot
(\nu_2(\mu)-2-C_1-C_2)$) is a lower bound on the total weight of
bins containing containers of class $k$ that are not convenient.

For a temporary type 1 container of class $k$, as all bins
containing such a containers have (at termination) exactly one
item, and its size is above $1-a$, the weight of such a bin is
$v$. Recall that a type 1 container that is matched to a positive
container is always a regular type 1 container.  Thus, for a
regular type 1 container, it is combined with a positive container
in its bin, and therefore the total weight of such a bin is at
least $u+w$ (once again using constraint (\ref{cons1}) and $w \leq
1$). Letting $C_3$ denote the final number of temporary type 1
containers of class $k$, the total weight of all bins containing a
type 1 container of class $k$ is at least
\begin{equation}\label{lem16-1a} (w+u) \cdot (\nu_1(\mu)-C_3) + vC_3
\end{equation} because $\nu_1(\mu)-C_3$ is the number of regular type 1
containers and $C_3$ is the number of temporary type 1 containers.

Thus, we have \begin{equation}\label{lem16-2}\phi_k \geq 2u\cdot
(\nu_2(\mu)-2-C_1-C_2)+C_1 \cdot (u+v)+C_2\cdot(2u+1-w) + (w+u)
\cdot (\nu_1(\mu)-C_3) + v\cdot C_3 .\end{equation}

If $C_3-\lambda<0$, we use simpler properties as follows. The
total weight of bins with type 2 containers is at least $2u\cdot
(\nu_2(\mu)-4)+4u$. In total, we get a lower estimation on the
total weight of bins containing items of class $k$ of $\phi_k \geq
2u\cdot (\nu_2(\mu)-2)+(w+u) \cdot (\nu_1(\mu)-C_3)$, similarly to
(\ref{lem16-1}) and (\ref{lem16-1a}). As $(w+u)C_3 \leq 2C_3 <
2\lambda$ because of our assumption and by $4u\leq 4$, we have
using Corollary \ref{minibnum2} that
$$\phi_k \geq 2u\cdot \nu_2(\mu) +(w+u) \cdot \nu_1(\mu) -2\lambda -4 \geq 2u\cdot
(\alpha_{2k}\nu(\mu)-2) +(w+u) \cdot (\alpha_{1k}\nu(\mu)) -2\lambda -4 $$
$$\geq \nu(\mu)(u(1+\alpha_{2k})+w\alpha_{1k})-2\lambda -8 \geq \nu(\mu)-2\lambda -8 \ ,
$$ where in the second inequality we used
$\nu_1(\mu) \geq \alpha_{1k}\nu(\mu)$ and $\nu_2(\mu) \geq
\alpha_{2k}\nu(\mu)-2$ (which holds by Lemma \ref{relationss}),
and the last inequality follows by constraint (\ref{cons2pr}), and
the lemma follows.

Thus, in the remaining part of the proof, we assume that $C_3-\lambda \geq 0$.

By Corollary \ref{cor-conv}, $C_3 -\lambda \leq
\frac{\alpha_{1k}}{\alpha_{2k}}(C_1 + C_2)$. Let
$C'_1=\frac{\alpha_{2k}}{\alpha_{1k}} \cdot (C_3 -\lambda) -C_2 \leq C_1$
and $C'_2=C_2$. If $C'_1<0$, instead of these values we let
$C'_1=0$ and $C'_2=\frac{\alpha_{2k}}{\alpha_{1k}}(C_3-\lambda) \leq C_2$
(and $C'_2 \geq 0$, as the case where this last value is negative
was considered earlier in the case $C_3-\lambda <0$).

Based on (\ref{lem16-1}), the total weight of bins with type 2
containers is at least $2u\cdot
(\nu_2(\mu)-2-C_1-C_2)+C_1(u+v)+C_2(2u+1-w) \geq 2u\cdot
(\nu_2(\mu)-2-C_1-C_2)+C'_1(u+v)+(C_1-C'_1)2u
+C'_2(2u+1-w)+(C_2-C'_2)\cdot(2u)=2u\cdot
(\nu_2(\mu)-2-C'_1-C'_2)+C'_1(u+v)+C'_2(2u+1-w)$.

In total, we get similarly to (\ref{lem16-2}) that $\phi_k \geq
2u\cdot (\nu_2(\mu)-2-C'_1-C'_2)+C'_1(u+v)+C'_2(2u+1-w)+(w+u)
\cdot (\nu_1(\mu)-C_3) + vC_3$.

Since in both possible definitions of $C'_1$ and $C'_2$,
\begin{equation}\label{lem16-3}
C_3=\frac{\alpha_{1k}}{\alpha_{2k}}\cdot (C'_1+C'_2) +\lambda \ .
\end{equation}  We finally get

\begin{eqnarray}
\phi_k &\geq& 2u\cdot (\nu_2(\mu)-2-C'_1-C'_2)+ C'_1(u+v)
\nonumber \\ && +C'_2(2u+1-w)+(w+u) \cdot
(\nu_1(\mu)-C_3) + vC_3 \nonumber \\
&\geq & 2u\cdot (\alpha_{2k}\nu(\mu)-C'_1-C'_2)+(w+u) \cdot
(\alpha_{1k}\nu(\mu)-C_3) \nonumber \\ && + C'_1(u+v)+C'_2(2u+1-w) + vC_3 -8u \label{ineql1}\\
&=& 2u\cdot
(\alpha_{2k}\nu(\mu)-\frac{\alpha_{2k}}{\alpha_{1k}}(C_3-\lambda))+(w+u)
\cdot (\alpha_{1k}\nu(\mu)-C_3) \nonumber \\ && + C'_1(u+v)+C'_2(2u+1-w) + vC_3 -8u \label{ineql2}\\
&=& (\nu(\mu) - \frac{C_3}{\alpha_{1k}} ) \cdot (u(1+\alpha_{2k})+\alpha_{1k}w)  \nonumber \\ && + C'_1(u+v)+C'_2(2u+1-w) + vC_3 -8u + 2u \cdot \frac{\alpha_{2k}}{\alpha_{1k}} \lambda \label{ineql3}\\
&\geq& \nu(\mu) - \frac{C_3}{\alpha_{1k}} + C'_1(u+v)+C'_2(2u+1-w) + vC_3 -8u + 2u \cdot \frac{\alpha_{2k}}{\alpha_{1k}} \lambda \label{ineql4}\\
&\geq& \nu(\mu) - \frac{C_3}{\alpha_{1k}} + C'_1(u+v)+ \frac{C'_2}{\alpha_{2k}}-\frac{\alpha_{1k}}{\alpha_{2k}} \cdot vC'_2 + vC_3 -8u + 2u \cdot \frac{\alpha_{2k}}{\alpha_{1k}} \lambda \label{ineql5}\\
&=& \nu(\mu) - \frac{C_3}{\alpha_{1k}} + C'_1(u+v)+ \frac{C'_2}{\alpha_{2k}}-\frac{\alpha_{1k}}{\alpha_{2k}} \cdot vC'_2 \nonumber\\&& + v\frac{\alpha_{1k}}{\alpha_{2k}} (C'_1+C'_2) -8u + 2u \cdot \frac{\alpha_{2k}}{\alpha_{1k}} \lambda +v\lambda \label{ineql6} \\
&=& \nu(\mu) - \frac{C_3}{\alpha_{1k}} + C'_1(u+v+v\frac{\alpha_{1k}}{\alpha_{2k}})+ \frac{C'_2}{\alpha_{2k}}-8u + 2u \cdot \frac{\alpha_{2k}}{\alpha_{1k}} \lambda +v\lambda \label{ineql7}\\
&\geq& \nu(\mu) - \frac{C_3}{\alpha_{1k}} + \frac{C'_1}{\alpha_{2k}}+ \frac{C'_2}{\alpha_{2k}}-8u + 2u \cdot \frac{\alpha_{2k}}{\alpha_{1k}} \lambda +v\lambda \label{ineql8}\\
&=& \nu(\mu) -8u + 2u \cdot \frac{\alpha_{2k}}{\alpha_{1k}} \lambda +v\lambda - \frac{\lambda}{\alpha_{1k}} , \label{ineql9}
\end{eqnarray}
where (\ref{ineql1}) follows by $\nu_1(\mu) \geq
\alpha_{1k}\nu(\mu)$ and $\nu_2(\mu) \geq \alpha_{2k}\nu(\mu)-2$
(by Lemma \ref{relationss}), (\ref{ineql2}) follows by
(\ref{lem16-3}), (\ref{ineql3}) holds by simple algebraic
transformation and by substituting $\alpha_{1k}+\alpha_{2k}=1$,
(\ref{ineql4}) holds by constraint (\ref{cons2pr}), (\ref{ineql5})
follows by constraint (\ref{cons4pr}) and by applying
$\alpha_{1k}+\alpha_{2k}=1$ again,  (\ref{ineql6}) follows by
(\ref{lem16-3}), (\ref{ineql7}) holds by simple algebraic
transformations, (\ref{ineql8}) holds because
$u+v+v(\frac{\alpha_{1k}}{\alpha_{2k}})=\frac{1}{\alpha_{2k}}(\alpha_{2k}u+v)
\geq \frac{1}{\alpha_{2k}}$ where the last inequality holds using
constraint (\ref{cons3pr}), and (\ref{ineql9}) holds by
(\ref{lem16-3}).  The claim follows using $u,v \geq 0$.
\end{proof}

%\newpage

\subsubsection{The total weight of items of basic classes and the class of huge items}

\begin{lemma}
For any bin (excluding bins containing at least one item of class
$k$, if $k$ is not a basic class), the total required weight of
the containers that are packed in this bin is at least $1$.
\end{lemma}
\begin{proof}
If the bin contains both a positive container and a negative
container, we are done, as their required weights are at least $w$
and $1-w$, respectively (using $0\leq w \leq 1$). A positive
container that was not combined with a negative one has volume of
at least $a$, and a negative container that was not combined with
a positive one has volume above $1-a$. Such containers have
required weights of $1$.
\end{proof}

Next, we show that for a basic class and the class of huge items,
we have that the total required weight of the containers of class
$j$ is at most the total weight of the items of class $j$ plus a
constant.

\begin{lemma}
If $j=1$, then the total required weight of the containers of class $j$ equals the total weight of the items of class $j$.
\end{lemma}
\begin{proof}
The lemma follows by our definition of a weight of a huge item (it
is $w$ if its size is smaller than $a$ and $1$ otherwise).
\end{proof}

For a basic class $2 \leq j\leq M+1$, let $\zeta_j$ be the number
of strictly positive values of $\alpha_{ij}$. For $j=M+1$, we let
$\gamma_{M+1}=\zeta_{M+1}=p$. For $j \leq M$, we have $\zeta_j
\leq \gamma_j$, and for our parameters we actually have $\zeta_j
\leq 3$ for all $j$, and $\zeta_j=2$ for most values of $j$ (but
we sometimes have $\zeta_j=3$ and this is an important new feature
of our algorithm). Let ${\cal{R}}_j$ be the total required weight
of all containers of class $j$, and let $W_j$ be the total weight
of items of class $j$. By definition, letting $n_{ij}$ denote the
number of containers of class $j$ and type $i$ (for class $M+1$ it
is denoted by $n_{i,M+1}$), ${\cal{R}}_j=\sum_{i=1}^{\gamma_j}
r_{x,y}(i,j) \cdot n_{ij}$ and for $j \leq M$, $W_j = \omega_j
\cdot N_j$.

%We will use the property that for all classes we have $\sum_i
%r_{x,y}(i,j)\alpha_{ij} \leq \sum_i \alpha_{ij} =1$, as
%$r_{x,y}(i,j) \leq 1$ for all $x,y,i,j$.

\begin{lemma}
For any small or tiny class $j$, ${\cal{R}}_j \leq n_j\cdot
\sum_{i=1}^{\gamma_j} \alpha_{ij}\cdot r_{x,y}(i,j)+\zeta_j$. For
any basic large class $j'$, ${\cal{R}}_{j'} \leq n_{j'}\cdot
\sum_{i=1}^{2} \alpha_{ij'}\cdot r_{x,y}(i,j')+2$.
\end{lemma}
\begin{proof}
For any small or tiny class $j$, a container of type $i$ is opened
only in the case where there are at most $\lfloor \alpha_{ij}\cdot
\hat{n}_j \rfloor$ such containers, where $\hat{n}_j$ is the
number of containers of class $j$ before the new container is
opened. Moreover, it is never opened if $\alpha_{ij}=0$. Thus, the
last container of class $j$ and type $i$ was opened when there
were at most $\lfloor\alpha_{ij}\cdot (n_j-1)\rfloor$ such
containers and finally there are at most $\lfloor\alpha_{ij}\cdot
(n_j-1)\rfloor+1 < \alpha_{ij}n_j+1$ such containers. By
$r_{x,y}(i,j) \leq 1$ we have
$${\cal{R}}_j=\sum_{i=1}^{\gamma_j} n_{ij}\cdot r_{x,y}(i,j) \leq
\sum_{i: \alpha_{ij}>0} ( \alpha_{ij}n_j+1)\cdot r_{x,y}(i,j) \leq
n_j\cdot \sum_{i=1}^{\gamma_j} \alpha_{ij}\cdot
r_{x,y}(i,j)+\zeta_j \ . $$

For any basic large class $j'$, $n_{1j'} \leq \alpha_{1j'}\cdot
n_{j'}+2$ and $n_{2j'} \leq \alpha_{2j'}\cdot n_{j'}$ by Lemma
\ref{minibnum}. Thus, ${\cal{R}}_{j'}=\sum_{i=1}^{2} n_{ij'}\cdot
r_{x,y}(i,j') \leq \sum_{i=1}^{2} \alpha_{ij'} n_{j'}\cdot
r_{x,y}(i,j')+2r_{x,y}(1,j') \leq n_{j'} \sum_{i=1}^2 \alpha_{ij'}
\cdot r_{x,y}(i,j') + 2$ (by $r_{x,y}(i,j') \leq 1$).
\end{proof}

\begin{lemma}
For the tiny class $M+1$, $W_{M+1} \geq  n_{M+1}\cdot
\sum_{i=1}^{p} \alpha_{i,M+1}\cdot r_{x,y}(i,M+1)-2p$.
\end{lemma}
\begin{proof}
For class $M+1$, we have $W_{M+1} \geq \rho \cdot (\sum_{i=1}^{p}
(A_{i,M+1}-t_M) \cdot n_{i,M+1} - 1) $, as a class $M+1$ type $i$
container has items of total size of at least $A_{i,M+1}-t_M$,
except for at most one container of class $M+1$ and some type (and
$A_{i,M+1}-t_M\leq 1$ for all $i$). Let
$n_{i,M+1}=\alpha_{i,M+1}n_{M+1}+\delta_{i,M+1}$. For every $i$ we
have $\delta_{i,M+1}\leq 1$. As $$n_{M+1}=\sum_{i=1}^{p}
n_{i,M+1}=\sum_{i=1}^{p} (\alpha_{i,M+1}n_{M+1}+\delta_{i,M+1})$$
$$=n_{M+1}\sum_{i=1}^{p}
\alpha_{i,M+1}+\sum_{i=1}^{p}\delta_{M+1}=n_{M+1}+\sum_{i=1}^{p}\delta_{i,M+1}
\ ,$$ we get $\sum_{i=1}^{p}\delta_{i,M+1}=0$, which implies
$\sum_{i=1}^{p} \min\{0,\delta_{i,M+1}\}+\sum_{i=1}^{p}
\max\{0,\delta_{i,M+1}\}=0$. If for all $i$ we have
$\delta_{i,M+1}\geq 0$, then $\delta_{i,M+1}=0$ for all $i$.
Otherwise, if there is a value $i$ such that $\delta_{i,M+1}>0$,
there is also at least one negative value. Thus, $ -
\sum_{i=1}^{p} \min\{0,\delta_{i,M+1}\} = \sum_{i=1}^{p}
\max\{0,\delta_{i,M+1}\} \leq p-1$, and $$ \sum_{i=1}^{p}
(A_{i,M+1}-t_M) \cdot n_{i,M+1} = \sum_{i=1}^{p} (A_{i,M+1}-t_M)
\cdot( \alpha_{i,M+1} n_{M+1} +\delta_{i,M+1}) $$ $$= n_{M+1}
\sum_{i=1}^{p} (A_{i,M+1}-t_M) \cdot \alpha_{i,M+1}
+\sum_{i=1}^{p} (A_{i,M+1}-t_M) \cdot \delta_{i,M+1}
$$ $$\geq n_{M+1} \sum_{i=1}^{p} (A_{i,M+1}-t_M) \cdot \alpha_{i,M+1}
+\sum_{i=1}^{p}  \min\{0,\delta_{i,M+1}\} \ . $$ Therefore,
$W_{M+1} \geq \rho (\sum_{i=1}^{p} (A_{i,M+1}-t_M) \cdot n_{i,M+1}
- 1) \geq n_{M+1} \sum_{i=1}^{p} \alpha_{i,M+1} \cdot
r_{x,y}(i,M+1) -2p$, as $ - \sum_{i=1}^{p} (A_{i,M+1}-t_M) \cdot
\min\{0,\delta_{i,M+1}\} \leq p-1)$ and $\rho\leq 2$.
\end{proof}

\begin{lemma}
For any basic large class $j$, $W_{j} \geq n_{j}\cdot
\sum_{i=1}^{2} \alpha_{ij} \cdot r_{x,y}(i,j)-4$.
\end{lemma}
\begin{proof}
For a (basic) large class $j$, we have ${\sum_i i\cdot
\alpha_{ij}}=\alpha_{1j}+2\alpha_{2j}=1+\alpha_{2j}$. Using
Corollary \ref{relationss} we find $n_{j}=n_{1j}+n_{2j} \leq
\alpha_{1j} \cdot (\frac{N_{j}}{1+\alpha_{2j}}+4) +
\alpha_{2j}\cdot (\frac{N_{j}}{1+\alpha_{2j}}+2) \leq
\frac{N_{j}}{1+\alpha_{2j}}+4$. Thus, the total weight of items of
class $j$ (i.e., $W_j$) is at least $$\omega_{j} \cdot N_{j} =
\frac{\sum_{i} \alpha_{ij} \cdot r_{x,y}(i,j)}{\sum_i i\cdot
\alpha_{ij}} N_{j} \geq (n_{j}-4)(\sum_{i} \alpha_{ij} \cdot
r_{x,y}(i,j)) \geq n_{j}\cdot \sum r_{x,y}(i,j)\alpha_{ij} - 4 \ ,
$$ as $\sum_i r_{x,y}(i,j)\alpha_{ij} \leq \sum_i \alpha_{ij} =1$,
by $r_{x,y}(i,j) \leq 1$ for all $x,y,i$.
\end{proof}

\begin{lemma}
For any small class $j$, $W_j \geq  n_j\cdot \sum_{i=1}^{\gamma_j}
\alpha_{ij}\cdot r_{x,y}(i,j)-\gamma_j \cdot \zeta_j$.
\end{lemma}
\begin{proof}
For a small class $j$, we have $W_j \geq \omega_j
(\sum_{i=1}^{\gamma_j} i \cdot n_{ij} - (\gamma_j-1))$, as a class
$j$ type $i$ container has $i$ items, except for at most one
container of class $j$ and some type, which has at least one item
instead of $i \leq \gamma_j$ items of class $j$. Let
$n_{ij}=\alpha_{ij}n_j+\delta_{ij}$, for some (positive or
negative or zero) value $\delta_{ij}$. For $i$ such that
$\alpha_{ij}=0$ we have $\delta_{ij}=0$. For every $i$ such that
$\alpha_{ij}>0$, we have $\delta_{ij}\leq 1$. As
$$n_j=\sum_{i=1}^{\gamma_j} n_{ij}=\sum_{i=1}^{\gamma_j}
(\alpha_{ij}n_{j} + \delta_{ij})=n_j\sum_{i=1}^{\gamma_j}
\alpha_{ij} +
\sum_{i=1}^{\gamma_j}\delta_{ij}=n_j+\sum_{i=1}^{\gamma_j}\delta_{ij}
\ , $$ we get $\sum_{i=1}^{\gamma_j}\delta_{ij}=0$, which implies
$\sum_{i=1}^{\gamma_j} \min\{0,\delta_{ij}\}+\sum_{i=1}^{\gamma_j}
\max\{0,\delta_{ij}\}=0$. If for all $i$ we have $\delta_{ij}\geq
0$, then $\delta_{ij}=0$ for all $i$. Otherwise, if there is a
value $i$ such that $\delta_{ij}>0$, there is at least one
negative value as well. Thus, $ - \sum_{i=1}^{\gamma_j}
\min\{0,\delta_{ij}\} = \sum_{i=1}^{\gamma_j}
\max\{0,\delta_{ij}\} \leq \zeta_j-1$, and $$
\sum_{i=1}^{\gamma_j} i \cdot n_{ij} = \sum_{i=1}^{\gamma_j} i
\cdot( \alpha_{ij} n_j +\delta_{ij}) = n_j \sum_{i=1}^{\gamma_j} i
\cdot \alpha_{ij} +\sum_{i=1}^{\gamma_j} i \cdot \delta_{ij} \geq
n_j \sum_{i=1}^{\gamma_j} i \cdot \alpha_{ij}
+\sum_{i=1}^{\gamma_j} i \cdot \min\{0,\delta_{ij}\} \ . $$
Therefore, $W_j \geq \omega_j (\sum_{i=1}^{\gamma_j} i \cdot
n_{ij} - (\gamma_j-1)) \geq n_j \sum_{i=1}^{\gamma_j} \alpha_{ij}
\cdot r_{x,y}(i,j) - \gamma_j \cdot \zeta_j$, as $\omega_j \leq
1$, and $ \sum_{i=1}^{\gamma_j} i \cdot \min\{0,\delta_{ij}\} \geq
- \gamma_j(\zeta_j-1)$.
\end{proof}

%%\begin{lemma}
%%At termination of the algorithm we have that for any basic class
%%$j\geq 2$, the total required weight of the containers of class
%%$j$ is at most $n_j\cdot \sum_i r_{x,y}(i,j)\alpha_{ij}+\zeta_j$,
%%while the total weight of items satisfies $\omega_j \cdot N_j \geq
%%n_j\cdot \sum_i r_{x,y}(i,j)\alpha_{ij} - \zeta_j - 4$ for large
%%and small items and at least $\sum_i (A_{i,j}-t_M) \cdot
%%\alpha_{ij} $.
%%\end{lemma}
%%\begin{proof}
%%
%%
%%
%%$$\frac{\sum_{i} \alpha_{ij} \cdot r_{x,y}(i,j)}{\sum_i i\cdot
%%\alpha_{ij}} .$$
%%
%%
%%$$\frac{\sum_{i} \alpha_{ij} \cdot r_{x,y}(i,j)}{\sum_i
%%(A_{i,j}-t_M) \cdot \alpha_{ij}} .$$
%%
%% $$\rho=\frac{\sum_{i} \alpha_{i,M+1} \cdot r_{x,y}(i,M+1)}{\sum_i s(A_{i,M+1}-t_M) \cdot \alpha_{i,M+1}} .$$

%%
%%
%%\end{proof}
%%
%%Let $W$ denote the total weight of all input items.

\begin{corollary}\label{randw}
For any class $j$, we have ${\cal{R}}_j \leq W_j +\xi_j$, where
$\xi_j$ is a constant independent of the input such that for
$j=1$, $\xi_1=0$, for $j=M+1$, $\xi_{M+1}=3p$, for any small class
$j$, $\xi_j \leq (\gamma_j+1)\zeta_j$, and for any basic large
class $j$, $\xi_j \leq 6$.
\end{corollary}

\subsubsection{The relation between $\boldsymbol{W}$ and the cost of the algorithm}

We have proved the next theorem, which follows from Lemmas
\ref{consts1},\ref{consts2} and from Corollary \ref{randw}. The
theorem shows that our weight function is valid, and it remains to
find an upper bound on the supremum total weight of any bin (of
the optimal solution). We showed that there is a constant $\Psi$
that is independent of the input (and depends on our set of
parameters) such that the following holds.
\begin{theorem}
Assume that for input $I$, the output of the algorithm belongs to
the scenario of index $(x,y]$. If $y>\frac 13$ and $u,v,w$ satisfy
$0 \leq u,v,w \leq 1$ and the constraints (\ref{cons1}),
(\ref{cons2}), (\ref{cons3}), (\ref{cons4}), then assigning
weights to the items according to our definition in section
\ref{weight_sec} satisfies that the final number of bins of the
algorithm (applied on input $I$) is at most $W+\Psi$, where $\Psi$
is a constant independent of the input. If $y \leq \frac 13$ and
$w$ satisfies $0 \leq w \leq 1$, then assigning weights to the
items according to our definition in section \ref{weight_sec}
satisfies that the final number of bins of the algorithm (applied
on input $I$) is at most $W+\Psi$, where $\Psi$ is a constant
independent of the input.

\end{theorem}

%%\begin{proof}
%%If all large classes are basic, the proof follows from the
%%following properties:
%%
%%
%%
%%\end{proof}

%\newpage

\subsection{Analysis of weights of bins of optimal solutions} We
provide the remaining part of the proof, where given our sets of
weights (which are based on our set of parameters), we find upper
bounds on total weights of bins.

In Appendix \ref{allallall} we provide a table with all boundary
points and all strictly positive values of $\alpha_{ij}$ (all
other values of $\alpha_{ij}$ are equal to zero). It can be seen
that $\zeta_j \leq 3$ for all $j$. For classes 173, 174, 176, 184,
190, 191 indeed $\zeta_j=3$, and for classes 2 and 171 we have
$\zeta_j=1$. For all other classes $\zeta_j=2$ (for large classes
the case $\zeta_j=3$ is impossible as $\gamma_j=2$). For class
$M+1$ (the class of tiny items), we have $\zeta_{M+1}=2$, which is
an interesting feature of AH. There are containers where the total
size of tiny items is at most $\frac{17}{60} \approx 0.283333$
(and at least $\frac{17}{60}-\frac {1}{43}\approx 0.26$, except
for at most one container of tiny items of type 1). There is a
relatively big number of classes of large items whose sizes are in
$(\frac 13,0.35]$. The reason for this is that the volume of type
1 containers of these classes is defined by the exact size of an
item, while type 2 containers are defined by $2 \cdot t_{j-1}$,
though we still would like the volume to be close to the total
size and to $2 \cdot t_j$ (the difference $t_{j-1}-t_j$ is small,
much smaller than the size of a tiny item, and this limits the
possible bin kinds of optimal solutions).

To solve the knapsack problems, which are standard knapsack
problems, we use a branch and bound type approach similar to that
of Ramanan et al. \cite{RaBrLL89} and Seiden \cite{Seiden02J}.
Note that we could use existing solvers for knapsack, and we
actually did so (in addition to the branch and bound algorithm
described here) in order to verify the results (using rounded
values). However, as we were interested in precise results, we
represented all our parameters as big fractions with integer
numerators and denominators, with a common denominator of $q$ for
an appropriate value of $q$. Then, after representing every size
in the form $\frac{p_i}q$, we allow the total size of a multiset
of items of $\Delta$ to be at most $\frac{q-1}{q}$ in the case
where their total size should be below $1$ (and an analogous
condition is given in the case that the total size should be below
$1-a$). The branch and bound approach is standard as well, where
the branching rule is according to the size of the next item, and
the bounding rule is according to density, which is the ratio
between weight and size. That is, items are sorted by non
increasing density. Then, the algorithm iteratively generates the
possible packing patterns (multisets of items of $\Delta$) using
this sorted order. When a new item is added to the actual pattern,
an upper bound is calculated estimating the largest possible
weight of the patterns containing the given items. The bound is
based on upper bounding the weight of the remaining space in the
knapsack by assigning the weight of the current item to it. This
has the following meaning. If the maximum density of further items
that can still be packed into the bin is such that no matter what
additional items the bin will contain, its total weight is no
larger than the maximum weight of any bin calculated so far, there
is no need to compute an exact maximum (or even a supremum) of the
possible weight of a bin containing the items already inserted
into the bin, but it is sufficient to use the resulting upper
bound. If the estimated weight is lower than the current maximal
weight, no further items are added to this pattern. After all
possible patterns are generated and checked, the algorithm finds
the one with the largest weight. The pseudo-code of the algorithm
is given as
Algorithm \ref{BBB}.% in Appendix \ref{pseudc}.
%\newpage

%\section{A branch and bound algorithm for computing $\boldsymbol{R}$ \label{pseudc}}
\begin{algorithm}[h!]
\caption{Branch and Bound Knapsack Solver\label{BBB}} \bf Input:
$sizes[N],weights[N]$ \textmd{ sorted by} $weights[i]/sizes[i]$
\textmd{
in non increasing order} \\
Output: $worstbound,worstpat$ \\
Require: $actpat = \emptyset, worstbound = 0 $ \\
Procedure \textmd{KnapsackSolver}($i$) \\
If $i=N+1$ Then \\
\tab $totalweight \leftarrow \textmd{\textrm{The total weight of a
bin containing
the items of }} actpat$ \\
\tab \textmd{and filled with tiny items } \\
\tab If $totalweight > worstbound$ Then\\
\tab \tab $worstbound \leftarrow totalweight$ \\
\tab \tab $worstpat \leftarrow actpat$ \\
\tab End If \\
Else \\
\tab $es \leftarrow \textmd{\textrm{The empty space in a bin
containing the items of
}} actpat$ \\
\tab $actpw \leftarrow \textmd{\textrm{The total weight of the items of }} actpat$ \\
\tab $li \leftarrow \textmd{\textrm{The index of the last item of }} actpat$ \\
\tab If $actpat <> \emptyset$ Then \\
\tab\tab If $actpw + weights[li]/sizes[li] \cdot es < worstbound $ Then\\
\tab\tab\tab Return \\
\tab\tab End If \\
\tab End If \\
\tab $counter \leftarrow \textmd{\textrm{The maximal number of the
items of }}
sizes[i]$ \\
\tab \textmd{that can be packed into a bin containing the items of } $actpat$ \\
\tab For $j=counter \textrm{ to } 1 \textmd{\textrm{ step }} -1$ \\
\tab\tab Add $sizes[i] \textrm{ to } actpat$ \\
\tab\tab KnapsackSolver($i+1$) \\
\tab End For \\
\tab Remove all items of $sizes[i] \textrm{ from } actpat$ \\
\tab KnapsackSolver($i+1$) \\
End If \\
End Procedure
\end{algorithm}

Using this branch and bound procedure, for each scenario we
calculate the weight function corresponding to the values of $u$,
$v$, and $w$ (for scenarios where the threshold class is a large
class) or the value of $w$ (for the other scenarios). In Appendix
\ref{uvw_w} we report the values of $u$, $v$, and $w$ that we use,
and the resulting upper bound on the competitive ratio of the
algorithm. In this way, we prove that the competitive ratio of AH
is at most $1.57828956$.

\newpage

\appendix
\section{All parameters of the algorithm}\label{allallall}
We provide all required data for defining the algorithm and its
analysis according to our method of analysis. Recall that we use
exact values of parameters and exact calculations. In many cases
we write an approximate value in the table in order to provide
intuition, but these values were not used in our calculations.

The next table contains the values $\alpha_{ij}$ for all $j$ such
that $2 \leq j \leq 5$ and $j \geq 166$. The values $\alpha_{ij}$
are only given for $i$ such that $\alpha_{ij}\neq 0$.

For $6 \leq j \leq 165$, $\alpha_{1j}=\frac{22145926}{78181827}
\approx 0.2832618$, $\alpha_{2j}=\frac{56035901}{78181827} \approx
0.71673816$. For these values of $j$, $t_j=0.35-\frac{j-5}{9600}$
and $t_{j-1}=0.35-\frac{j-6}{9600}$. For class $166$, the right
endpoint is $0.35-\frac{160}{9600} = \frac 13$. There are many
boundary points between $\frac 13$ and $0.35$ as AH packs such
items carefully, and we would like very similar pairs of such
items to be packed together in one bin of the algorithm.

\begin{center}
\renewcommand{\arraystretch}{1.2}
\begin{longtable}{|c|c|c|c|c|}
\hline
Class & & & &\\
index: $j$ &  Left endpoint $t_j$ & Right endpoint $t_{j-1}$ & $i$
& $\alpha_{ij}$
or $\alpha_{i,j}$\\
\hline
1 & $\frac{1}{2} = 0.5$ & $1$ & 1 & \\
\hline
2 & $\frac{3}{7} \approx 0.42857$ & $\frac{1}{2} = 0.5$ & 2 & $1$\\
\hline 3 & $\frac{43}{120} \approx 0.35833$ & $\frac{3}{7} \approx
0.42857$ & 1 &
$\frac{31755722}{150095589} \approx 0.21156998824262585$\\
\hline
3 & & & 2 & $\frac{118339867}{150095589} \approx 0.7884300117573741$\\
\hline 4 & $\frac{59}{166} \approx 0.35542$ & $\frac{43}{120}
\approx 0.35833$ & 1 &
$\frac{33382666}{150909061} \approx 0.22121048119171585$\\
\hline
4 & & & 2 & $\frac{117526395}{150909061} \approx 0.7787895188082842$\\
\hline 5 & $\frac{7}{20} = 0.35$ & $\frac{59}{166} \approx
0.35542$ & 1 &
$\frac{4493270}{19023851} \approx 0.23619139994315558$\\
\hline
5 & & & 2 & $\frac{14530581}{19023851} \approx 0.7638086000568445$\\
\hline 166 & $\frac{271}{960} \approx 0.28229$ & $\frac{1}{3}
\approx 0.33333$ & 1 &
$\frac{3445801}{1433952966} \approx 0.0024030083843070765$\\
\hline
166 & & & 3 & $\frac{1430507165}{1433952966} \approx 0.9975969916156929$\\
\hline 167 & $\frac{1}{4} = 0.25$ & $\frac{271}{960} \approx
0.28229$ & 1 &
$\frac{18718929}{79588150} \approx 0.23519743831211054$\\
\hline
167 & & & 3 & $\frac{60869221}{79588150} \approx 0.7648025616878895$\\
\hline 168 & $\frac{97}{480} \approx 0.20208$ & $\frac{1}{4} =
0.25$ & 1 &
$\frac{10193524}{41199575} \approx 0.2474181833186386$\\
\hline
168 & & & 4 & $\frac{31006051}{41199575} \approx 0.7525818166813614$\\
\hline 169 & $\frac{1}{5} = 0.2$ & $\frac{97}{480} \approx
0.20208$ & 1 &
$\frac{22658284}{84102577} \approx 0.26941248185534195$\\
\hline
169 & & & 4 & $\frac{61444293}{84102577} \approx 0.730587518144658$\\
\hline
170 & $\frac{15}{88} \approx 0.17045$ & $\frac{1}{5} = 0.2$ & 5 & $1$\\
\hline 171 & $\frac{1}{6} \approx 0.16667$ & $\frac{15}{88}
\approx 0.17045$ & 1 &
$\frac{76872685}{1135239972} \approx 0.06771492098236301$\\
\hline
171 & & & 5 & $\frac{1058367287}{1135239972} \approx 0.932285079017637$\\
\hline 172 & $\frac{3}{20} = 0.15$ & $\frac{1}{6} \approx 0.16667$
& 2 &
$\frac{24797889}{191010959} \approx 0.12982443064955243$\\
\hline
172 & & & 1 & $\frac{48313566}{191010959} \approx 0.25293609462481154$\\
\hline
172 & & & 6 & $\frac{117899504}{191010959} \approx 0.617239474725636$\\
\hline 173 & $\frac{12}{83} \approx 0.14458$ & $\frac{3}{20} =
0.15$ & 2 &
$\frac{158075552}{480752651} \approx 0.3288084874231926$\\
\hline
173 & & & 1 & $\frac{20946010}{480752651} \approx 0.0435692033240603$\\
\hline
173 & & & 6 & $\frac{301731089}{480752651} \approx 0.6276223092527471$\\
\hline 174 & $\frac{1}{7} \approx 0.14286$ & $\frac{12}{83}
\approx 0.14458$ & 2 &
$\frac{5682641}{14973238} \approx 0.3795198473436407$\\
\hline
174 & & & 6 & $\frac{9290597}{14973238} \approx 0.6204801526563593$\\
\hline 175 & $\frac{11}{83} \approx 0.13253$ & $\frac{1}{7}
\approx 0.14286$ & 2 &
$\frac{11653567744}{42727973215} \approx 0.2727386034755545$\\
\hline
175 & & & 1 & $\frac{103268403}{85455946430} \approx 0.001208440223461697$\\
\hline
175 & & & 7 & $\frac{62045542539}{85455946430} \approx 0.7260529563009838$\\
\hline 176 & $\frac{1}{8} = 0.125$ & $\frac{11}{83} \approx
0.13253$ & 2 &
$\frac{4313813}{11469903} \approx 0.3760984726723495$\\
\hline
176 & & & 7 & $\frac{7156090}{11469903} \approx 0.6239015273276505$\\
\hline 177 & $\frac{1}{9} \approx 0.11111$ & $\frac{1}{8} = 0.125$
& 2 &
$\frac{35844844}{93992497} \approx 0.38135856737586193$\\
\hline
177 & & & 8 & $\frac{58147653}{93992497} \approx 0.6186414326241381$\\
\hline 178 & $\frac{1}{10} = 0.1$ & $\frac{1}{9} \approx 0.11111$
& 2 &
$\frac{145576935}{381661961} \approx 0.3814289865790424$\\
\hline
178 & & & 9 & $\frac{236085026}{381661961} \approx 0.6185710134209577$\\
\hline 179 & $\frac{1}{11} \approx 0.09091$ & $\frac{1}{10} = 0.1$
& 2 &
$\frac{8723245}{23755812} \approx 0.3672046655361644$\\
\hline
179 & & & 10 & $\frac{15032567}{23755812} \approx 0.6327953344638356$\\
\hline 180 & $\frac{1}{12} \approx 0.08333$ & $\frac{1}{11}
\approx 0.09091$ & 3 &
$\frac{145045373}{508140728} \approx 0.2854433132547486$\\
\hline
180 & & & 11 & $\frac{363095355}{508140728} \approx 0.7145566867452514$\\
\hline 181 & $\frac{1}{13} \approx 0.07692$ & $\frac{1}{12}
\approx 0.08333$ & 3 &
$\frac{16276212}{45761591} \approx 0.3556740848455203$\\
\hline
181 & & & 12 & $\frac{29485379}{45761591} \approx 0.6443259151544797$\\
\hline 182 & $\frac{1}{14} \approx 0.07143$ & $\frac{1}{13}
\approx 0.07692$ & 3 &
$\frac{72087509}{189669658} \approx 0.3800687456293088$\\
\hline
182 & & & 13 & $\frac{117582149}{189669658} \approx 0.6199312543706912$\\
\hline 183 & $\frac{1}{15} \approx 0.06667$ & $\frac{1}{14}
\approx 0.07143$ & 3 &
$\frac{36413948928}{97499546341} \approx 0.3734781370227504$\\
\hline
183 & & & 1 & $\frac{182118174}{97499546341} \approx 0.0018678873988095329$\\
\hline
183 & & & 14 & $\frac{60903479239}{97499546341} \approx 0.62465397557844$\\
\hline 184 & $\frac{1}{16} = 0.0625$ & $\frac{1}{15} \approx
0.06667$ & 4 &
$\frac{36527825}{116265557} \approx 0.3141758053074996$\\
\hline
184 & & & 15 & $\frac{79737732}{116265557} \approx 0.6858241946925003$\\
\hline 185 & $\frac{1}{17} \approx 0.05882$ & $\frac{1}{16} =
0.0625$ & 4 &
$\frac{30799804}{90208717} \approx 0.3414282457869343$\\
\hline
185 & & & 16 & $\frac{59408913}{90208717} \approx 0.6585717542130657$\\
\hline 186 & $\frac{1}{18} \approx 0.05556$ & $\frac{1}{17}
\approx 0.05882$ & 4 &
$\frac{61076393}{180923205} \approx 0.3375818651897085$\\
\hline
186 & & & 17 & $\frac{119846812}{180923205} \approx 0.6624181348102914$\\
\hline 187 & $\frac{1}{19} \approx 0.05263$ & $\frac{1}{18}
\approx 0.05556$ & 5 &
$\frac{246282282}{848959177} \approx 0.29009908682570257$\\
\hline
187 & & & 18 & $\frac{602676895}{848959177} \approx 0.7099009131742974$\\
\hline 188 & $\frac{1}{20} = 0.05$ & $\frac{1}{19} \approx
0.05263$ & 5 &
$\frac{3612237}{11491762} \approx 0.31433273679005885$\\
\hline
188 & & & 19 & $\frac{7879525}{11491762} \approx 0.6856672632099412$\\
\hline 189 & $\frac{1}{21} \approx 0.04762$ & $\frac{1}{20} =
0.05$ & 5 &
$\frac{11086689792}{34169004389} \approx 0.3244662813637359$\\
\hline
189 & & & 3 & $\frac{99039140}{34169004389} \approx 0.002898508217344594$\\
\hline
189 & & & 20 & $\frac{22983275457}{34169004389} \approx 0.6726352104189195$\\
\hline 190 & $\frac{1}{22} \approx 0.04545$ & $\frac{1}{21}
\approx 0.04762$ & 5 &
$\frac{1773973504}{5453794899} \approx 0.32527323393207785$\\
\hline
190 & & & 3 & $\frac{11750995}{4674681342} \approx 0.002513753160974282$\\
\hline
190 & & & 21 & $\frac{21996671405}{32722769394} \approx 0.6722130129069478$\\
\hline 191 & $\frac{1}{23} \approx 0.04348$ & $\frac{1}{22}
\approx 0.04545$ & 6 &
$\frac{72660709}{254170744} \approx 0.28587361336912953$\\
\hline
191 & & & 22 & $\frac{181510035}{254170744} \approx 0.7141263866308705$\\
\hline 192 & $\frac{1}{24} \approx 0.04167$ & $\frac{1}{23}
\approx 0.04348$ & 6 &
$\frac{17990899}{63629269} \approx 0.2827456496474916$\\
\hline
192 & & & 23 & $\frac{45638370}{63629269} \approx 0.7172543503525084$\\
\hline 193 & $\frac{1}{25} = 0.04$ & $\frac{1}{24} \approx
0.04167$ & 6 &
$\frac{6868668}{21928717} \approx 0.3132270802710437$\\
\hline
193 & & & 24 & $\frac{15060049}{21928717} \approx 0.6867729197289564$\\
\hline 194 & $\frac{1}{26} \approx 0.03846$ & $\frac{1}{25} =
0.04$ & 7 &
$\frac{6623739}{23559574} \approx 0.2811485046376475$\\
\hline
194 & & & 25 & $\frac{16935835}{23559574} \approx 0.7188514953623525$\\
\hline 195 & $\frac{1}{27} \approx 0.03704$ & $\frac{1}{26}
\approx 0.03846$ & 7 &
$\frac{20598370}{73772911} \approx 0.27921319249554893$\\
\hline
195 & & & 26 & $\frac{53174541}{73772911} \approx 0.720786807504451$\\
\hline 196 & $\frac{1}{28} \approx 0.03571$ & $\frac{1}{27}
\approx 0.03704$ & 7 &
$\frac{20611449}{73987996} \approx 0.2785782845098278$\\
\hline
196 & & & 27 & $\frac{53376547}{73987996} \approx 0.7214217154901722$\\
\hline 197 & $\frac{1}{29} \approx 0.03448$ & $\frac{1}{28}
\approx 0.03571$ & 7 &
$\frac{5843252}{21159655} \approx 0.27615062721958367$\\
\hline
197 & & & 28 & $\frac{15316403}{21159655} \approx 0.7238493727804163$\\
\hline 198 & $\frac{1}{30} \approx 0.03333$ & $\frac{1}{29}
\approx 0.03448$ & 8 &
$\frac{97422165}{338982541} \approx 0.2873958190076816$\\
\hline
198 & & & 29 & $\frac{241560376}{338982541} \approx 0.7126041809923184$\\
\hline 199 & $\frac{1}{31} \approx 0.03226$ & $\frac{1}{30}
\approx 0.03333$ & 8 &
$\frac{246577815}{717694643} \approx 0.34356925665390536$\\
\hline
199 & & & 30 & $\frac{471116828}{717694643} \approx 0.6564307433460946$\\
\hline 200 & $\frac{1}{32} = 0.03125$ & $\frac{1}{31} \approx
0.03226$ & 8 &
$\frac{136787965}{369923301} \approx 0.36977385482403013$\\
\hline
200 & & & 31 & $\frac{233135336}{369923301} \approx 0.6302261451759699$\\
\hline 201 & $\frac{1}{33} \approx 0.0303$ & $\frac{1}{32} =
0.03125$ & 9 &
$\frac{193885600}{743335051} \approx 0.2608320430190504$\\
\hline
201 & & & 32 & $\frac{549449451}{743335051} \approx 0.7391679569809496$\\
\hline 202 & $\frac{1}{34} \approx 0.02941$ & $\frac{1}{33}
\approx 0.0303$ & 9 &
$\frac{2009051}{7752584} \approx 0.25914598281037654$\\
\hline
202 & & & 33 & $\frac{5743533}{7752584} \approx 0.7408540171896235$\\
\hline 203 & $\frac{1}{35} \approx 0.02857$ & $\frac{1}{34}
\approx 0.02941$ & 9 &
$\frac{63841426}{248268817} \approx 0.25714637372280225$\\
\hline
203 & & & 34 & $\frac{184427391}{248268817} \approx 0.7428536262771978$\\
\hline 204 & $\frac{1}{36} \approx 0.02778$ & $\frac{1}{35}
\approx 0.02857$ & 9 &
$\frac{389848025}{1497560942} \approx 0.2603219769336105$\\
\hline
204 & & & 35 & $\frac{1107712917}{1497560942} \approx 0.7396780230663895$\\
\hline 205 & $\frac{1}{37} \approx 0.02703$ & $\frac{1}{36}
\approx 0.02778$ & 10 &
$\frac{99052686}{407082371} \approx 0.24332344767639177$\\
\hline
205 & & & 36 & $\frac{308029685}{407082371} \approx 0.7566765523236082$\\
\hline 206 & $\frac{1}{38} \approx 0.02632$ & $\frac{1}{37}
\approx 0.02703$ & 10 &
$\frac{407411107}{1639477277} \approx 0.24850061218628283$\\
\hline
206 & & & 37 & $\frac{1232066170}{1639477277} \approx 0.7514993878137172$\\
\hline 207 & $\frac{1}{39} \approx 0.02564$ & $\frac{1}{38}
\approx 0.02632$ & 10 &
$\frac{44138539}{166740862} \approx 0.26471339101029717$\\
\hline
207 & & & 38 & $\frac{122602323}{166740862} \approx 0.7352866089897029$\\
\hline 208 & $\frac{1}{40} = 0.025$ & $\frac{1}{39} \approx
0.02564$ & 11 &
$\frac{73429915}{298784748} \approx 0.2457619255719171$\\
\hline
208 & & & 39 & $\frac{225354833}{298784748} \approx 0.7542380744280829$\\
\hline 209 & $\frac{1}{41} \approx 0.02439$ & $\frac{1}{40} =
0.025$ & 11 &
$\frac{479473800}{1824013513} \approx 0.2628674604561441$\\
\hline
209 & & & 40 & $\frac{1344539713}{1824013513} \approx 0.7371325395438559$\\
\hline 210 & $\frac{1}{42} \approx 0.02381$ & $\frac{1}{41}
\approx 0.02439$ & 11 &
$\frac{478583529}{1826578078} \approx 0.26201098916287335$\\
\hline
210 & & & 41 & $\frac{1347994549}{1826578078} \approx 0.7379890108371266$\\
\hline 211 & $\frac{1}{43} \approx 0.02326$ & $\frac{1}{42}
\approx 0.02381$ & 11 &
$\frac{119627466}{457395215} \approx 0.26154070282523617$\\
\hline
211 & & & 42 & $\frac{337767749}{457395215} \approx 0.7384592971747639$\\
\hline 212 & $0$ & $\frac{1}{43} \approx 0.02326$ &
$A_{1,212}=\frac{17}{60}$ &
$\frac{13701867480}{32568497273} \approx 0.4207092321499018$\\
\hline
212 & & & $A_{2,212}=1$ & $\frac{18866629793}{32568497273} \approx 0.5792907678500983$\\
\hline
\end{longtable}
\end{center}
\section{The values $\boldsymbol{u}$, $\boldsymbol{v}$, $\boldsymbol{w}$, and the competitive ratio in
all scenarios}\label{uvw_w}

First, consider the case where the scenario is such that $k$ is
not a large class (it is small or tiny). For all scenarios whose
interval $(x,y]$ is contained in $(0,\frac 16]$, we use $w=0$, and
find $R < \frac{82081796062891}{52009705144320} \approx
1.57820153$. The scenarios contained in $(\frac 3{10},\frac 13]$
also have common features. Every scenario has the form $(\frac 3
{10}+\frac{\ell-1}{4800},\frac 3 {10}+\frac{\ell}{4800}]$ for $1
\leq \ell \leq 160$, $w=\frac{413913}{524288}\approx 0.7894764$,
and $$R < \frac{10060574276093395247}{6374352691333693440} \approx
1.57828956 \ . $$ The remaining cases are shown in the following
table. The values in the right column (RBB) are the bounds we
obtained on the total weight of a bin using the branch and bound
procedure.

%%82081796062891/52009705144320

%10060574276093395247/6374352691333693440
%413913/524288

\begin{center}
\renewcommand{\arraystretch}{1.3}
{\scriptsize
\begin{longtable}{|c|c|c|c|}
\hline
Threshold class & Scenario Interval & $w$ & RBB \\
\hline
$\left(\frac{1}{6},\frac{15}{88}\right]$ &$\left(\frac{1}{6},\frac{15}{88}\right]$ &$\frac{40165}{4194304} \approx 0.009576082229614258$ &$\frac{134279683919467}{85106790236160} \approx 1.5777787359487858$ \\
\hline
$\left(\frac{15}{88},\frac{1}{5}\right]$ &$\left(\frac{15}{88},\frac{23}{120}\right]$ &$\frac{40165}{4194304} \approx 0.009576082229614258$ &$\frac{134279683919467}{85106790236160} \approx 1.5777787359487858$ \\
\hline
$\left(\frac{15}{88},\frac{1}{5}\right]$ &$\left(\frac{23}{120},\frac{1}{5}\right]$ &$\frac{40165}{4194304} \approx 0.009576082229614258$ &$\frac{134279683919467}{85106790236160} \approx 1.5777787359487858$ \\
\hline
$\left(\frac{1}{5},\frac{97}{480}\right]$ &$\left(\frac{1}{5},\frac{97}{480}\right]$ &$\frac{2754177}{536870912} \approx 0.0051300544291734695$ &$\frac{134279683919467}{85106790236160} \approx 1.5777787359487858$ \\
\hline
$\left(\frac{97}{480},\frac{1}{4}\right]$ &$\left(\frac{97}{480},\frac{3}{14}\right]$ &$\frac{2754177}{536870912} \approx 0.0051300544291734695$ &$\frac{134279683919467}{85106790236160} \approx 1.5777787359487858$ \\
\hline
$\left(\frac{97}{480},\frac{1}{4}\right]$ &$\left(\frac{3}{14},\frac{2}{9}\right]$ &$\frac{2754177}{536870912} \approx 0.0051300544291734695$ &$\frac{134279683919467}{85106790236160} \approx 1.5777787359487858$ \\
\hline
$\left(\frac{97}{480},\frac{1}{4}\right]$ &$\left(\frac{2}{9},\frac{3}{13}\right]$ &$\frac{9224745}{1073741824} \approx 0.008591213263571262$ &$\frac{176162272658562716766643}{111689991334728680079360} \approx 1.5772431401719262$ \\
\hline
$\left(\frac{97}{480},\frac{1}{4}\right]$ &$\left(\frac{3}{13},\frac{4}{17}\right]$ &$\frac{9224745}{1073741824} \approx 0.008591213263571262$ &$\frac{176162272658562716766643}{111689991334728680079360} \approx 1.5772431401719262$ \\
\hline
$\left(\frac{97}{480},\frac{1}{4}\right]$ &$\left(\frac{4}{17},\frac{5}{21}\right]$ &$\frac{9224745}{1073741824} \approx 0.008591213263571262$ &$\frac{176162272658562716766643}{111689991334728680079360} \approx 1.5772431401719262$ \\
\hline
$\left(\frac{97}{480},\frac{1}{4}\right]$ &$\left(\frac{5}{21},\frac{1}{4}\right]$ &$\frac{9224745}{1073741824} \approx 0.008591213263571262$ &$\frac{176162272658562716766643}{111689991334728680079360} \approx 1.5772431401719262$ \\
\hline
$\left(\frac{1}{4},\frac{271}{960}\right]$ &$\left(\frac{1}{4},\frac{9}{35}\right]$ &$\frac{2179203}{16777216} \approx 0.12989062070846558$ &$\frac{10460110890923925809177}{6663323346674154209280} \approx 1.569803887146622$ \\
\hline
$\left(\frac{1}{4},\frac{271}{960}\right]$ &$\left(\frac{9}{35},\frac{8}{31}\right]$ &$\frac{2179203}{16777216} \approx 0.12989062070846558$ &$\frac{10460110890923925809177}{6663323346674154209280} \approx 1.569803887146622$ \\
\hline
$\left(\frac{1}{4},\frac{271}{960}\right]$ &$\left(\frac{8}{31},\frac{7}{27}\right]$ &$\frac{2179203}{16777216} \approx 0.12989062070846558$ &$\frac{10460110890923925809177}{6663323346674154209280} \approx 1.569803887146622$ \\
\hline
$\left(\frac{1}{4},\frac{271}{960}\right]$ &$\left(\frac{7}{27},\frac{6}{23}\right]$ &$\frac{2179203}{16777216} \approx 0.12989062070846558$ &$\frac{10460110890923925809177}{6663323346674154209280} \approx 1.569803887146622$ \\
\hline
$\left(\frac{1}{4},\frac{271}{960}\right]$ &$\left(\frac{6}{23},\frac{11}{42}\right]$ &$\frac{2179203}{16777216} \approx 0.12989062070846558$ &$\frac{10460110890923925809177}{6663323346674154209280} \approx 1.569803887146622$ \\
\hline
$\left(\frac{1}{4},\frac{271}{960}\right]$ &$\left(\frac{11}{42},\frac{5}{19}\right]$ &$\frac{2179203}{16777216} \approx 0.12989062070846558$ &$\frac{10460110890923925809177}{6663323346674154209280} \approx 1.569803887146622$ \\
\hline
$\left(\frac{1}{4},\frac{271}{960}\right]$ &$\left(\frac{5}{19},\frac{9}{34}\right]$ &$\frac{2179203}{16777216} \approx 0.12989062070846558$ &$\frac{10460110890923925809177}{6663323346674154209280} \approx 1.569803887146622$ \\
\hline
$\left(\frac{1}{4},\frac{271}{960}\right]$ &$\left(\frac{9}{34},\frac{22}{83}\right]$ &$\frac{2179203}{16777216} \approx 0.12989062070846558$ &$\frac{10460110890923925809177}{6663323346674154209280} \approx 1.569803887146622$ \\
\hline
$\left(\frac{1}{4},\frac{271}{960}\right]$ &$\left(\frac{22}{83},\frac{4}{15}\right]$ &$\frac{2750781}{16777216} \approx 0.1639593243598938$ &$\frac{10446431817507488684327}{6663323346674154209280} \approx 1.5677509966136622$ \\
\hline
$\left(\frac{1}{4},\frac{271}{960}\right]$ &$\left(\frac{4}{15},\frac{11}{41}\right]$ &$\frac{2750781}{16777216} \approx 0.1639593243598938$ &$\frac{10446431817507488684327}{6663323346674154209280} \approx 1.5677509966136622$ \\
\hline
$\left(\frac{1}{4},\frac{271}{960}\right]$ &$\left(\frac{11}{41},\frac{7}{26}\right]$ &$\frac{2750781}{16777216} \approx 0.1639593243598938$ &$\frac{10446431817507488684327}{6663323346674154209280} \approx 1.5677509966136622$ \\
\hline
$\left(\frac{1}{4},\frac{271}{960}\right]$ &$\left(\frac{7}{26},\frac{10}{37}\right]$ &$\frac{2750781}{16777216} \approx 0.1639593243598938$ &$\frac{10446431817507488684327}{6663323346674154209280} \approx 1.5677509966136622$ \\
\hline
$\left(\frac{1}{4},\frac{271}{960}\right]$ &$\left(\frac{10}{37},\frac{3}{11}\right]$ &$\frac{2750781}{16777216} \approx 0.1639593243598938$ &$\frac{10446431817507488684327}{6663323346674154209280} \approx 1.5677509966136622$ \\
\hline
$\left(\frac{1}{4},\frac{271}{960}\right]$ &$\left(\frac{3}{11},\frac{11}{40}\right]$ &$\frac{2750781}{16777216} \approx 0.1639593243598938$ &$\frac{10446431817507488684327}{6663323346674154209280} \approx 1.5677509966136622$ \\
\hline
$\left(\frac{1}{4},\frac{271}{960}\right]$ &$\left(\frac{11}{40},\frac{8}{29}\right]$ &$\frac{2750781}{16777216} \approx 0.1639593243598938$ &$\frac{10446431817507488684327}{6663323346674154209280} \approx 1.5677509966136622$ \\
\hline
$\left(\frac{1}{4},\frac{271}{960}\right]$ &$\left(\frac{8}{29},\frac{5}{18}\right]$ &$\frac{2750781}{16777216} \approx 0.1639593243598938$ &$\frac{10446431817507488684327}{6663323346674154209280} \approx 1.5677509966136622$ \\
\hline
$\left(\frac{1}{4},\frac{271}{960}\right]$ &$\left(\frac{5}{18},\frac{7}{25}\right]$ &$\frac{2750781}{16777216} \approx 0.1639593243598938$ &$\frac{10446431817507488684327}{6663323346674154209280} \approx 1.5677509966136622$ \\
\hline
$\left(\frac{1}{4},\frac{271}{960}\right]$ &$\left(\frac{7}{25},\frac{9}{32}\right]$ &$\frac{2750781}{16777216} \approx 0.1639593243598938$ &$\frac{10446431817507488684327}{6663323346674154209280} \approx 1.5677509966136622$ \\
\hline
$\left(\frac{1}{4},\frac{271}{960}\right]$ &$\left(\frac{9}{32},\frac{11}{39}\right]$ &$\frac{2750781}{16777216} \approx 0.1639593243598938$ &$\frac{10446431817507488684327}{6663323346674154209280} \approx 1.5677509966136622$ \\
\hline
$\left(\frac{1}{4},\frac{271}{960}\right]$ &$\left(\frac{11}{39},\frac{271}{960}\right]$ &$\frac{2750781}{16777216} \approx 0.1639593243598938$ &$\frac{10446431817507488684327}{6663323346674154209280} \approx 1.5677509966136622$ \\
\hline
$\left(\frac{271}{960},\frac{1}{3}\right]$ &$\left(\frac{271}{960},\frac{17}{60}\right]$ &$\frac{10600561}{134217728} \approx 0.07898033410310745$ &$\frac{1382826099045786640337}{888443112889887227904} \approx 1.5564599229632081$ \\
\hline
$\left(\frac{271}{960},\frac{1}{3}\right]$ &$\left(\frac{17}{60},\frac{2}{7}\right]$ &$\frac{13203731}{16777216} \approx 0.7870036959648132$ &$\frac{906785414291053674997}{576099831014536249344} \approx 1.574007429743845$ \\
\hline
$\left(\frac{271}{960},\frac{1}{3}\right]$ &$\left(\frac{2}{7},\frac{24}{83}\right]$ &$\frac{13203731}{16777216} \approx 0.7870036959648132$ &$\frac{906785414291053674997}{576099831014536249344} \approx 1.574007429743845$ \\
\hline
$\left(\frac{271}{960},\frac{1}{3}\right]$ &$\left(\frac{24}{83},\frac{3}{10}\right]$ &$\frac{6614407}{8388608} \approx 0.7884987592697144$ &$\frac{2517076902114799893609}{1596120743936377487360} \approx 1.5769965472080427$ \\
\hline
\end{longtable}
}
\end{center}

The case where $k$ is a large class is given in the next table.
Here, for each scenario, we present the values of $u$, $v$, and
$w$, and the resulting upper bound on the resulting upper bound on
the total weight of a bin as computed by our branch and bound
procedure.

\renewcommand{\arraystretch}{1.3}

\begin{center}
\begin{longtable}{|c|c|c|c|c|}
\hline
Scenario Interval & $u$ & $v$ & $w$ & RBB \\
\hline
$\left(\frac{1}{3},\frac{1067}{3200}\right]$ &$\frac{4849071}{8388608}$ &$\frac{1228277}{2097152}$ &$\frac{13637073}{16777216}$ &$\frac{976466097504059936537}{618737167905457176576} \approx 1.578159755311909$ \\
\hline
$\left(\frac{1067}{3200},\frac{1601}{4800}\right]$ &$\frac{4848391}{8388608}$ &$\frac{9827191}{16777216}$ &$\frac{13636371}{16777216}$ &$\frac{65094396453287453231117}{41249144527030478438400} \approx 1.578078701986928$ \\
\hline
$\left(\frac{1601}{4800},\frac{3203}{9600}\right]$ &$\frac{4847711}{8388608}$ &$\frac{4914083}{8388608}$ &$\frac{3408917}{4194304}$ &$\frac{2219013169161278425799}{1406220836148766310400} \approx 1.5779976459732428$ \\
\hline
$\left(\frac{3203}{9600},\frac{267}{800}\right]$ &$\frac{4847031}{8388608}$ &$\frac{2457285}{4194304}$ &$\frac{13634965}{16777216}$ &$\frac{195263112314549308140883}{123747433581091435315200} \approx 1.5779164598722266$ \\
\hline
$\left(\frac{267}{800},\frac{641}{1920}\right]$ &$\frac{9692701}{16777216}$ &$\frac{9830115}{16777216}$ &$\frac{6817131}{8388608}$ &$\frac{8135544964525835898911}{5156143065878809804800} \approx 1.5778353821024589$ \\
\hline
$\left(\frac{641}{1920},\frac{1603}{4800}\right]$ &$\frac{9691341}{16777216}$ &$\frac{9831091}{16777216}$ &$\frac{13633559}{16777216}$ &$\frac{3549873807061904759441}{2249953337838026096640} \approx 1.5777544126639393$ \\
\hline
$\left(\frac{1603}{4800},\frac{1069}{3200}\right]$ &$\frac{9689981}{16777216}$ &$\frac{4916033}{8388608}$ &$\frac{1704107}{2097152}$ &$\frac{1743151995805501132663}{1104887799831173529600} \approx 1.577673313138089$ \\
\hline
$\left(\frac{1069}{3200},\frac{401}{1200}\right]$ &$\frac{2422155}{4194304}$ &$\frac{9833041}{16777216}$ &$\frac{13632153}{16777216}$ &$\frac{21691442925099755308363}{13749714842343492812800} \approx 1.577592202734197$ \\
\hline
$\left(\frac{401}{1200},\frac{3209}{9600}\right]$ &$\frac{9687259}{16777216}$ &$\frac{307313}{524288}$ &$\frac{13631449}{16777216}$ &$\frac{938523792217307183471}{594939584524478054400} \approx 1.5775110895797064$ \\
\hline
$\left(\frac{3209}{9600},\frac{107}{320}\right]$ &$\frac{9685899}{16777216}$ &$\frac{9834991}{16777216}$ &$\frac{6815373}{8388608}$ &$\frac{97601454441832379372417}{61873716790545717657600} \approx 1.577429957411995$ \\
\hline
$\left(\frac{107}{320},\frac{3211}{9600}\right]$ &$\frac{4842269}{8388608}$ &$\frac{9835967}{16777216}$ &$\frac{13630043}{16777216}$ &$\frac{106662777263302497317}{67621548404967997440} \approx 1.577348933575532$ \\
\hline
$\left(\frac{3211}{9600},\frac{803}{2400}\right]$ &$\frac{9683177}{16777216}$ &$\frac{4918471}{8388608}$ &$\frac{3407335}{4194304}$ &$\frac{48795709950512261582249}{30936858395272858828800} \approx 1.5772677796517383$ \\
\hline
$\left(\frac{803}{2400},\frac{1071}{3200}\right]$ &$\frac{1210227}{2097152}$ &$\frac{9837917}{16777216}$ &$\frac{13628637}{16777216}$ &$\frac{799888507648052575393}{507161613037259980800} \approx 1.5771866148499032$ \\
\hline
$\left(\frac{1071}{3200},\frac{1607}{4800}\right]$ &$\frac{9680455}{16777216}$ &$\frac{9838893}{16777216}$ &$\frac{13627933}{16777216}$ &$\frac{844860460332803406563}{535703175675720499200} \approx 1.5771055664680742$ \\
\hline
$\left(\frac{1607}{4800},\frac{643}{1920}\right]$ &$\frac{4839547}{8388608}$ &$\frac{2459967}{4194304}$ &$\frac{6813615}{8388608}$ &$\frac{48788179926934140662201}{30936858395272858828800} \approx 1.5770243799024195$ \\
\hline
$\left(\frac{643}{1920},\frac{67}{200}\right]$ &$\frac{9677733}{16777216}$ &$\frac{2460211}{4194304}$ &$\frac{13626527}{16777216}$ &$\frac{39028537296861897028793}{24749486716218287063040} \approx 1.576943301668013$ \\
\hline
$\left(\frac{67}{200},\frac{3217}{9600}\right]$ &$\frac{2419093}{4194304}$ &$\frac{9841819}{16777216}$ &$\frac{13625823}{16777216}$ &$\frac{2710175530013886781293}{1718714355292936601600} \approx 1.5768621014118234$ \\
\hline
$\left(\frac{3217}{9600},\frac{1609}{4800}\right]$ &$\frac{9675011}{16777216}$ &$\frac{9842795}{16777216}$ &$\frac{425785}{524288}$ &$\frac{6097581320136123887941}{3867107299409107353600} \approx 1.5767810014135972$ \\
\hline
$\left(\frac{1609}{4800},\frac{1073}{3200}\right]$ &$\frac{9673649}{16777216}$ &$\frac{9843771}{16777216}$ &$\frac{425763}{524288}$ &$\frac{234510295646344141999}{148734896131119513600} \approx 1.5766998985874037$ \\
\hline
$\left(\frac{1073}{3200},\frac{161}{480}\right]$ &$\frac{302259}{524288}$ &$\frac{9844747}{16777216}$ &$\frac{13623713}{16777216}$ &$\frac{65034175789299211990013}{41249144527030478438400} \approx 1.5766187768253583$ \\
\hline
$\left(\frac{161}{480},\frac{3221}{9600}\right]$ &$\frac{4835463}{8388608}$ &$\frac{9845723}{16777216}$ &$\frac{13623009}{16777216}$ &$\frac{1393517774329773967033}{883910239864938823680} \approx 1.5765376522198715$ \\
\hline
$\left(\frac{3221}{9600},\frac{537}{1600}\right]$ &$\frac{9669565}{16777216}$ &$\frac{4923349}{8388608}$ &$\frac{6811153}{8388608}$ &$\frac{8867383288693055361043}{5624883344595065241600} \approx 1.576456389484717$ \\
\hline
$\left(\frac{537}{1600},\frac{3223}{9600}\right]$ &$\frac{9668203}{16777216}$ &$\frac{4923837}{8388608}$ &$\frac{6810801}{8388608}$ &$\frac{1477821141624138598811}{937480557432510873600} \approx 1.5763752430999367$ \\
\hline
$\left(\frac{3223}{9600},\frac{403}{1200}\right]$ &$\frac{9666841}{16777216}$ &$\frac{4924325}{8388608}$ &$\frac{6810449}{8388608}$ &$\frac{8866470349521856454129}{5624883344595065241600} \approx 1.5762940858216417$ \\
\hline
$\left(\frac{403}{1200},\frac{43}{128}\right]$ &$\frac{1208185}{2097152}$ &$\frac{4924813}{8388608}$ &$\frac{6810097}{8388608}$ &$\frac{34632866359469257597}{21972200564824473600} \approx 1.5762129176498314$ \\
\hline
$\left(\frac{43}{128},\frac{1613}{4800}\right]$ &$\frac{4832059}{8388608}$ &$\frac{4925301}{8388608}$ &$\frac{13619491}{16777216}$ &$\frac{866854473982499251647}{549988593693739712512} \approx 1.5761317305885907$ \\
\hline
$\left(\frac{1613}{4800},\frac{3227}{9600}\right]$ &$\frac{2415689}{4194304}$ &$\frac{9851579}{16777216}$ &$\frac{13618787}{16777216}$ &$\frac{7501239398070167193691}{4759516676195824435200} \approx 1.576050659846777$ \\
\hline
$\left(\frac{3227}{9600},\frac{269}{800}\right]$ &$\frac{4830697}{8388608}$ &$\frac{9852555}{16777216}$ &$\frac{13618083}{16777216}$ &$\frac{27860310850528313052091}{17678204797298776473600} \approx 1.5759694590021585$ \\
\hline
$\left(\frac{269}{800},\frac{3229}{9600}\right]$ &$\frac{75469}{131072}$ &$\frac{9853531}{16777216}$ &$\frac{13617379}{16777216}$ &$\frac{16251010517460625461749}{10312286131757619609600} \approx 1.575888247264025$ \\
\hline
$\left(\frac{3229}{9600},\frac{323}{960}\right]$ &$\frac{4829335}{8388608}$ &$\frac{9854507}{16777216}$ &$\frac{13616675}{16777216}$ &$\frac{195002075117312289565019}{123747433581091435315200} \approx 1.575807024632376$ \\
\hline
$\left(\frac{323}{960},\frac{1077}{3200}\right]$ &$\frac{9657307}{16777216}$ &$\frac{2463871}{4194304}$ &$\frac{13615971}{16777216}$ &$\frac{19499203742889618437113}{12374743358109143531520} \approx 1.575725910316502$ \\
\hline
$\left(\frac{1077}{3200},\frac{101}{300}\right]$ &$\frac{9654005}{16777216}$ &$\frac{4928925}{8388608}$ &$\frac{13611571}{16777216}$ &$\frac{7457779830355679220749}{4733508388347759820800} \approx 1.575529019598681$ \\
\hline
$\left(\frac{101}{300},\frac{3233}{9600}\right]$ &$\frac{2413161}{4194304}$ &$\frac{4929413}{8388608}$ &$\frac{13610869}{16777216}$ &$\frac{31015987293220839318841}{19687091706082728345600} \approx 1.5754479003944406$ \\
\hline
$\left(\frac{3233}{9600},\frac{539}{1600}\right]$ &$\frac{4825641}{8388608}$ &$\frac{4929901}{8388608}$ &$\frac{13610167}{16777216}$ &$\frac{1364633163438840555319321}{866232035067640047206400} \approx 1.5753667703276324$ \\
\hline
$\left(\frac{539}{1600},\frac{647}{1920}\right]$ &$\frac{9649921}{16777216}$ &$\frac{4930389}{8388608}$ &$\frac{13609465}{16777216}$ &$\frac{227427146094410028878003}{144372005844606674534400} \approx 1.575285629398257$ \\
\hline
$\left(\frac{647}{1920},\frac{809}{2400}\right]$ &$\frac{9648559}{16777216}$ &$\frac{4930877}{8388608}$ &$\frac{13608763}{16777216}$ &$\frac{272898516056915156775019}{173246407013528009441280} \approx 1.5752044776063134$ \\
\hline
$\left(\frac{809}{2400},\frac{1079}{3200}\right]$ &$\frac{9647197}{16777216}$ &$\frac{4931365}{8388608}$ &$\frac{13608061}{16777216}$ &$\frac{15504798574922583944779}{9843545853041364172800} \approx 1.5751233149518027$ \\
\hline
$\left(\frac{1079}{3200},\frac{1619}{4800}\right]$ &$\frac{2411459}{4194304}$ &$\frac{4931853}{8388608}$ &$\frac{13607359}{16777216}$ &$\frac{151594662165810553673821}{96248003896404449689600} \approx 1.5750421414347242$ \\
\hline
$\left(\frac{1619}{4800},\frac{3239}{9600}\right]$ &$\frac{4822237}{8388608}$ &$\frac{4932341}{8388608}$ &$\frac{53151}{65536}$ &$\frac{2664612582117044275993}{1691859443491484467200} \approx 1.5749609652076608$ \\
\hline
$\left(\frac{3239}{9600},\frac{27}{80}\right]$ &$\frac{1205389}{2097152}$ &$\frac{4932829}{8388608}$ &$\frac{6802977}{8388608}$ &$\frac{682105654058662162308929}{433116017533820023603200} \approx 1.5748797699577104$ \\
\hline
$\left(\frac{27}{80},\frac{3241}{9600}\right]$ &$\frac{4820875}{8388608}$ &$\frac{4933317}{8388608}$ &$\frac{3401313}{4194304}$ &$\frac{5683920686588408219107}{3609300146115166863360} \approx 1.5747985638451925$ \\
\hline
$\left(\frac{3241}{9600},\frac{1621}{4800}\right]$ &$\frac{2410097}{4194304}$ &$\frac{9867611}{16777216}$ &$\frac{6802275}{8388608}$ &$\frac{682035357649256549876341}{433116017533820023603200} \approx 1.5747174660793963$ \\
\hline
$\left(\frac{1621}{4800},\frac{1081}{3200}\right]$ &$\frac{4819513}{8388608}$ &$\frac{9868587}{16777216}$ &$\frac{13603847}{16777216}$ &$\frac{682000180089309563487227}{433116017533820023603200} \approx 1.5746362463633783$ \\
\hline
$\left(\frac{1081}{3200},\frac{811}{2400}\right]$ &$\frac{9637663}{16777216}$ &$\frac{9869563}{16777216}$ &$\frac{13603145}{16777216}$ &$\frac{41331211775978752332917}{26249455608110304460800} \approx 1.574555007655421$ \\
\hline
$\left(\frac{811}{2400},\frac{649}{1920}\right]$ &$\frac{9636301}{16777216}$ &$\frac{2467635}{4194304}$ &$\frac{6801221}{8388608}$ &$\frac{332973563943163219793}{211482430436435558400} \approx 1.5744738854003466$ \\
\hline
$\left(\frac{649}{1920},\frac{541}{1600}\right]$ &$\frac{9634939}{16777216}$ &$\frac{2467879}{4194304}$ &$\frac{3400435}{4194304}$ &$\frac{68189466375708315383477}{43311601753382002360320} \approx 1.5743926249595173$ \\
\hline
$\left(\frac{541}{1600},\frac{3247}{9600}\right]$ &$\frac{1204197}{2097152}$ &$\frac{9872493}{16777216}$ &$\frac{6800519}{8388608}$ &$\frac{37881084193626042951837}{24062000974101112422400} \approx 1.5743114728654097$ \\
\hline
$\left(\frac{3247}{9600},\frac{203}{600}\right]$ &$\frac{4816107}{8388608}$ &$\frac{9873469}{16777216}$ &$\frac{13600335}{16777216}$ &$\frac{1363648628756210055261719}{866232035067640047206400} \approx 1.5742301987823957$ \\
\hline
$\left(\frac{203}{600},\frac{1083}{3200}\right]$ &$\frac{9630851}{16777216}$ &$\frac{4937223}{8388608}$ &$\frac{849977}{1048576}$ &$\frac{21305911255966877993659}{13534875547931875737600} \approx 1.5741490330306298$ \\
\hline
$\left(\frac{1083}{3200},\frac{65}{192}\right]$ &$\frac{601843}{1048576}$ &$\frac{9875423}{16777216}$ &$\frac{6799465}{8388608}$ &$\frac{227251332599389259383933}{144372005844606674534400} \approx 1.574067848333346$ \\
\hline
$\left(\frac{65}{192},\frac{3251}{9600}\right]$ &$\frac{4814063}{8388608}$ &$\frac{9876399}{16777216}$ &$\frac{13598227}{16777216}$ &$\frac{2478977391127248024761}{1574967336486618267648} \approx 1.573986541623945$ \\
\hline
$\left(\frac{3251}{9600},\frac{271}{800}\right]$ &$\frac{9626763}{16777216}$ &$\frac{77167}{131072}$ &$\frac{3399381}{4194304}$ &$\frac{340841807120908724567593}{216558008766910011801600} \approx 1.573905343245792$ \\
\hline
$\left(\frac{271}{800},\frac{3253}{9600}\right]$ &$\frac{1203175}{2097152}$ &$\frac{9878353}{16777216}$ &$\frac{6798411}{8388608}$ &$\frac{28402018238670299772763}{18046500730575834316800} \approx 1.5738241259453316$ \\
\hline
$\left(\frac{3253}{9600},\frac{1627}{4800}\right]$ &$\frac{9624037}{16777216}$ &$\frac{4939665}{8388608}$ &$\frac{13596119}{16777216}$ &$\frac{1363226519980709081267741}{866232035067640047206400} \approx 1.573742905818833$ \\
\hline
$\left(\frac{1627}{4800},\frac{217}{640}\right]$ &$\frac{4811337}{8388608}$ &$\frac{9880307}{16777216}$ &$\frac{1699427}{2097152}$ &$\frac{7745205426604181330903}{4921772926520682086400} \approx 1.5736616748142933$ \\
\hline
$\left(\frac{217}{640},\frac{407}{1200}\right]$ &$\frac{9621311}{16777216}$ &$\frac{2470321}{4194304}$ &$\frac{13594713}{16777216}$ &$\frac{30290795128023441955693}{19249600779280889937920} \approx 1.573580432931712$ \\
\hline
$\left(\frac{407}{1200},\frac{3257}{9600}\right]$ &$\frac{2404987}{4194304}$ &$\frac{9882261}{16777216}$ &$\frac{6797005}{8388608}$ &$\frac{170376924627108250321189}{108279004383455005900800} \approx 1.5734991801710894$ \\
\hline
$\left(\frac{3257}{9600},\frac{543}{1600}\right]$ &$\frac{1202323}{2097152}$ &$\frac{4941619}{8388608}$ &$\frac{3398327}{4194304}$ &$\frac{340736249228769438239761}{216558008766910011801600} \approx 1.5734179085268438$ \\
\hline
$\left(\frac{543}{1600},\frac{3259}{9600}\right]$ &$\frac{2404215}{4194304}$ &$\frac{4942237}{8388608}$ &$\frac{6795959}{8388608}$ &$\frac{1428453243143512141799}{907925684877891993600} \approx 1.5733151588674645$ \\
\hline
$\left(\frac{3259}{9600},\frac{163}{480}\right]$ &$\frac{9615497}{16777216}$ &$\frac{9885451}{16777216}$ &$\frac{13591215}{16777216}$ &$\frac{638961729657449450353963}{406145423035377018470400} \approx 1.5732338552090321$ \\
\hline
$\left(\frac{163}{480},\frac{1087}{3200}\right]$ &$\frac{9614133}{16777216}$ &$\frac{2471607}{4194304}$ &$\frac{849407}{1048576}$ &$\frac{23959826404898791776701}{15230453363826638192640} \approx 1.5731525406725584$ \\
\hline
$\left(\frac{1087}{3200},\frac{1631}{4800}\right]$ &$\frac{9612769}{16777216}$ &$\frac{4943703}{8388608}$ &$\frac{13589809}{16777216}$ &$\frac{1916687167806180270563231}{1218436269106131055411200} \approx 1.5730713344673333$ \\
\hline
$\left(\frac{1631}{4800},\frac{3263}{9600}\right]$ &$\frac{4805703}{8388608}$ &$\frac{9888383}{16777216}$ &$\frac{13589105}{16777216}$ &$\frac{91266098781049354432547}{58020774719339574067200} \approx 1.5729900061232447$ \\
\hline
$\left(\frac{3263}{9600},\frac{17}{50}\right]$ &$\frac{4805021}{8388608}$ &$\frac{9889361}{16777216}$ &$\frac{6794201}{8388608}$ &$\frac{958244551649243177648759}{609218134553065527705600} \approx 1.5729087781541997$ \\
\hline
$\left(\frac{17}{50},\frac{653}{1920}\right]$ &$\frac{4804339}{8388608}$ &$\frac{4945169}{8388608}$ &$\frac{13587699}{16777216}$ &$\frac{1916389973691814198771709}{1218436269106131055411200} \approx 1.572827420097824$ \\
\hline
$\left(\frac{653}{1920},\frac{1633}{4800}\right]$ &$\frac{4803657}{8388608}$ &$\frac{2472829}{4194304}$ &$\frac{13586995}{16777216}$ &$\frac{42584244127472991954513}{27076361535691801231360} \approx 1.5727461782979537$ \\
\hline
$\left(\frac{1633}{4800},\frac{1089}{3200}\right]$ &$\frac{4802975}{8388608}$ &$\frac{9892293}{16777216}$ &$\frac{3396573}{4194304}$ &$\frac{43549814309358760573049}{27691733388775705804800} \approx 1.5726647984777584$ \\
\hline
$\left(\frac{1089}{3200},\frac{817}{2400}\right]$ &$\frac{4802293}{8388608}$ &$\frac{9893271}{16777216}$ &$\frac{3396397}{4194304}$ &$\frac{43547563979990343602953}{27691733388775705804800} \approx 1.5725835348985948$ \\
\hline
$\left(\frac{817}{2400},\frac{3269}{9600}\right]$ &$\frac{4801611}{8388608}$ &$\frac{9894249}{16777216}$ &$\frac{13584885}{16777216}$ &$\frac{638664592575390310874761}{406145423035377018470400} \approx 1.5725022525238697$ \\
\hline
$\left(\frac{3269}{9600},\frac{109}{320}\right]$ &$\frac{4800929}{8388608}$ &$\frac{9895227}{16777216}$ &$\frac{13584181}{16777216}$ &$\frac{24881749827237400000967}{15823847650728974745600} \approx 1.5724209671654128$ \\
\hline
$\left(\frac{109}{320},\frac{3271}{9600}\right]$ &$\frac{9600493}{16777216}$ &$\frac{9896205}{16777216}$ &$\frac{6791739}{8388608}$ &$\frac{191579567278604859795569}{121843626910613105541120} \approx 1.5723396630268682$ \\
\hline
$\left(\frac{3271}{9600},\frac{409}{1200}\right]$ &$\frac{9599129}{16777216}$ &$\frac{4948591}{8388608}$ &$\frac{6791387}{8388608}$ &$\frac{319282743328592446082311}{203072711517688509235200} \approx 1.5722582366798283$ \\
\hline
$\left(\frac{409}{1200},\frac{1091}{3200}\right]$ &$\frac{2399441}{4194304}$ &$\frac{618635}{1048576}$ &$\frac{6791035}{8388608}$ &$\frac{957798689566464461366261}{609218134553065527705600} \approx 1.5721769186485635$ \\
\hline
$\left(\frac{1091}{3200},\frac{1637}{4800}\right]$ &$\frac{599775}{1048576}$ &$\frac{4949569}{8388608}$ &$\frac{13581367}{16777216}$ &$\frac{1915498275440243753880953}{1218436269106131055411200} \approx 1.572095581860421$ \\
\hline
$\left(\frac{1637}{4800},\frac{131}{384}\right]$ &$\frac{9595035}{16777216}$ &$\frac{9900117}{16777216}$ &$\frac{13580663}{16777216}$ &$\frac{212822145923762135255601}{135381807678459006156800} \approx 1.5720143612591524$ \\
\hline
$\left(\frac{131}{384},\frac{273}{800}\right]$ &$\frac{9593671}{16777216}$ &$\frac{9901095}{16777216}$ &$\frac{13579959}{16777216}$ &$\frac{383060038533099768465509}{243687253821226211082240} \approx 1.5719330105550788$ \\
\hline
$\left(\frac{273}{800},\frac{3277}{9600}\right]$ &$\frac{4796153}{8388608}$ &$\frac{9902073}{16777216}$ &$\frac{1697407}{2097152}$ &$\frac{34200018735559994763271}{21757790519752340275200} \approx 1.5718516411173389$ \\
\hline
$\left(\frac{3277}{9600},\frac{1639}{4800}\right]$ &$\frac{9590941}{16777216}$ &$\frac{9903051}{16777216}$ &$\frac{1697319}{2097152}$ &$\frac{79795912583596606079363}{50768177879422127308800} \approx 1.571770268633972$ \\
\hline
$\left(\frac{1639}{4800},\frac{1093}{3200}\right]$ &$\frac{1198697}{2097152}$ &$\frac{9904029}{16777216}$ &$\frac{1697231}{2097152}$ &$\frac{239375342693527943715041}{152304533638266381926400} \approx 1.5716888852570905$ \\
\hline
$\left(\frac{1093}{3200},\frac{41}{120}\right]$ &$\frac{9588211}{16777216}$ &$\frac{619063}{1048576}$ &$\frac{1697143}{2097152}$ &$\frac{239362964133249590154553}{152304533638266381926400} \approx 1.5716076101959833$ \\
\hline
$\left(\frac{41}{120},\frac{3281}{9600}\right]$ &$\frac{9574953}{16777216}$ &$\frac{4957255}{8388608}$ &$\frac{13553857}{16777216}$ &$\frac{1495266333842777226977}{951903335239164887040} \approx 1.5708174123237975$ \\
\hline
$\left(\frac{3281}{9600},\frac{547}{1600}\right]$ &$\frac{9573585}{16777216}$ &$\frac{9915491}{16777216}$ &$\frac{3388287}{4194304}$ &$\frac{3737971784757563262733}{2379758338097912217600} \approx 1.5707358704939935$ \\
\hline
$\left(\frac{547}{1600},\frac{3283}{9600}\right]$ &$\frac{1196527}{2097152}$ &$\frac{1239559}{2097152}$ &$\frac{1694055}{2097152}$ &$\frac{155740737142497672059}{99156597420746342400} \approx 1.5706543103899644$ \\
\hline
$\left(\frac{3283}{9600},\frac{821}{2400}\right]$ &$\frac{299089}{524288}$ &$\frac{9917453}{16777216}$ &$\frac{13551731}{16777216}$ &$\frac{14950334357491487905151}{9519033352391648870400} \approx 1.57057274662613$ \\
\hline
$\left(\frac{821}{2400},\frac{219}{640}\right]$ &$\frac{9569479}{16777216}$ &$\frac{4959217}{8388608}$ &$\frac{6775511}{8388608}$ &$\frac{3737389461217747706737}{2379758338097912217600} \approx 1.5704911718914114$ \\
\hline
$\left(\frac{219}{640},\frac{1643}{4800}\right]$ &$\frac{4784055}{8388608}$ &$\frac{9919415}{16777216}$ &$\frac{13550313}{16777216}$ &$\frac{332195138395961895179}{211534074497592197120} \approx 1.5704095861858092$ \\
\hline
$\left(\frac{1643}{4800},\frac{3287}{9600}\right]$ &$\frac{9566741}{16777216}$ &$\frac{2480099}{4194304}$ &$\frac{3387401}{4194304}$ &$\frac{8341520371837818791}{5311960576111411200} \approx 1.5703279895093232$ \\
\hline
$\left(\frac{3287}{9600},\frac{137}{400}\right]$ &$\frac{2391343}{4194304}$ &$\frac{9921377}{16777216}$ &$\frac{13548895}{16777216}$ &$\frac{14947227680416248798199}{9519033352391648870400} \approx 1.5702463818619536$ \\
\hline
$\left(\frac{137}{400},\frac{3289}{9600}\right]$ &$\frac{9564003}{16777216}$ &$\frac{4961179}{8388608}$ &$\frac{6774093}{8388608}$ &$\frac{1245537562505576488103}{793252779365970739200} \approx 1.5701647632437$ \\
\hline
$\left(\frac{3289}{9600},\frac{329}{960}\right]$ &$\frac{4781317}{8388608}$ &$\frac{9923339}{16777216}$ &$\frac{13547477}{16777216}$ &$\frac{1358697610480488950053}{865366668399240806400} \approx 1.5700831336545629$ \\
\hline
$\left(\frac{329}{960},\frac{1097}{3200}\right]$ &$\frac{9561265}{16777216}$ &$\frac{9924321}{16777216}$ &$\frac{846673}{1048576}$ &$\frac{46702805346340105397}{29746979226223902720} \approx 1.5700016123038312$ \\
\hline
$\left(\frac{1097}{3200},\frac{823}{2400}\right]$ &$\frac{1194987}{2097152}$ &$\frac{4962651}{8388608}$ &$\frac{13546059}{16777216}$ &$\frac{4981373489060959129669}{3173011117463882956800} \approx 1.5699199607729266$ \\
\hline
$\left(\frac{823}{2400},\frac{3293}{9600}\right]$ &$\frac{9558527}{16777216}$ &$\frac{2481571}{4194304}$ &$\frac{6772675}{8388608}$ &$\frac{339621460315043539411}{216341667099810201600} \approx 1.5698384174804276$ \\
\hline
$\left(\frac{3293}{9600},\frac{549}{1600}\right]$ &$\frac{9557157}{16777216}$ &$\frac{9927265}{16777216}$ &$\frac{13544641}{16777216}$ &$\frac{2134652400193077494953}{1359861907484521267200} \approx 1.5697567440077553$ \\
\hline
$\left(\frac{549}{1600},\frac{659}{1920}\right]$ &$\frac{2388947}{4194304}$ &$\frac{9928247}{16777216}$ &$\frac{13543931}{16777216}$ &$\frac{18865897030054472569}{12018981505545011200} \approx 1.5696751859839877$ \\
\hline
$\left(\frac{659}{1920},\frac{103}{300}\right]$ &$\frac{4777209}{8388608}$ &$\frac{2482307}{4194304}$ &$\frac{6771611}{8388608}$ &$\frac{298820255727094090723}{190380667047832977408} \approx 1.5695934905618107$ \\
\hline
$\left(\frac{103}{300},\frac{1099}{3200}\right]$ &$\frac{9553049}{16777216}$ &$\frac{4965105}{8388608}$ &$\frac{13542513}{16777216}$ &$\frac{7470118077615630740381}{4759516676195824435200} \approx 1.56951190337804$ \\
\hline
$\left(\frac{1099}{3200},\frac{1649}{4800}\right]$ &$\frac{9551679}{16777216}$ &$\frac{9931191}{16777216}$ &$\frac{3385451}{4194304}$ &$\frac{1244954857076533900907}{793252779365970739200} \approx 1.5694301860140958$ \\
\hline
$\left(\frac{1649}{4800},\frac{3299}{9600}\right]$ &$\frac{4775155}{8388608}$ &$\frac{9932173}{16777216}$ &$\frac{6770547}{8388608}$ &$\frac{1867335189159116883697}{1189879169048956108800} \approx 1.5693485840681085$ \\
\hline
$\left(\frac{3299}{9600},\frac{11}{32}\right]$ &$\frac{2387235}{4194304}$ &$\frac{9933155}{16777216}$ &$\frac{13540385}{16777216}$ &$\frac{14937904568779219079413}{9519033352391648870400} \approx 1.5692669639639494$ \\
\hline
$\left(\frac{11}{32},\frac{3301}{9600}\right]$ &$\frac{2386613}{4194304}$ &$\frac{4967469}{8388608}$ &$\frac{13537557}{16777216}$ &$\frac{3808805900044886000665}{2427353504859870461952} \approx 1.5691187511086342$ \\
\hline
$\left(\frac{3301}{9600},\frac{1651}{4800}\right]$ &$\frac{9545083}{16777216}$ &$\frac{9935919}{16777216}$ &$\frac{13536849}{16777216}$ &$\frac{761721555357534853652867}{485470700971974092390400} \approx 1.569037130011083$ \\
\hline
$\left(\frac{1651}{4800},\frac{1101}{3200}\right]$ &$\frac{4771857}{8388608}$ &$\frac{2484225}{4194304}$ &$\frac{13536141}{16777216}$ &$\frac{380840962693780962134227}{242735350485987046195200} \approx 1.5689554979581215$ \\
\hline
$\left(\frac{1101}{3200},\frac{413}{1200}\right]$ &$\frac{9542345}{16777216}$ &$\frac{9937881}{16777216}$ &$\frac{6767717}{8388608}$ &$\frac{54403020465723070833259}{34676478640855292313600} \approx 1.5688738475776567$ \\
\hline
$\left(\frac{413}{1200},\frac{661}{1920}\right]$ &$\frac{596311}{1048576}$ &$\frac{9938863}{16777216}$ &$\frac{6767363}{8388608}$ &$\frac{8654576179425707744141}{5516712511045160140800} \approx 1.5687923128309016$ \\
\hline
$\left(\frac{661}{1920},\frac{551}{1600}\right]$ &$\frac{9539607}{16777216}$ &$\frac{2484961}{4194304}$ &$\frac{6767009}{8388608}$ &$\frac{76156305786765362781011}{48547070097197409239040} \approx 1.568710647919447$ \\
\hline
$\left(\frac{551}{1600},\frac{3307}{9600}\right]$ &$\frac{9538237}{16777216}$ &$\frac{4970413}{8388608}$ &$\frac{13533311}{16777216}$ &$\frac{380761730466132202387531}{242735350485987046195200} \approx 1.568629083912989$ \\
\hline
$\left(\frac{3307}{9600},\frac{827}{2400}\right]$ &$\frac{2384217}{4194304}$ &$\frac{9941807}{16777216}$ &$\frac{13532603}{16777216}$ &$\frac{761483804377150609715159}{485470700971974092390400} \approx 1.5685473970984514$ \\
\hline
$\left(\frac{827}{2400},\frac{1103}{3200}\right]$ &$\frac{4767749}{8388608}$ &$\frac{9942789}{16777216}$ &$\frac{13531895}{16777216}$ &$\frac{11897565630876931130309}{7585479702687095193600} \approx 1.5684658185377933$ \\
\hline
$\left(\frac{1103}{3200},\frac{331}{960}\right]$ &$\frac{9534129}{16777216}$ &$\frac{9943771}{16777216}$ &$\frac{13531187}{16777216}$ &$\frac{761404591056565992060113}{485470700971974092390400} \approx 1.5683842290217251$ \\
\hline
$\left(\frac{331}{960},\frac{3311}{9600}\right]$ &$\frac{9532759}{16777216}$ &$\frac{621547}{1048576}$ &$\frac{13530479}{16777216}$ &$\frac{10876641693512292118391}{6935295728171058462720} \approx 1.568302509340957$ \\
\hline
$\left(\frac{3311}{9600},\frac{69}{200}\right]$ &$\frac{9531389}{16777216}$ &$\frac{4972867}{8388608}$ &$\frac{13529771}{16777216}$ &$\frac{761325298589241682199147}{485470700971974092390400} \approx 1.568220897914069$ \\
\hline
$\left(\frac{69}{200},\frac{3313}{9600}\right]$ &$\frac{9530019}{16777216}$ &$\frac{2486679}{4194304}$ &$\frac{13529063}{16777216}$ &$\frac{190321418328523083301661}{121367675242993523097600} \approx 1.5681392755317705$ \\
\hline
$\left(\frac{3313}{9600},\frac{1657}{4800}\right]$ &$\frac{9528649}{16777216}$ &$\frac{4973849}{8388608}$ &$\frac{13528355}{16777216}$ &$\frac{1134494847571404472891}{723503280137070182400} \approx 1.5680576421940624$ \\
\hline
$\left(\frac{1657}{4800},\frac{221}{640}\right]$ &$\frac{595455}{1048576}$ &$\frac{1243585}{2097152}$ &$\frac{13527647}{16777216}$ &$\frac{380603203404100943254399}{242735350485987046195200} \approx 1.567975997900944$ \\
\hline
$\left(\frac{221}{640},\frac{829}{2400}\right]$ &$\frac{9525909}{16777216}$ &$\frac{4974831}{8388608}$ &$\frac{13526939}{16777216}$ &$\frac{13839395737772014341881}{8826740017672256225280} \approx 1.5678943426524157$ \\
\hline
$\left(\frac{829}{2400},\frac{3317}{9600}\right]$ &$\frac{9524539}{16777216}$ &$\frac{2487661}{4194304}$ &$\frac{13526231}{16777216}$ &$\frac{95140889878523638524229}{60683837621496761548800} \approx 1.5678126764484774$ \\
\hline
$\left(\frac{3317}{9600},\frac{553}{1600}\right]$ &$\frac{9523169}{16777216}$ &$\frac{4975813}{8388608}$ &$\frac{13525523}{16777216}$ &$\frac{108726781022912406382847}{69352957281710584627200} \approx 1.567730999289129$ \\
\hline
$\left(\frac{553}{1600},\frac{3319}{9600}\right]$ &$\frac{9521799}{16777216}$ &$\frac{311019}{524288}$ &$\frac{6762407}{8388608}$ &$\frac{190261953374162263057889}{121367675242993523097600} \approx 1.567649318430411$ \\
\hline
$\left(\frac{3319}{9600},\frac{83}{240}\right]$ &$\frac{9520429}{16777216}$ &$\frac{4976795}{8388608}$ &$\frac{6762053}{8388608}$ &$\frac{380504075494014799243069}{242735350485987046195200} \approx 1.5675676193525057$ \\
\hline
$\left(\frac{83}{240},\frac{1107}{3200}\right]$ &$\frac{4759529}{8388608}$ &$\frac{2488643}{4194304}$ &$\frac{6761699}{8388608}$ &$\frac{19024212079021989045911}{12136767524299352309760} \approx 1.5674859093191904$ \\
\hline
$\left(\frac{1107}{3200},\frac{1661}{4800}\right]$ &$\frac{1189711}{2097152}$ &$\frac{9955555}{16777216}$ &$\frac{13522689}{16777216}$ &$\frac{760928871399143283572207}{485470700971974092390400} \approx 1.5674043147725845$ \\
\hline
$\left(\frac{1661}{4800},\frac{3323}{9600}\right]$ &$\frac{9516317}{16777216}$ &$\frac{9956537}{16777216}$ &$\frac{13521981}{16777216}$ &$\frac{380444596465588028345497}{242735350485987046195200} \approx 1.5673225828207122$ \\
\hline
$\left(\frac{3323}{9600},\frac{277}{800}\right]$ &$\frac{4757473}{8388608}$ &$\frac{9957519}{16777216}$ &$\frac{13521273}{16777216}$ &$\frac{760849509144678246905501}{485470700971974092390400} \approx 1.5672408399134297$ \\
\hline
$\left(\frac{277}{800},\frac{133}{384}\right]$ &$\frac{1189197}{2097152}$ &$\frac{4979251}{8388608}$ &$\frac{3380141}{4194304}$ &$\frac{1698236342438227951421}{1083639957526727884800} \approx 1.5671592124696463$ \\
\hline
$\left(\frac{133}{384},\frac{1663}{4800}\right]$ &$\frac{9512205}{16777216}$ &$\frac{2489871}{4194304}$ &$\frac{844991}{1048576}$ &$\frac{1901925467462527805549}{1213676752429935230976} \approx 1.5670774476438072$ \\
\hline
$\left(\frac{1663}{4800},\frac{1109}{3200}\right]$ &$\frac{4755417}{8388608}$ &$\frac{9960467}{16777216}$ &$\frac{13519147}{16777216}$ &$\frac{34578661300196807203331}{22066850044180640563200} \approx 1.566995798265993$ \\
\hline
$\left(\frac{1109}{3200},\frac{26}{75}\right]$ &$\frac{9509463}{16777216}$ &$\frac{9961449}{16777216}$ &$\frac{13518439}{16777216}$ &$\frac{760690843536197568224759}{485470700971974092390400} \approx 1.566914011521597$ \\
\hline
$\left(\frac{26}{75},\frac{3329}{9600}\right]$ &$\frac{2377023}{4194304}$ &$\frac{155663}{262144}$ &$\frac{6758865}{8388608}$ &$\frac{95081399313398370675127}{60683837621496761548800} \approx 1.5668323402097522$ \\
\hline
$\left(\frac{3329}{9600},\frac{111}{320}\right]$ &$\frac{9506721}{16777216}$ &$\frac{9963415}{16777216}$ &$\frac{13517021}{16777216}$ &$\frac{760611540152133730827083}{485470700971974092390400} \approx 1.5667506579270236$ \\
\hline
$\left(\frac{111}{320},\frac{3331}{9600}\right]$ &$\frac{4752675}{8388608}$ &$\frac{9964397}{16777216}$ &$\frac{13516313}{16777216}$ &$\frac{76057181912089774301981}{48547070097197409239040} \approx 1.5666688383009235$ \\
\hline
$\left(\frac{3331}{9600},\frac{833}{2400}\right]$ &$\frac{9503979}{16777216}$ &$\frac{2491345}{4194304}$ &$\frac{3378901}{4194304}$ &$\frac{27161862647054116763159}{17338239320427646156800} \approx 1.5665871340841644$ \\
\hline
$\left(\frac{833}{2400},\frac{1111}{3200}\right]$ &$\frac{9502607}{16777216}$ &$\frac{9966363}{16777216}$ &$\frac{13514895}{16777216}$ &$\frac{95061560473511270495863}{60683837621496761548800} \approx 1.5665054188965213$ \\
\hline
$\left(\frac{1111}{3200},\frac{1667}{4800}\right]$ &$\frac{2375309}{4194304}$ &$\frac{4983673}{8388608}$ &$\frac{6757093}{8388608}$ &$\frac{380226404066311213407583}{242735350485987046195200} \approx 1.5664236927379946$ \\
\hline
$\left(\frac{1667}{4800},\frac{667}{1920}\right]$ &$\frac{9493465}{16777216}$ &$\frac{2493229}{4194304}$ &$\frac{13501373}{16777216}$ &$\frac{141611096872472700200291}{90430816847720664268800} \approx 1.5659606073329642$ \\
\hline
$\left(\frac{667}{1920},\frac{139}{400}\right]$ &$\frac{2373023}{4194304}$ &$\frac{2493475}{4194304}$ &$\frac{13500661}{16777216}$ &$\frac{18880492143682657223507}{12057442246362755235840} \approx 1.5658787127409333$ \\
\hline
$\left(\frac{139}{400},\frac{3337}{9600}\right]$ &$\frac{4745359}{8388608}$ &$\frac{2493721}{4194304}$ &$\frac{6749975}{8388608}$ &$\frac{35399070916863617830051}{22607704211930166067200} \approx 1.5657968002864882$ \\
\hline
$\left(\frac{3337}{9600},\frac{1669}{4800}\right]$ &$\frac{9489345}{16777216}$ &$\frac{9975869}{16777216}$ &$\frac{6749619}{8388608}$ &$\frac{12871716969171545606503}{8220983349792787660800} \approx 1.5657150028768738$ \\
\hline
$\left(\frac{1669}{4800},\frac{1113}{3200}\right]$ &$\frac{9487971}{16777216}$ &$\frac{9976853}{16777216}$ &$\frac{13498527}{16777216}$ &$\frac{257889758215323713103}{164719156371075891200} \approx 1.5656330684110293$ \\
\hline
$\left(\frac{1113}{3200},\frac{167}{480}\right]$ &$\frac{4743299}{8388608}$ &$\frac{4988919}{8388608}$ &$\frac{13497815}{16777216}$ &$\frac{283148156523474240613187}{180861633695441328537600} \approx 1.5655512489745418$ \\
\hline
$\left(\frac{167}{480},\frac{3341}{9600}\right]$ &$\frac{1185653}{2097152}$ &$\frac{9978823}{16777216}$ &$\frac{13497103}{16777216}$ &$\frac{12639881988121217147}{8074180075689345024} \approx 1.5654694185207492$ \\
\hline
$\left(\frac{3341}{9600},\frac{557}{1600}\right]$ &$\frac{4741925}{8388608}$ &$\frac{9979807}{16777216}$ &$\frac{13496391}{16777216}$ &$\frac{8579349484586716593901}{5480655566528525107200} \approx 1.565387457840362$ \\
\hline
$\left(\frac{557}{1600},\frac{3343}{9600}\right]$ &$\frac{2370619}{4194304}$ &$\frac{1247599}{2097152}$ &$\frac{105435}{131072}$ &$\frac{1105873936667833345961}{706490756622817689600} \approx 1.5653055985532716$ \\
\hline
$\left(\frac{3343}{9600},\frac{209}{600}\right]$ &$\frac{4740551}{8388608}$ &$\frac{9981777}{16777216}$ &$\frac{1686871}{2097152}$ &$\frac{35386115227622790352823}{22607704211930166067200} \approx 1.5652237350553009$ \\
\hline
$\left(\frac{209}{600},\frac{223}{640}\right]$ &$\frac{592483}{1048576}$ &$\frac{4991381}{8388608}$ &$\frac{843391}{1048576}$ &$\frac{1474344343033289213191}{941987675497090252800} \approx 1.565141860540025$ \\
\hline
$\left(\frac{223}{640},\frac{1673}{4800}\right]$ &$\frac{4739177}{8388608}$ &$\frac{4991873}{8388608}$ &$\frac{1686693}{2097152}$ &$\frac{7076482058770152376727}{4521540842386033213440} \approx 1.5650598557981548$ \\
\hline
$\left(\frac{1673}{4800},\frac{3347}{9600}\right]$ &$\frac{2369245}{4194304}$ &$\frac{9984731}{16777216}$ &$\frac{421651}{524288}$ &$\frac{4422569850109370951347}{2825963026491270758400} \approx 1.564977959248269$ \\
\hline
$\left(\frac{3347}{9600},\frac{279}{800}\right]$ &$\frac{4737803}{8388608}$ &$\frac{2496429}{4194304}$ &$\frac{1686515}{2097152}$ &$\frac{561566778711479378089}{358852447808415334400} \approx 1.5648960516810781$ \\
\hline
$\left(\frac{279}{800},\frac{3349}{9600}\right]$ &$\frac{4736809}{8388608}$ &$\frac{9987141}{16777216}$ &$\frac{6745125}{8388608}$ &$\frac{582566669044251570404083}{372299971115762266931200} \approx 1.5647776369639024$ \\
\hline
$\left(\frac{3349}{9600},\frac{67}{192}\right]$ &$\frac{9472245}{16777216}$ &$\frac{9988125}{16777216}$ &$\frac{3372385}{4194304}$ &$\frac{2621412859662547736738383}{1675349870020930201190400} \approx 1.5646957728476132$ \\
\hline
$\left(\frac{67}{192},\frac{1117}{3200}\right]$ &$\frac{1183859}{2097152}$ &$\frac{9989109}{16777216}$ &$\frac{6744415}{8388608}$ &$\frac{149787182298284044767131}{95734278286910297210880} \approx 1.5646138977449666$ \\
\hline
$\left(\frac{1117}{3200},\frac{419}{1200}\right]$ &$\frac{4734749}{8388608}$ &$\frac{4995047}{8388608}$ &$\frac{13488121}{16777216}$ &$\frac{3494851587101814851987261}{2233799826694573601587200} \approx 1.5645321238444452$ \\
\hline
$\left(\frac{419}{1200},\frac{3353}{9600}\right]$ &$\frac{9468125}{16777216}$ &$\frac{4995539}{8388608}$ &$\frac{13487411}{16777216}$ &$\frac{10484005936339043623382941}{6701399480083720804761600} \approx 1.5644502267768203$ \\
\hline
$\left(\frac{3353}{9600},\frac{559}{1600}\right]$ &$\frac{9466751}{16777216}$ &$\frac{4996031}{8388608}$ &$\frac{13486701}{16777216}$ &$\frac{10483457037748672197869699}{6701399480083720804761600} \approx 1.5643683187228383$ \\
\hline
$\left(\frac{559}{1600},\frac{671}{1920}\right]$ &$\frac{4732689}{8388608}$ &$\frac{9993047}{16777216}$ &$\frac{13485991}{16777216}$ &$\frac{1164767651600377921797473}{744599942231524533862400} \approx 1.5642865188917878$ \\
\hline
$\left(\frac{671}{1920},\frac{839}{2400}\right]$ &$\frac{2366001}{4194304}$ &$\frac{9994031}{16777216}$ &$\frac{6742641}{8388608}$ &$\frac{1048235977172324906266867}{670139948008372080476160} \approx 1.5642045818752313$ \\
\hline
$\left(\frac{839}{2400},\frac{1119}{3200}\right]$ &$\frac{9462631}{16777216}$ &$\frac{1249377}{2097152}$ &$\frac{3371143}{4194304}$ &$\frac{2620452862795471260412937}{1675349870020930201190400} \approx 1.564122760079203$ \\
\hline
$\left(\frac{1119}{3200},\frac{1679}{4800}\right]$ &$\frac{9461257}{16777216}$ &$\frac{9996001}{16777216}$ &$\frac{6741931}{8388608}$ &$\frac{249553882309917869707153}{159557130478183828684800} \approx 1.5640409272968172$ \\
\hline
$\left(\frac{1679}{4800},\frac{3359}{9600}\right]$ &$\frac{9459883}{16777216}$ &$\frac{9996985}{16777216}$ &$\frac{842697}{1048576}$ &$\frac{655044611897385906503719}{418837467505232550297600} \approx 1.5639589643187841$ \\
\hline
$\left(\frac{3359}{9600},\frac{7}{20}\right]$ &$\frac{9458509}{16777216}$ &$\frac{4998985}{8388608}$ &$\frac{6741221}{8388608}$ &$\frac{5240082624472449502119691}{3350699740041860402380800} \approx 1.5638771095636832$ \\
\hline
$\left(\frac{7}{20},\frac{59}{166}\right]$ &$\frac{9075501}{16777216}$ &$\frac{2461325}{4194304}$ &$\frac{13365787}{16777216}$ &$\frac{690817828427621169838459}{446759965338914720317440} \approx 1.5462840944209553$ \\
\hline
$\left(\frac{59}{166},\frac{43}{120}\right]$ &$\frac{4456807}{8388608}$ &$\frac{9835417}{16777216}$ &$\frac{13254243}{16777216}$ &$\frac{58585852828075806845873}{38187188798677831385088} \approx 1.5341755879685033$ \\
\hline
$\left(\frac{43}{120},\frac{3}{7}\right]$ &$\frac{1045721}{2097152}$ &$\frac{636339}{1048576}$ &$\frac{6779541}{8388608}$ &$\frac{447130049349370498541651}{284301796124763912929280} \approx 1.5727303008425262$ \\
\hline
$\left(\frac{3}{7},\frac{1}{2}\right]$ &$\frac{8388625}{16777216}$ &$\frac{8388625}{16777216}$ &$1$ &$\frac{470739775199623}{298261319516160} \approx 1.5782796641658323$ \\
\hline
\end{longtable}
\end{center}

Using this table, we completed the proof of  the following
theorem.
\begin{theorem}
The competitive ratio of AH is at most $1.57828956$.
\end{theorem}

\section{Several weight functions and examples}
The weight functions are calculated easily using our definitions
of weights, which are based on the parameters in the previous two
sections and on our formulas for weights. Recall that weights of
items in all scenarios are based directly on values of the form
$\alpha_{ij}$, $u$, $v$, and $w$, and are defined precisely in
Section \ref{amor}. Here, we provide three examples of such weight
functions, while all 412 weight functions can be downloaded from
http://math.haifa.ac.il/lea/WeightFunctionsForBP.pdf.

We note that since the large classes in the interval $(\frac
13,\frac 7{20}]$ use common values of $\alpha_{1j}$ (and
$\alpha_{2j}$), for all large classes $j$, the number of distinct
weights in this interval is at most four. Thus, a table reporting
the weight function of a given scenario is sufficiently short.

\subsection{Weight function for the scenario interval $\boldsymbol{(17/50,653/1920]}$}

In this section we discuss the threshold class
${(17/50=0.34,653/1920\approx 0.340104]}$.

The large threshold classes that are close to $\frac 13$ are
significant ones in the analysis. In the threshold class discussed
here we have $a \approx 0.66$, that is, there can be positive bins
whose single huge item has size below $\frac 23$, and there may be
negative bins whose single item has size slightly above $0.34$. If
one chooses parameters in such a way that $\alpha_{1k}$ is too
big, there may be a big fraction of bins that are relatively
empty.

The weights table contains the weights assigned to classes in the
case of a given threshold class, based on the $\alpha_{ij}$ values
and $u$, $v$, $w$ (if $u$ and $v$ are undefined, obviously weights
are not based on them). The set $\Delta$ as it is defined above is
based on these weights. Note that in the case where $u$ and $v$
are defined, the threshold class is split if necessary into
$(t_k,1-a]$ and $(1-a,t_{k-1}]$.  While the value of $a$ is
unknown, as explained above, we overcome this difficulty by
considering bins with an item of size at least $a$ and other items
of sizes no larger than $1-a$, and bins without such an item (in
which case, the value $1-a$ can be only slightly larger than $x$,
but it can also be $y$).

For solutions of the knapsack problem, recall that in the set
$\Delta$, the size of every item is of the form $t_{\ell}$ or $x$,
where its weight in $\Delta$ is the weight of an item that is just
slightly larger (an item of the following interval in the table).

Consider solutions of the knapsack problems on $\Delta$. We
provide several examples in order to allow the reader to see how
total weights are computed and to provide some intuition to the
knapsack problems.

Recall that tiny items are added to each solution such that the
total size becomes $1$. The first set of items from $\Delta$ is as
follows:
$$\frac {17}{50}, \  \frac {17}{50}, \ \frac 14,  \ \frac 1{15},
\mbox{ \ (and tiny items of total size \ } \frac 1{300}) \ .
$$ The total weight of these items is $2\cdot
\frac{4945169}{8388608}+\frac{1080920410736029}{3377699720527872}+\frac{569257302249594179}{8070450532247928832}+
\frac{1936246260875168533}{1983131948414926848}\cdot \frac{1}{300}
\approx 1.572827420097824 $.

The second set of items from $\Delta$ contains $\frac {17}{50}$,
an item of size $a$, and tiny items of total size $\frac 1{9600}$.
The total weight of these items is $1+\frac{4804339}{8388608} +
\frac{1936246260875168533}{1983131948414926848}\cdot
\frac{1}{9600} \approx 1.572823543$.

The third set of items from $\Delta$ is as follows:
$$\frac {1}{2}, \  \frac {17}{50}, \ \frac 3{20},
\mbox{ \ (and tiny items of total size \ } \frac 1{100}) \ .
$$ The total weight of these items is $\frac{13587699}{16777216}+
\frac{4945169}{8388608}+\frac{2211215673518099}{13510798882111488}+
\frac{1936246260875168533}{1983131948414926848}\cdot \frac{1}{100}
\approx 1.57282647$.

It can be seen that those are all knapsack solutions with total
weight very close to the maximum.

%\newpage

\renewcommand{\arraystretch}{1.2}
\begin{center}
\begin{longtable}{|  c   l   c |}
\hline
Interval & & Weight \\
\hline
$\left(0,\frac{1}{43}\right]$ & & $\rho=\frac{1936246260875168533}{1983131948414926848} \approx 0.9763577569423795$\\
\hline
$\left(\frac{1}{43},\frac{1}{42}\right]$ & &$\frac{3024178159640353}{130041439240323072} \approx 0.023255495919662354$\\
\hline
$\left(\frac{1}{42},\frac{1}{41}\right]$ & &$\frac{24142046017061077}{1015561715972046848} \approx 0.023772111174901342$\\
\hline
$\left(\frac{1}{41},\frac{1}{40}\right]$ & &$\frac{3010865377716701}{123848989752688640} \approx 0.024310778664638547$\\
\hline
$\left(\frac{1}{40},\frac{1}{39}\right]$ & &$\frac{4015032672085983}{161003686678495232} \approx 0.02493752009606783$\\
\hline
$\left(\frac{1}{39},\frac{1}{38}\right]$ & &$\frac{2197706275568431}{85568392920039424} \approx 0.02568362219472917$\\
\hline
$\left(\frac{1}{38},\frac{1}{37}\right]$ & &$\frac{21970084912148039}{833165931063541760} \approx 0.026369399051282716$\\
\hline
$\left(\frac{1}{37},\frac{1}{36}\right]$ & &$\frac{2741905392774811}{101330991615836160} \approx 0.027058902208022033$\\
\hline
$\left(\frac{1}{36},\frac{1}{35}\right]$ & &$\frac{19827765777652997}{709316941310853120} \approx 0.027953323292984226$\\
\hline
$\left(\frac{1}{35},\frac{1}{34}\right]$ & &$\frac{1648900744327349}{57420895248973824} \approx 0.028716040339980885$\\
\hline
$\left(\frac{1}{34},\frac{1}{33}\right]$ & &$\frac{102768396062495}{3483252836794368} \approx 0.02950357062138291$\\
\hline
$\left(\frac{1}{33},\frac{1}{32}\right]$ & &$\frac{307394798055113}{10133099161583616} \approx 0.03033571399562549$\\
\hline
$\left(\frac{1}{32},\frac{1}{31}\right]$ & &$\frac{4347649429067481}{139611588448485376} \approx 0.03114103547838151$\\
\hline
$\left(\frac{1}{31},\frac{1}{30}\right]$ & &$\frac{8690492917356203}{270215977642229760} \approx 0.03216128444063568$\\
\hline
$\left(\frac{1}{30},\frac{1}{29}\right]$ & &$\frac{4363440256637521}{130604389193744384} \approx 0.03340959889307088$\\
\hline
$\left(\frac{1}{29},\frac{1}{28}\right]$ & &$\frac{68900938265833}{1970324836974592} \approx 0.03496932940845911$\\
\hline
$\left(\frac{1}{28},\frac{1}{27}\right]$ & &$\frac{961250425333285}{26599385299156992} \approx 0.03613806915168636$\\
\hline
$\left(\frac{1}{27},\frac{1}{26}\right]$ & &$\frac{478909805672573}{12807111440334848} \approx 0.03739405313240967$\\
\hline
$\left(\frac{1}{26},\frac{1}{25}\right]$ & &$\frac{305262690079423}{7881299347898368} \approx 0.03873253338116443$\\
\hline
$\left(\frac{1}{25},\frac{1}{24}\right]$ & &$\frac{68643357099235}{1688849860263936} \approx 0.04064503228753994$\\
\hline
$\left(\frac{1}{24},\frac{1}{23}\right]$ & &$\frac{823067069583703}{19421773393035264} \approx 0.04237857444464153$\\
\hline
$\left(\frac{1}{23},\frac{1}{22}\right]$ & &$\frac{3276985629950113}{74309393851613184} \approx 0.044099210881653195$\\
\hline
$\left(\frac{1}{22},\frac{1}{21}\right]$ & &$\frac{134417993107429931}{2882303761517117440} \approx 0.04663560964742947$\\
\hline
$\left(\frac{1}{21},\frac{1}{20}\right]$ & &$\frac{105318112328898389}{2161727821137838080} \approx 0.04871941384066732$\\
\hline
$\left(\frac{1}{20},\frac{1}{19}\right]$ & &$\frac{143717784221929}{2814749767106560} \approx 0.0510588137892146$\\
\hline
$\left(\frac{1}{19},\frac{1}{18}\right]$ & &$\frac{5448380985431057}{101330991615836160} \approx 0.05376816015071519$\\
\hline
$\left(\frac{1}{18},\frac{1}{17}\right]$ & &$\frac{2205500045607573}{38280596832649216} \approx 0.057614045445773185$\\
\hline
$\left(\frac{1}{17},\frac{1}{16}\right]$ & &$\frac{273738166045219}{4503599627370496} \approx 0.06078208293241327$\\
\hline
$\left(\frac{1}{16},\frac{1}{15}\right]$ & &$\frac{1454283271924637}{22517998136852480} \approx 0.06458315091271757$\\
\hline
$\left(\frac{1}{15},\frac{1}{14}\right]$ & &$\frac{569257302249594179}{8070450532247928832} \approx 0.07053600043454257$\\
\hline
$\left(\frac{1}{14},\frac{1}{13}\right]$ & &$\frac{2202625446960337}{29273397577908224} \approx 0.07524324571817362$\\
\hline
$\left(\frac{1}{13},\frac{1}{12}\right]$ & &$\frac{136648956798617}{1688849860263936} \approx 0.08091243633537755$\\
\hline
$\left(\frac{1}{12},\frac{1}{11}\right]$ & &$\frac{6554353882386521}{74309393851613184} \approx 0.08820357080929456$\\
\hline
$\left(\frac{1}{11},\frac{1}{10}\right]$ & &$\frac{280027561816137}{2814749767106560} \approx 0.09948577493053427$\\
\hline
$\left(\frac{1}{10},\frac{1}{9}\right]$ & &$\frac{4425169584558011}{40532396646334464} \approx 0.10917611468104983$\\
\hline
$\left(\frac{1}{9},\frac{1}{8}\right]$ & &$\frac{272470868393599}{2251799813685248} \approx 0.12100137265207378$\\
\hline
$\left(\frac{1}{8},\frac{11}{83}\right]$ & &$\frac{133818247543761}{985162418487296} \approx 0.1358336910062375$\\
\hline
$\left(\frac{11}{83},\frac{1}{7}\right]$ & &$\frac{1115619350201202071}{8070450532247928832} \approx 0.13823507693200116$\\
\hline
$\left(\frac{1}{7},\frac{12}{83}\right]$ & &$\frac{173995232712349}{1125899906842624} \approx 0.15453881082580967$\\
\hline
$\left(\frac{12}{83},\frac{3}{20}\right]$ & &$\frac{2816599984716889}{18014398509481984} \approx 0.1563527077095777$\\
\hline
$\left(\frac{3}{20},\frac{1}{6}\right]$ & &$\frac{2211215673518099}{13510798882111488} \approx 0.16366283687678776$\\
\hline
$\left(\frac{1}{6},\frac{15}{88}\right]$ & &$\frac{18001643316976137}{90071992547409920} \approx 0.1998583889159648$\\
\hline
$\left(\frac{15}{88},\frac{1}{5}\right]$ & &$\frac{1}{5} = 0.2$\\
\hline
$\left(\frac{1}{5},\frac{97}{480}\right]$ & &$\frac{275783289409279}{1125899906842624} \approx 0.2449447661672357$\\
\hline
$\left(\frac{97}{480},\frac{1}{4}\right]$ & &$\frac{138176908255481}{562949953421312} \approx 0.2454514960268046$\\
\hline
$\left(\frac{1}{4},\frac{271}{960}\right]$ & &$\frac{1080920410736029}{3377699720527872} \approx 0.32001672740971215$\\
\hline
$\left(\frac{271}{960},\frac{1}{3}\right]$ & &$\frac{24010918137620757}{72057594037927936} \approx 0.33321842698470433$\\
\hline
$\left(\frac{1}{3},\frac{17}{50}\right]$ & &$\frac{505380611144679}{1125899906842624} \approx 0.4488681525535645$\\
\hline
$\left(\frac{17}{50},1-a\right]$ & & $u=\frac{4804339}{8388608} \approx 0.5727218389511108$\\
\hline
$\left(1-a,\frac{653}{1920}\right]$ & &$v=\frac{4945169}{8388608} \approx 0.5895100831985474$\\
\hline
$\left(\frac{653}{1920},\frac{7}{20}\right]$ & &$\frac{78181827}{134217728} \approx 0.5825000032782555$\\
\hline
$\left(\frac{7}{20},\frac{59}{166}\right]$ & &$\frac{19023851}{33554432} \approx 0.5669549405574799$\\
\hline
$\left(\frac{59}{166},\frac{43}{120}\right]$ & &$\frac{150909061}{268435456} \approx 0.5621800608932972$\\
\hline
$\left(\frac{43}{120},\frac{3}{7}\right]$ & &$\frac{150095589}{268435456} \approx 0.5591496415436268$\\
\hline
$\left(\frac{3}{7},\frac{1}{2}\right]$ & &$\frac{1}{2} = 0.5$\\
\hline
$\left(\frac{1}{2},a\right)$ & & $w= \frac{13587699}{16777216} \approx 0.8098899722099304$\\
\hline
$\left[a,1\right]$ & & $1$\\
\hline
\end{longtable}
\end{center}

\subsection{Weight function for the scenario interval $\boldsymbol{(3/7,1/2]}$}

In this section we discuss the threshold class ${(3/7\approx
0.42857,1/2 =0.5]}$.

A set of items from $\Delta$ resulting in a large weight is as
follows:
$$\frac{1}{2}, \  \frac {1}{3}, \ \frac 1{11},  \ \frac 1{37}, \ \frac
1{40},\ \frac 1{43}, \mbox{ \ (and tiny items of total size \
}\frac{997}{2100120}) \ .$$

The total weight of these items is $1+\frac{56035901}{134217728}+
\frac{15032567}{167772160}+\frac{61605937}{2415919104}+\frac{20486803}{872415232}+\frac{30706159}{1409286144}
+\frac{1209038869}{1409286144}\cdot\frac {997}{2100120} \approx
1.578279665$.

Other sets of items give smaller bounds. For example, one could
expect that at least one of $\{\frac 12$, $\frac 37$, $\frac
1{15}\}$, and $\{\frac 12$, $\frac 13$, $\frac {12}{83}\}$ should
be a subset of $\Delta$ resulting in a large weight. However, the
resulting knapsack solutions are $$1+\frac{8388625}{16777216}
+\frac{60903479239}{962072674304}+\frac{1209038869}{1409286144}\cdot
\frac{1}{210}\approx 1.567391$$ and $1+\frac{56035901}{134217728}
+\frac{301731089}{2147483648}+\frac{1209038869}{1409286144}\cdot
\frac{11}{498}\approx 1.576955$, respectively.

One can apply arguments of linear programming to observe that as
in this case $u=v$ and $w=1$, there is a single set giving the
maximum in this scenario, unlike the previous case we have
considered.

%%% 1+8388625/16777216+60903479239/962072674304+1209038869/1409286144/210

%%% 1+56035901/134217728+301731089/2147483648+11*1209038869/1409286144/498

\renewcommand{\arraystretch}{1.2}
\begin{center}
\renewcommand{\arraystretch}{1.2}
\begin{longtable}{|  c   l   c |}
\hline
Interval & & Weight \\
\hline
$\left(0,\frac{1}{43}\right]$ & & $\rho=\frac{1209038869}{1409286144} \approx 0.8579087179331553$\\
\hline
$\left(\frac{1}{43},\frac{1}{42}\right]$ & &$\frac{30706159}{1409286144} \approx 0.02178844880490076$\\
\hline
$\left(\frac{1}{42},\frac{1}{41}\right]$ & &$\frac{122544959}{5502926848} \approx 0.02226905106771283$\\
\hline
$\left(\frac{1}{41},\frac{1}{40}\right]$ & &$\frac{122230883}{5368709120} \approx 0.022767276130616666$\\
\hline
$\left(\frac{1}{40},\frac{1}{39}\right]$ & &$\frac{20486803}{872415232} \approx 0.023482857988430902$\\
\hline
$\left(\frac{1}{39},\frac{1}{38}\right]$ & &$\frac{122602323}{5100273664} \approx 0.024038381286357578$\\
\hline
$\left(\frac{1}{38},\frac{1}{37}\right]$ & &$\frac{123206617}{4966055936} \approx 0.024809752163049335$\\
\hline
$\left(\frac{1}{37},\frac{1}{36}\right]$ & &$\frac{61605937}{2415919104} \approx 0.025499999937083986$\\
\hline
$\left(\frac{1}{36},\frac{1}{35}\right]$ & &$\frac{123079213}{4697620480} \approx 0.02620033132178443$\\
\hline
$\left(\frac{1}{35},\frac{1}{34}\right]$ & &$\frac{61475797}{2281701376} \approx 0.026942963547566357$\\
\hline
$\left(\frac{1}{34},\frac{1}{33}\right]$ & &$\frac{1914511}{69206016} \approx 0.02766393892692797$\\
\hline
$\left(\frac{1}{33},\frac{1}{32}\right]$ & &$\frac{61049939}{2147483648} \approx 0.028428593184798956$\\
\hline
$\left(\frac{1}{32},\frac{1}{31}\right]$ & &$\frac{29141917}{1040187392} \approx 0.02801602598159544$\\
\hline
$\left(\frac{1}{31},\frac{1}{30}\right]$ & &$\frac{117779207}{4026531840} \approx 0.029250782479842505$\\
\hline
$\left(\frac{1}{30},\frac{1}{29}\right]$ & &$\frac{30195047}{973078528} \approx 0.0310304319036418$\\
\hline
$\left(\frac{1}{29},\frac{1}{28}\right]$ & &$\frac{15316403}{469762048} \approx 0.032604598573275974$\\
\hline
$\left(\frac{1}{28},\frac{1}{27}\right]$ & &$\frac{7625221}{226492416} \approx 0.03366656215102584$\\
\hline
$\left(\frac{1}{27},\frac{1}{26}\right]$ & &$\frac{7596363}{218103808} \approx 0.03482911678460928$\\
\hline
$\left(\frac{1}{26},\frac{1}{25}\right]$ & &$\frac{2419405}{67108864} \approx 0.03605194389820099$\\
\hline
$\left(\frac{1}{25},\frac{1}{24}\right]$ & &$\frac{15060049}{402653184} \approx 0.03740203628937403$\\
\hline
$\left(\frac{1}{24},\frac{1}{23}\right]$ & &$\frac{7606395}{192937984} \approx 0.039424041043157165$\\
\hline
$\left(\frac{1}{23},\frac{1}{22}\right]$ & &$\frac{60503345}{1476395008} \approx 0.04098045893690803$\\
\hline
$\left(\frac{1}{22},\frac{1}{21}\right]$ & &$\frac{4399334281}{103079215104} \approx 0.04267915967890682$\\
\hline
$\left(\frac{1}{21},\frac{1}{20}\right]$ & &$\frac{7661091819}{171798691840} \approx 0.04459342348272912$\\
\hline
$\left(\frac{1}{20},\frac{1}{19}\right]$ & &$\frac{1575905}{33554432} \approx 0.04696562886238098$\\
\hline
$\left(\frac{1}{19},\frac{1}{18}\right]$ & &$\frac{120535379}{2415919104} \approx 0.049892142001125545$\\
\hline
$\left(\frac{1}{18},\frac{1}{17}\right]$ & &$\frac{29961703}{570425344} \approx 0.05252519600531634$\\
\hline
$\left(\frac{1}{17},\frac{1}{16}\right]$ & &$\frac{59408913}{1073741824} \approx 0.05532886181026697$\\
\hline
$\left(\frac{1}{16},\frac{1}{15}\right]$ & &$\frac{19934433}{335544320} \approx 0.059409239888191225$\\
\hline
$\left(\frac{1}{15},\frac{1}{14}\right]$ & &$\frac{60903479239}{962072674304} \approx 0.06330444764275203$\\
\hline
$\left(\frac{1}{14},\frac{1}{13}\right]$ & &$\frac{117582149}{1744830464} \approx 0.06738886752953896$\\
\hline
$\left(\frac{1}{13},\frac{1}{12}\right]$ & &$\frac{29485379}{402653184} \approx 0.0732277308901151$\\
\hline
$\left(\frac{1}{12},\frac{1}{11}\right]$ & &$\frac{121031785}{1476395008} \approx 0.08197791535745967$\\
\hline
$\left(\frac{1}{11},\frac{1}{10}\right]$ & &$\frac{15032567}{167772160} \approx 0.08960108160972595$\\
\hline
$\left(\frac{1}{10},\frac{1}{9}\right]$ & &$\frac{118042513}{1207959552} \approx 0.09772058410776986$\\
\hline
$\left(\frac{1}{9},\frac{1}{8}\right]$ & &$\frac{58147653}{536870912} \approx 0.1083084437996149$\\
\hline
$\left(\frac{1}{8},\frac{11}{83}\right]$ & &$\frac{3578045}{29360128} \approx 0.1218674864087786$\\
\hline
$\left(\frac{11}{83},\frac{1}{7}\right]$ & &$\frac{62045542539}{481036337152} \approx 0.12898306790365105$\\
\hline
$\left(\frac{1}{7},\frac{12}{83}\right]$ & &$\frac{9290597}{67108864} \approx 0.13844068348407745$\\
\hline
$\left(\frac{12}{83},\frac{3}{20}\right]$ & &$\frac{301731089}{2147483648} \approx 0.14050448732450604$\\
\hline
$\left(\frac{3}{20},\frac{1}{6}\right]$ & &$\frac{7368719}{50331648} \approx 0.14640329281489053$\\
\hline
$\left(\frac{1}{6},\frac{15}{88}\right]$ & &$\frac{1058367287}{5368709120} \approx 0.19713626932352782$\\
\hline
$\left(\frac{15}{88},\frac{1}{5}\right]$ & &$\frac{1}{5} = 0.2$\\
\hline
$\left(\frac{1}{5},\frac{97}{480}\right]$ & &$\frac{61444293}{268435456} \approx 0.2288978286087513$\\
\hline
$\left(\frac{97}{480},\frac{1}{4}\right]$ & &$\frac{31006051}{134217728} \approx 0.23101308196783066$\\
\hline
$\left(\frac{1}{4},\frac{271}{960}\right]$ & &$\frac{60869221}{201326592} \approx 0.30234069128831226$\\
\hline
$\left(\frac{271}{960},\frac{1}{3}\right]$ & &$\frac{1430507165}{4294967296} \approx 0.3330659039784223$\\
\hline
$\left(\frac{1}{3},\frac{7}{20}\right]$ & &$\frac{56035901}{134217728} \approx 0.41749999672174454$\\
\hline
$\left(\frac{7}{20},\frac{59}{166}\right]$ & &$\frac{14530581}{33554432} \approx 0.43304505944252014$\\
\hline
$\left(\frac{59}{166},\frac{43}{120}\right]$ & &$\frac{117526395}{268435456} \approx 0.4378199391067028$\\
\hline
$\left(\frac{43}{120},\frac{3}{7}\right]$ & &$\frac{118339867}{268435456} \approx 0.4408503584563732$\\
\hline
$\left(\frac{3}{7},\frac{1}{2}\right]$ & & $u=v=\frac{8388625}{16777216} \approx 0.5000010132789612$\\
\hline
$\left(\frac{1}{2},1\right]$ & & $1$\\
\hline
\end{longtable}
\end{center}

\subsection{Weight function for the scenario interval $\boldsymbol{(2/9,3/13]}$}

In this section we discuss the threshold class ${(2/9\approx
0.2222,3/13\approx 0.23076923]}$.

Consider the two multisets in $\Delta$: $\{\frac 14, \frac 14,
\frac 14, \frac 3{20}, \frac 1{11}\}$ and $\{\frac 14, \frac 19,
\frac 19, \frac 19, \frac 19, \frac 1{10}, \frac 1{10}, \frac
1{10}\}$. The resulting knapsack solutions are $$3 \cdot
\frac{39794075}{100663296}+\frac{102325387597129273}{432345564227567616}
+\frac{25427139196783563}{180143985094819840}
+\frac{1400445382739}{945631002624} \cdot \frac{1}{110}\approx
1.5772432$$ and $$\frac{39794075}{100663296}+ 4
\cdot\frac{93992497}{536870912}+ 3 \cdot
\frac{408463500052300289}{2594073385365405696}
+\frac{1400445382739}{945631002624} \cdot \frac{1}{180}\approx
1.576226 \ , $$ respectively. Note that these two multisets giving
the worst case for this scenario have relatively small items.

%3*39794075/100663296+102325387597129273/432345564227567616 +25427139196783563/180143985094819840 +1400445382739/945631002624/110

%39794075/100663296+ 3*408463500052300289/2594073385365405696 +4*93992497/536870912+1400445382739/945631002624/180

\begin{center}
\renewcommand{\arraystretch}{1.2}
\begin{longtable}{|  c   l   c |}
\hline
Interval & & Weight \\
\hline
$\left(0,\frac{1}{43}\right]$ & & $\rho=\frac{1400445382739}{945631002624} \approx 1.4809639054271178$\\
\hline
$\left(\frac{1}{43},\frac{1}{42}\right]$ & &$\frac{457395215}{15502147584} \approx 0.02950528064073422$\\
\hline
$\left(\frac{1}{42},\frac{1}{41}\right]$ & &$\frac{913289039}{30266097664} \approx 0.0301753152698741$\\
\hline
$\left(\frac{1}{41},\frac{1}{40}\right]$ & &$\frac{1824013513}{59055800320} \approx 0.030886272019283338$\\
\hline
$\left(\frac{1}{40},\frac{1}{39}\right]$ & &$\frac{74696187}{2399141888} \approx 0.031134543302175882$\\
\hline
$\left(\frac{1}{39},\frac{1}{38}\right]$ & &$\frac{83370431}{2550136832} \approx 0.03269253239819878$\\
\hline
$\left(\frac{1}{38},\frac{1}{37}\right]$ & &$\frac{1639477277}{49660559360} \approx 0.03301366915976679$\\
\hline
$\left(\frac{1}{37},\frac{1}{36}\right]$ & &$\frac{407082371}{12079595520} \approx 0.033700000163581635$\\
\hline
$\left(\frac{1}{36},\frac{1}{35}\right]$ & &$\frac{748780471}{21139292160} \approx 0.03542126507040054$\\
\hline
$\left(\frac{1}{35},\frac{1}{34}\right]$ & &$\frac{248268817}{6845104128} \approx 0.03626954570120456$\\
\hline
$\left(\frac{1}{34},\frac{1}{33}\right]$ & &$\frac{969073}{25952256} \approx 0.03734060730596985$\\
\hline
$\left(\frac{1}{33},\frac{1}{32}\right]$ & &$\frac{743335051}{19327352832} \approx 0.038460261861069336$\\
\hline
$\left(\frac{1}{32},\frac{1}{31}\right]$ & &$\frac{369923301}{8321499136} \approx 0.04445392530291311$\\
\hline
$\left(\frac{1}{31},\frac{1}{30}\right]$ & &$\frac{717694643}{16106127360} \approx 0.04456034818043311$\\
\hline
$\left(\frac{1}{30},\frac{1}{29}\right]$ & &$\frac{338982541}{7784628224} \approx 0.043545116252940275$\\
\hline
$\left(\frac{1}{29},\frac{1}{28}\right]$ & &$\frac{21159655}{469762048} \approx 0.045043347137314935$\\
\hline
$\left(\frac{1}{28},\frac{1}{27}\right]$ & &$\frac{18496999}{396361728} \approx 0.04666696528278331$\\
\hline
$\left(\frac{1}{27},\frac{1}{26}\right]$ & &$\frac{73772911}{1526726656} \approx 0.048320968727489096$\\
\hline
$\left(\frac{1}{26},\frac{1}{25}\right]$ & &$\frac{11779787}{234881024} \approx 0.05015214426176889$\\
\hline
$\left(\frac{1}{25},\frac{1}{24}\right]$ & &$\frac{21928717}{402653184} \approx 0.054460557798544564$\\
\hline
$\left(\frac{1}{24},\frac{1}{23}\right]$ & &$\frac{63629269}{1157627904} \approx 0.05496521704438804$\\
\hline
$\left(\frac{1}{23},\frac{1}{22}\right]$ & &$\frac{31771343}{553648128} \approx 0.057385442834911926$\\
\hline
$\left(\frac{1}{22},\frac{1}{21}\right]$ & &$\frac{11711682431972778577}{184467440737095516160} \approx 0.063489157681026$\\
\hline
$\left(\frac{1}{21},\frac{1}{20}\right]$ & &$\frac{9171943871524336559}{138350580552821637120} \approx 0.06629494314281197$\\
\hline
$\left(\frac{1}{20},\frac{1}{19}\right]$ & &$\frac{5745881}{83886080} \approx 0.06849623918533325$\\
\hline
$\left(\frac{1}{19},\frac{1}{18}\right]$ & &$\frac{848959177}{12079595520} \approx 0.07028043079707358$\\
\hline
$\left(\frac{1}{18},\frac{1}{17}\right]$ & &$\frac{180923205}{2281701376} \approx 0.07929311298272189$\\
\hline
$\left(\frac{1}{17},\frac{1}{16}\right]$ & &$\frac{90208717}{1073741824} \approx 0.08401341456919909$\\
\hline
$\left(\frac{1}{16},\frac{1}{15}\right]$ & &$\frac{116265557}{1342177280} \approx 0.08662459030747413$\\
\hline
$\left(\frac{1}{15},\frac{1}{14}\right]$ & &$\frac{52175875670169513497}{516508834063867445248} \approx 0.101016424558806$\\
\hline
$\left(\frac{1}{14},\frac{1}{13}\right]$ & &$\frac{94834829}{872415232} \approx 0.10870377490153679$\\
\hline
$\left(\frac{1}{13},\frac{1}{12}\right]$ & &$\frac{45761591}{402653184} \approx 0.11365014066298802$\\
\hline
$\left(\frac{1}{12},\frac{1}{11}\right]$ & &$\frac{63517591}{553648128} \approx 0.1147255590467742$\\
\hline
$\left(\frac{1}{11},\frac{1}{10}\right]$ & &$\frac{25427139196783563}{180143985094819840} \approx 0.14114897693308961$\\
\hline
$\left(\frac{1}{10},\frac{1}{9}\right]$ & &$\frac{408463500052300289}{2594073385365405696} \approx 0.15746027169341759$\\
\hline
$\left(\frac{1}{9},\frac{1}{8}\right]$ & &$\frac{93992497}{536870912} \approx 0.17507466860115528$\\
\hline
$\left(\frac{1}{8},\frac{11}{83}\right]$ & &$\frac{11469903}{58720256} \approx 0.1953312839780535$\\
\hline
$\left(\frac{11}{83},\frac{1}{7}\right]$ & &$\frac{91756671166710256085}{516508834063867445248} \approx 0.17764782539104518$\\
\hline
$\left(\frac{1}{7},\frac{12}{83}\right]$ & &$\frac{7486619}{33554432} \approx 0.2231186330318451$\\
\hline
$\left(\frac{12}{83},\frac{3}{20}\right]$ & &$\frac{258005503388278987}{1152921504606846976} \approx 0.2237841018294306$\\
\hline
$\left(\frac{3}{20},\frac{1}{6}\right]$ & &$\frac{102325387597129273}{432345564227567616} \approx 0.2366750027375549$\\
\hline
$\left(\frac{1}{6},\frac{15}{88}\right]$ & &$\frac{1218245507296398603}{5764607523034234880} \approx 0.21133190810103372$\\
\hline
$\left(\frac{15}{88},\frac{1}{5}\right]$ & &$\frac{1}{5} = 0.2$\\
\hline
$\left(\frac{1}{5},\frac{97}{480}\right]$ & &$\frac{22523859384760717}{72057594037927936} \approx 0.3125813411547595$\\
\hline
$\left(\frac{97}{480},\frac{1}{4}\right]$ & &$\frac{41199575}{134217728} \approx 0.306960754096508$\\
\hline
$\left(\frac{1}{4},\frac{271}{960}\right]$ & &$\frac{39794075}{100663296} \approx 0.39531861742337543$\\
\hline
$\left(\frac{271}{960},\frac{1}{3}\right]$ & &$\frac{716976483}{2147483648} \approx 0.33386819204315543$\\
\hline
$\left(\frac{1}{3},\frac{7}{20}\right]$ & &$\frac{24295923874979269}{144115188075855872} \approx 0.16858683806588773$\\
\hline
$\left(\frac{7}{20},\frac{59}{166}\right]$ & &$\frac{4958652829951325}{36028797018963968} \approx 0.1376302635733663$\\
\hline
$\left(\frac{59}{166},\frac{43}{120}\right]$ & &$\frac{36928515705467059}{288230376151711744} \approx 0.12812152625450385$\\
\hline
$\left(\frac{43}{120},\frac{3}{7}\right]$ & &$\frac{150095589}{268435456} \approx 0.5591496415436268$\\
\hline
$\left(\frac{3}{7},\frac{1}{2}\right]$ & &$\frac{1}{2} = 0.5$\\
\hline
$\left(\frac{1}{2},a\right)$ & &$w=\frac{9224745}{1073741824} \approx 0.008591213263571262$\\
\hline
$\left[a,1\right]$ & & $1$\\
\hline
\end{longtable}
\end{center}

\section{A short discussion of the results of \cite{HvSbparxiv}}
\label{rvssh} We had some difficulties in verifying the result of
\cite{HvSbparxiv}. Here, we only address some issues in this last
manuscript. Both in the conference proceedings version
\cite{HvS16} and in the arxiv version published in September 2016
\cite{HvSbparxiv}, the authors use a linear program and its dual
(see page 27 of \cite{HvSbparxiv}). The dual variables $y_1$,
$y_2$ should be non negative (as the corresponding primal
constraints are inequalities). However, instead of proceeding to
solving the dual (or finding a feasible solution for it) the
authors just fix those variables to values that can be negative in
many of the cases. Going back to their primal linear program,
these negative values mean that the authors assume (without
providing a proof for it) that a pair of patterns (specified in
advance) are critical in all scenarios. There is no clear reason
as for why this should hold for their setting and we note that in
our setting (which is related to their setting, as we also try to
pack large items together with huge items as much as possible, and
their critical patterns are those where an optimal solution packs
such a pair of items together, leaving just a little gap for tiny
items according to their definition of tiny). The corresponding
patterns of our algorithm and its analysis are in fact not
critical in many of the scenarios.

We do not see a simple way to fix this flaw. Setting these dual
variables to zeroes instead of the resulting negative values would
be incorrect either as the resulting solution possibly becomes
infeasible for some scenarios. We believe that due to this, not
all required calculations were done, a solution for the dual
linear program was not found for all scenarios, and it is possible
that the true competitive ratio is higher (maybe even higher than
1.583333, the barrier they claim to break). Moreover, there should
be additional linear programs for other cases which are not
presented in the manuscripts (and do not appear in the accessible
additional data that the first author provides on her web page),
and we cannot tell if any problems of this kind occurs there as
well.

%\newpage

\bibliographystyle{abbrv}

\end{document}